\documentclass[aps,prx,twocolumn,superscriptaddress,notitlepage,noshowkeys,eqsecnum]{revtex4-2}
\usepackage[T1]{fontenc}
\usepackage[latin9]{inputenc}
\usepackage[english]{babel}
\usepackage[table]{xcolor}
\definecolor{note_fontcolor}{rgb}{0.800781, 0.800781, 0.800781}
\usepackage{graphicx,color}
\usepackage{amsmath}
\usepackage{amssymb}
\usepackage{amsthm}
\usepackage{revsymb}
\usepackage{bm}
\usepackage{hyperref}
\usepackage{xurl}
\hypersetup{breaklinks=true}
\usepackage{enumitem} 
\usepackage{tikz}
\usepackage{pgfplots}
\usepackage{hhline}
\usepackage{enumitem}
\usepackage{fancyvrb}
\usepackage[most]{tcolorbox}

\newtheorem{thm}{Theorem}

\newtheorem{lemma}[thm]{Lemma}
\newtheorem{corol}[thm]{Corollary}

\theoremstyle{definition}

\newtheorem{remark}[thm]{Remark}

\newcommand{\rmd}{\mathrm{d}}
\newcommand{\rmi}{\mathrm{i}}
\newcommand{\rme}{\mathrm{e}}

\newcommand{\sinc}{\operatorname{sinc}}

\newcommand{\T}{\operatorname{T}}

\newcommand{\Ci}{\operatorname{Ci}}

\definecolor{dgreen}{rgb}{0,0.5,0}

\definecolor{dblue}{rgb}{0,0,0.6}

\definecolor{dred}{rgb}{0.784,0,0}

\definecolor{dorange}{cmyk}{0,0.72,1,0.16}
\definecolor{dmagenta}{rgb}{0.847,0.149,0.490}

\definecolor{delete}{cmyk}{0.5,0,0,0}

\definecolor{0}{RGB}{51, 51, 51}
\definecolor{1}{RGB}{49, 130, 189}
\definecolor{2}{RGB}{255, 102, 0}
\definecolor{3}{RGB}{51, 160, 44}
\definecolor{4}{RGB}{227, 26, 28}
\definecolor{5}{RGB}{106, 61, 154}
\definecolor{6}{RGB}{140, 81, 10}

\begin{document}
\title{Strong Error Bounds for Trotter \& Strang-Splittings and Their Implications for Quantum Chemistry}
\author{Daniel Burgarth}
\affiliation{Department Physik, Friedrich-Alexander-Universit\"at Erlangen-N\"urnberg, Staudtstra\ss e 7, 91058 Erlangen, Germany}
\affiliation{Center for Engineered Quantum Systems, Macquarie University, 2109 NSW, Australia} \author{Paolo Facchi}
\affiliation{Dipartimento di Fisica, Universit\`a di Bari, I-70126 Bari, Italy}
\affiliation{INFN, Sezione di Bari, I-70126 Bari, Italy}
\author{Alexander Hahn}
\affiliation{Center for Engineered Quantum Systems, Macquarie University, 2109 NSW, Australia}
\author{Mattias Johnsson}
\affiliation{Center for Engineered Quantum Systems, Macquarie University, 2109 NSW, Australia}
\author{Kazuya Yuasa}
\affiliation{Department of Physics, Waseda University, Tokyo 169-8555, Japan}
\date{\today}
\begin{abstract}
	Efficient error estimates for the Trotter product formula are central in quantum computing, mathematical physics, and numerical simulations. However, the Trotter error's dependency on the input state and its application to unbounded operators remains unclear. Here, we present a general theory for error estimation, including higher-order product formulas, with explicit input state dependency. Our approach overcomes two limitations of the existing operator-norm estimates in the literature. First, previous bounds are too pessimistic as they quantify the worst-case scenario. Second, previous bounds become trivial for unbounded operators and cannot be applied to a wide class of Trotter scenarios, including atomic and molecular Hamiltonians.  Our method enables analytical treatment of Trotter errors in chemistry simulations, illustrated through a case study on the hydrogen atom. Our findings reveal: (i) for states with fat-tailed energy distribution, such as low-angular-momentum states of the hydrogen atom, the Trotter error scales worse than expected (sublinearly) in the number of Trotter steps; (ii) certain states do not admit an advantage in the scaling from higher-order Trotterization, and thus, the higher-order Trotter hierarchy breaks down for these states, including the hydrogen atom's ground state; (iii) the scaling of higher-order Trotter bounds might depend on the order of the Hamiltonians in the Trotter product for states with fat-tailed energy distribution. Physically, the enlarged Trotter error is caused by the atom's ionization due to the Trotter dynamics. Mathematically, we find that certain domain conditions are not satisfied by some states so higher moments of the potential and kinetic energies diverge. Our analytical error analysis agrees with numerical simulations, indicating that we can estimate the state-dependent Trotter error scaling genuinely.
\end{abstract}
\maketitle

\section{Introduction}
Solving the Schr{\"o}dinger equation is the most fundamental endeavor in quantum physics.
However, depending on the model, this task is extremely difficult or even impossible to perform analytically.
Therefore, methods to find approximate solutions are of paramount importance.
Besides perturbation theory, the Trotter product formula serves as a powerful tool to approximate the time-evolution of a quantum system.
By subdividing the total time into small steps and by switching between simpler components, it allows the implementation of the dynamics under the full Hamiltonian.
The advantage of this approach is that the approximation becomes more accurate as the smaller time steps are chosen.
Thus, one can tune the accuracy of the Trotter approximation according to the problem at hand.
Due to this, the Trotter product formula has become the basic building block in many areas of physics.
To name three of them:
(i) \emph{Quantum computing} --- One of the most promising and powerful potential applications of quantum computers is simulating the dynamics of quantum systems.
Already in 1982, Feynman conjectured an exponential quantum advantage from this approach~\cite{Feynman1982}\@.
Lloyd then delivered an explicit algorithm in 1996 to implement Feynman's idea~\cite{Lloyd1996} on the basis of the Trotter product formula.
Ever since, it has been an active area of research and the list of papers referring to the exponential quantum advantage is rapidly growing~\cite{QSimPapers2023}\@.
A particularly important subclass of quantum simulation is that of quantum chemistry --- the simulation of chemical materials~\cite{Lanyon2010, Wecker2014, Babbush2015, Poulin2015, Babbush2018, Cao2019, McArdle2020}\@.
It has been proposed as a potential pathway to developing new pharmaceuticals~\cite{Lam2020}, catalysts~\cite{Ahn2019}, or fertilizers~\cite{Berry2019, Lee2021}\@.
All these simulations can be understood in terms of the Trotter product formula.
(ii) \emph{Mathematical physics} --- Here, the study of the convergence of the Trotter product formula has a long history.
See, for instance, Refs.~\cite{Trotter1959, Kato1978} for some early works.
Apart from elementary interest in the Trotter product formula due to its fundamental importance for quantum physics, it also has practical implications for mathematical physics.
For instance, it forms the basis for the path-integral formulation of quantum mechanics~\cite{Nelson1964, Johnson2002, Folland2008}\@.
As such, the Trotter product formula can be the foundation for quantum field theory~\cite{Klauber2013, Simon2005, Nicola2019, Gaveau2004}\@.
(iii) \emph{Numerical simulations} --- The standard algorithms to solve the Schr{\"o}dinger equation numerically are based on the Trotter product formula.
Most prominently, the \emph{split-step} method~\cite{McLachlan2002} (also known as \emph{strang-splitting}) is used whenever the potential energy is reasonably well understood.
This is the case for instance in the simulation of atoms and molecules~\cite{Kosloff1988, Tannor2008, Choi2019, Roulet2019, Chan2023}, quantum optics~\cite{Fitzek2020}, or nonlinear physics such as Bose-Einstein condensation~\cite{Besse2002,Lubich2008,Bao2013,Eilinghoff2016,Henning2017,Ostermann2021,Bao2022,Bao2024}\@.
Furthermore, numerical algorithms for studying quantum many-body systems are based on Trotterization~\cite{Verstraete2004, Zwolak2004, Vidal2004, Cincio2008}\@.

In all practical simulations of quantum systems, it is important to understand the Trotter error.
That is, how close is the Trotter evolution to the actual target evolution?
On one hand, in quantum simulation, one uses the Quantum Phase Estimation algorithm to extract the spectral information of a Hamiltonian of interest.
The controlled unitary gates in the quantum circuit are implemented by Trotterizing between noncommuting elementary gates corresponding to a representation of the target Hamiltonian.
The Trotter error then determines the necessary resources (gate count) of the algorithm.
Due to the lack of large fault-tolerant quantum computers existing to date, providing such resource estimates is a difficult problem.
For this reason, the current approach is to use heuristics, so that resource estimates can be made based on a set of assumptions on the Trotter error~\cite{Reiher2017, Berry2019, Lee2021}\@.
On the other hand, for numerical simulations, one directly Trotterizes between noncommuting parts of the target Hamiltonian or a truncated (discretized) version of it.
Here, the Trotter error determines the convergence of the algorithm to a reasonable solution.
Therefore, efficient implementations of the Trotter product formula and a realistic estimate for the Trotter error are central in the development of both quantum technology and numerical physics.
The current state-of-the-art approaches are bounds, which quantify the Trotter error in terms of the operator norm of the involved commutators~\cite{Childs2021, Schubert2023, Zhuk2023}\@.
However, for particular input states, these results are far from tight.
This is because the operator norm quantifies the worst-case scenario corresponding to the input state with the largest error.
For this reason, a state-dependent analysis leads to a more realistic and tighter estimate of the actual Trotter error.
This has also been noticed e.g.\ in Ref.~\cite{Childs2021}, where the authors conclude: ``The spectral-norm error bound [\dots] would be overly pessimistic if we simulate with a low-energy initial state. It would then be beneficial to change to the error metric of the Euclidean distance to avoid the worst-case propagation.''
A simple example where this can be observed is given in Sec.~\ref{sec:first_order}:
For two Hamiltonians $A$ and $B$, we can have $[A, B]$ large, but $A\varphi\approx0$ and $B\varphi\approx0$ for the input state $\varphi$ of interest.
Even more convincing is the example of atomic and molecular Hamiltonians.
Here, the commutator norm bounds diverge and a state-dependent analysis is actually \emph{mandatory}.
See for instance Refs.~\cite{Simon2005, Ito1998, Ichinose2004, Burgarth2023} or the discussion in Sec.~\ref{sec:previous_results}\@.

These considerations lead to the immediate desire to be able to quantify state-dependent Trotter convergence speeds in a general scenario.
In this paper, we solve this issue.
We derive state-dependent Trotter bounds that apply to arbitrary Hamiltonians with eigenstates.
These bounds might have wide-ranging applications in quantum technology, mathematical physics, and numerical simulations.
They refine existing bounds in the literature, where it is usually assumed that the Trotter error scales as $N^{-1}$, i.e.\ with the inverse number of Trotter steps $N$.
This is justified by comparing the Taylor expansions of Trotterized and target dynamics, which coincide up to the first order.
Due to the importance of efficient Trotterization, a lot of effort has been put into improving this method and achieving agreement of Trotterized and target dynamics in higher Taylor orders as well.
This led to what is called \emph{higher-order Trotter-Suzuki decompositions}~\cite{Childs2021}\@.
In practice, the second-order Trotterization (also known as \emph{strang-splitting} or \emph{split-step method}) is the most commonly used higher-order Trotter scheme.
This is because it improves the asymptotic scaling to $N^{-2}$, while still being similarly cost-efficient as the first-order Trotterization.
As for the standard first-order Trotter product formula, state-dependent bounds for higher-order schemes are not available in the literature yet.
We also derive general state-dependent Trotter bounds for higher-order Trotter schemes, which are valid for arbitrary Hamiltonians with eigenstates.
This combines both advantages of faster error scaling due to higher-order Trotterization and smaller error prefactor due to an explicit dependency on the input state.

Of course, upper bounds are not a mathematical proof of the asymptotic convergence speed.
Therefore, we bolster our general theory with a detailed case study of chemistry simulation.
This numerical analysis convincingly indicates the tightness of the asymptotic scaling of our bounds.
As an application of our results, we study the hydrogen atom, which is \emph{the} paradigmatic example to look at.
Historically, it has been the first quantum system, for which the Schr{\"o}dinger equation was solved analytically~\cite{Schroedinger26, Pauli1926}\@.
This makes the hydrogen atom interesting from a fundamental perspective.
But also for practical purposes, the hydrogen atom is of outstanding importance.
It is simple enough to gain analytical insights while still reflecting many features of more complicated systems.
Therefore, its wave functions form the basis for our understanding of multi-electron atoms~\cite{Helgaker2014}\@.
For example, many molecular Hamiltonians are modeled using linear combinations of atomic orbitals (LCAO), which are derived from the hydrogen eigenfunctions~\cite{Huheey1973}\@.
By studying certain eigenstates of the Hamiltonian of the hydrogen atom, we find that the asymptotic Trotter error scaling with $N$ differs from the predictions in the literature.
More precisely, the scaling with $N$ is determined by the quantum number $\ell$ of the orbital angular momentum.
In particular, we find that all eigenstates with $\ell\geq2$ (d-orbitals and higher) admit the desired $N^{-1}$ scaling.
However, for s-orbitals ($\ell=0$), the Trotter error only scales as $N^{-1/4}$.
This is remarkable as it implies that the Trotter error for the ground state is larger than expected.
Notice that the ground state plays a particularly important role in chemistry simulations as it encodes many interesting properties of an atom or a molecule.

We then ask the question whether higher-order Trotter methods provide relief and admit faster scalings.
This leads to two more revelations.
First, we find an asymmetry in the scaling of the Trotter bound.
This shows that the ordering of the Trotter product could be crucial.
Second, for s-orbitals including the ground state, we do not find any improvement from higher-order schemes.
In fact, the Trotter error scaling remains $N^{-1/4}$.
This shows that our observation on the slower ground-state convergence carries over to $p$th-order product formulas and puts higher-order schemes into question for certain ground-state chemistry simulations.

\subsection{Relation to previous results}\label{sec:previous_results}
Historically, the study of product formulas started in the field of pure mathematics.
However, applications of product formulas are now ubiquitous in many fields, such as applied mathematics (see strang-splitting or symmetric integrators), physics (see Trotterization), and computer science (see split-step method).
In recent years, all these applied fields have contributed to advancing our understanding of product formulas.
Nevertheless, there has been a lack of transfer of the results among these research areas.
A prominent example where this can be observed is that of chemistry simulations, which have only been studied from the perspective of physics and computer science.
In this paper, we employ methods from different fields, shed light on their connections, and discuss the implications for chemistry simulations from different angles.
We hope that this will help to start a discourse among the different areas employing product formulas.

The main contribution of physics and computer science to product formulas has been the development of more and more efficient Trotter error bounds to improve digital simulations.
This led to the seminal and influential work of Childs \textit{et~al.}~\cite{Childs2021}.
These results have been further improved very recently by Schubert and Mendl~\cite{Schubert2023} and Zhuk \textit{et~al.}~\cite{Zhuk2023}\@.
Furthermore, the physical origin of Trotter errors in hopping models is discussed in Refs.~\cite{Abanin2017,Khvalyuk2023}\@.
The error bounds provided in these references quantify the Trotter error in some matrix norm, usually the operator norm, and give rise to a commutator scaling.
Such bounds are particularly well-suited for simulations of many-body systems, local observables, or quantum Monte-Carlo~\cite{Childs2021}\@.
If we wanted to apply these norm bounds in the context of chemistry simulations, we would first truncate (discretize) the atomic or molecular Hamiltonian of interest at some finite level.
Afterwards, we would either simulate the truncated Hamiltonian directly by Trotterization or go to a second-quantized picture.
In any case, the commutator of the noncommuting components of the resulting operator would determine the Trotter error.
The hope is that for a large enough truncation dimension, the actual target dynamics gets approximated well.
With the same logic, the so-computed Trotter error would be close to the true Trotter error of the full system.
However, this logic can sometimes be wrong and additional care has to be taken.
In fact, the norm Trotter errors simply diverge for unbounded operators such as the ones considered in chemistry simulations.
See Eq.~\eqref{eq:commutator_diverge} for an example.
Intuitively, this happens because norm error bounds increase with extending truncation dimension, which ultimately leads to a diverging error.
This shows that there is a very intricate interplay between the error of truncation and the error of Trotterization~\cite{Burgarth2023}\@.
The circumstance that norm error bounds are not the right quantities to look at for product formulas is well known in pure~\cite{Trotter1959, Kato1978, Simon2005} and applied mathematics~\cite{Jahnke2000, Thalhammer2008, Hansen2009, Neuhauser2009, Kieri2015} as well as mathematical physics~\cite{Ichinose2004, Simon2005}\@.
For this reason, it is necessary to study \emph{state-dependent} (also known as \emph{strong} in the mathematical literature) error bounds for chemistry simulations.
So far, Trotter methods considering low-energy subspaces have been developed~\cite{Sahinoglu2021, Yi2022, Gong2023} and some effort on infinite-dimensional and state-dependent Trotter bounds has been made~\cite{LeVeque1983,Rogava1993,Sheng1994,Tang1995,Kuroda1995,Dia1997,Ichinose1997,Ichinose1997a,Takanobu1997,Ichinose1998,Doumeki1998,Neidhardt1999,Jahnke2000,Ichinose2004a,Thalhammer2008,Hansen2009,Neuhauser2009,Kieri2015,An2021,Hatomura2022,Burgarth2023,Hatomura2023}\@.
However, these results are not general enough to apply to atomic and molecular Hamiltonians in quantum chemistry.
In particular, the bounds in Refs.~\cite{LeVeque1983, Sahinoglu2021, Yi2022, Hatomura2022, Burgarth2023, Hatomura2023, Gong2023} still diverge as they are considering finite-dimensional systems.
The bounds derived in Refs.~\cite{Rogava1993,Sheng1994,Dia1997,Ichinose1997,Ichinose1997a,Takanobu1997,Ichinose1998,Doumeki1998,Neidhardt1999,Ichinose2004a} do not consider the unitary Trotter product formula, but the semigroup case instead.
This is an unphysical scenario, in which the norm of the input state is not conserved, but decreases with time.
Due to this decay behavior, it is much easier to derive semigroup Trotter bounds than in the unitary setting.
In fact, the aforementioned references could use the decay behavior in semigroups to derive \emph{norm} error bounds in situations where only state-dependent bounds become non-trivial for unitary Trotterization.
On the contrary, Refs.~\cite{Tang1995,Jahnke2000,Thalhammer2008,Hansen2009,Neuhauser2009,Kieri2015} are concerned with state-dependent Trotter bounds for the semigroup setting.
Due to the semigroup nature of these results, they still do not apply to a physical unitary time-evolution.
In addition, the assumptions made in these papers are not satisfied by chemistry Hamiltonians.
In particular, Ref.~\cite{Thalhammer2008} and Ref.~\cite[Thms.~2.1 and~2.2]{Jahnke2000} still require at least one of the two operators to be bounded.
Reference~\cite{Tang1995} does not consider the setting of a quantum mechanical Hilbert space.
Furthermore, the assumptions on the (multi-) commutators in Refs.~\cite{Hansen2009, Neuhauser2009,Kieri2015} and Ref.~\cite[Thm.~2.3]{Jahnke2000} are not satisfied for chemistry simulations.
We should mention that Ref.~\cite{Neuhauser2009} shows how to extend their semigroup results to the unitary case by Friedrich's extension.
Additionally, Refs.~\cite{Kuroda1995,Hochbruck2003,An2021} and~\cite[Sec.~4.5]{Bao2013} consider state-dependent Trotter bounds for the unitary setting.
Nevertheless, all these references make additional assumptions about the commutator (or potential in Ref.~\cite{Bao2013}), which are not satisfied by chemistry Hamiltonians.
This shows that there is a pressing gap in the literature, which concerns one of the most important classes of models for quantum computing.
On the contrary, our bounds are directly applicable to chemistry simulations.
They apply to general $p$th-order product formulas as well and provide conditions under which the Trotter error scales as $N^{-p}$.
For the case where these conditions are violated, we derive error bounds with slower fractional scalings.
Notice the effect of potentially fractional Trotter error scaling is known in the mathematical physics literature discussed above~\cite{Rogava1993, Kuroda1995, Ichinose1997, Takanobu1997, Doumeki1998, Neidhardt1999, Jahnke2000, Ichinose2004a}\@.

Due to the generality, our results can explain some open questions in the literature.
For example, the authors of Ref.~\cite{Chan2023} numerically found larger absolute ground-state Trotter errors that imply increased resource estimates for quantum chemistry.
So far, an explanation of this circumstance is lacking.
Furthermore, the problem of slower Trotter convergence for the ground state of the hydrogen atom has already been observed numerically in 2006~\cite{Chin2006}\@.
Nevertheless, the explanation for this effect remained open.
Both effects are aligned with our investigation and can be explained by our results.
Notice that similar issues are known for the Numerov method in numerical simulation~\cite{Buendia1985}, which might be due to related reasons.
In addition, a recent analysis by Lee \textit{et~al.}~\cite{Lee2023} suggests that a proof for an exponential quantum advantage is yet to be found.
In their conclusion, the authors state~\cite{Lee2023}\@:
``Numerical calculations are neither mathematical proof of asymptotics with respect to size and error, nor can we exclude exponential quantum advantage in specific problems.
However, our results suggest that without new and fundamental insights, there may be a lack of generic exponential quantum advantage in this task.''
Our paper aims to kick off such a discussion and provides ideas on how to tackle problems in quantum chemistry analytically.
Notice that our results only imply a polynomially increased runtime and do not affect a potential exponential quantum advantage of quantum simulations.

\subsection{Structure of the paper}
The remainder of the paper is structured as follows.
In Sec.~\ref{sec:first_order}, we present state-dependent error bounds for the standard (first-order) Trotter product formula (Thm.~\ref{thm:trotter_thm}).
These bounds are valid even for unbounded operators such as atomic and molecular Hamiltonians under certain constraints on the input state.
Furthermore, we show how to obtain state-dependent error bounds when the input state does \emph{not} satisfy the conditions (Thm.~\ref{thm:first_order_alpha}).
In this case, the scaling of the Trotter error gets slowed down.
We apply these results to the hydrogen atom in Sec.~\ref{sec:bound} and find slower Trotter scalings for the ground state and other states with low angular momenta.
Afterwards, these findings are verified numerically.
Section~\ref{sec:second_order} is dedicated to the second-order Trotterization.
We develop state-dependent bounds here, which are also applicable to unbounded Hamiltonians (Thm.~\ref{thm:2nd_order}).
Again, these bounds rely on assumptions on the input state of interest, and error bounds with slower scalings are obtained when the assumptions are not satisfied (Thm.~\ref{thm:second_order_alpha}).
These general results are complemented by the application to the hydrogen atom in Sec.~\ref{sec:second_order_hydrogen}\@.
We find that there is no improvement from the second-order strategy for the ground state.
However, excited states can admit a provably advantageous scaling given the ordering of the operators in the Trotter product is correct.
In Sec.~\ref{sec:higher_order}, we present a method to obtain state-dependent bounds for higher-order Trotter-Suzuki product formulas.
We explicitly compute the bounds for a fourth-order Trotterization (Thm.~\ref{thm:4th_order}).
Furthermore, we provide a loose estimate for arbitrary-order Trotterization, that shows the desired asymptotic scaling if the input state satisfies certain conditions (Thm.~\ref{thm:loose_bound_higher_order}).
These conditions are not satisfied for the ground state of the hydrogen atom, indicating the breakdown of the higher-order Trotter hierarchy.
Section~\ref{sec:methods} then gives more details on the methods and proof ideas used to obtain the results.
Possible generalizations and an outlook on how our methodology might apply to other areas of research are given in Sec.~\ref{sec:outlook}\@.
Finally, we conclude in Sec.~\ref{sec:conclusion}\@.
The technical details of the derivations are collected in the appendices.
Appendix~\ref{sec:proof_trotter_thm} provides the proof of the first-order Trotter bounds presented in Sec.~\ref{sec:first_order} (Thm.~\ref{thm:trotter_thm} and Thm.~\ref{thm:first_order_alpha}).
How these results explicitly carry over to the hydrogen atom is discussed in Appendix~\ref{sec:hydrogen_bounds}\@.
The method to obtain the higher-order Trotter bounds is then outlined in detail in Appendix~\ref{sec:higher_order_derivation}\@.
This enables us to prove the second-order Trotter bounds presented in Sec.~\ref{sec:second_order} (Thm.~\ref{thm:2nd_order} and Thm.~\ref{thm:second_order_alpha}), and the loose bound on the general $p$th-order Trotter error (Thm.~\ref{thm:loose_bound_higher_order}).
Afterward, we provide a Mathematica script to explicitly calculate higher-order bounds with small prefactors in Appendix~\ref{sec:Mathematica}\@.
This script is used to compute the fourth-order Trotter bounds (Thm.~\ref{thm:4th_order}) as well as the bounds for a particular sixth-order product formula.

\section{First-order Trotter bounds}
\label{sec:first_order}
For two Hamiltonians $H_1$ and $H_2$, the commonly used bound in the literature to estimate the Trotter error is~\cite{Childs2021, Suzuki1985}
\begin{align}
&\left\|
\left(\rme^{-\rmi \frac{t}{N}H_2}\rme^{-\rmi \frac{t}{N}H_1}\right)^{N}
-\rme^{-\rmi t(H_{1}+H_{2})}
\right\|\nonumber \\
&\qquad\qquad\qquad\qquad\qquad
\le\frac{t^{2}}{2N}\Vert [H_{1},H_{2}]\Vert,
\label{eq:operator_norm_bound}
\end{align}
where $\Vert A\Vert=\sup_{\Vert\psi\Vert=1}\Vert A\psi\Vert$ is the operator norm of an operator $A$ with $\Vert\psi\Vert=\sqrt{\langle\psi\vert\psi\rangle}$ the (Euclidean) norm of a vector $\psi$.
If $A$ is a matrix, $\Vert A\Vert$ computes its largest singular value.
Thus, the standard error measure in Eq.~\eqref{eq:operator_norm_bound} considers a worst-case scenario of the input state with the largest Trotter error.
Clearly, by incorporating the explicit dependency on the input state, a tighter and more accurate error analysis is possible. 
In particular, consider the case mentioned in the introduction, where $[H_1, H_2]$ is large but $H_1\varphi\approx 0$ and $H_2\varphi\approx 0$ for an input state $\varphi$ of interest.
A simple example of this scenario is given by the two block matrices
\begin{equation}
H_1=\begin{pmatrix}
\medskip
0 & \epsilon\\
\epsilon & A
\end{pmatrix},
\qquad
H_2=\begin{pmatrix}
\medskip
0 & \epsilon\\
\epsilon & B
\end{pmatrix},
\label{eq:example}
\end{equation}
with
$\Vert[A, B]\Vert\gg 1$ and $\Vert\epsilon\Vert\ll 1$.
Here, $\Vert[H_1,H_2]\Vert\gg 1$ due to the noncommuting blocks $A$ and $B$.
Thus, the norm error~\eqref{eq:operator_norm_bound} is large in this case.
However, consider a state
\begin{equation}
\varphi=\begin{pmatrix}
\medskip
\varphi_0 \\
0
\end{pmatrix},
\label{eq:state}
\end{equation}
which is only supported in the zero blocks of $H_1$ and $H_2$.
For such a state, one has $\Vert H_1\varphi\Vert=\Vert H_2\varphi\Vert=\Vert\epsilon\varphi_0\Vert\ll 1$, which would lead to a small state-dependent Trotter error.
This example shows that norm bounds can indeed be very loose in situations, where state-dependent bounds are actually small.
Intuitively, one would expect this to happen for low-energy input states~\cite{Sahinoglu2021, Yi2022}.
Therefore, it is natural to consider the Trotter error evaluated on eigenstates of the target Hamiltonian.

Furthermore, we will see in Sec.~\ref{sec:bound} that state-dependent bounds are necessary to describe problems in quantum chemistry.
This is because in general we have $\Vert[H_1, H_2]\Vert=\infty$ since the involved operators are unbounded.
Our state-dependent bounds for the Trotter product formula are valid for both bounded and unbounded Hamiltonians.
For the case of unbounded Hamiltonians, we need to impose a condition on the input state.
This is because an unbounded operator $A$ is only defined on a certain (dense) subspace $\mathcal{D}(A)\subset\mathcal{H}$ of the Hilbert space $\mathcal{H}$.
The space $\mathcal{D}(A)$ is called the \emph{domain} of $A$ and incorporates many important physical properties of the operator, such as boundary conditions.

As a simple example, consider the standard one-dimensional position operator 
$X:\psi(x)\mapsto x \psi(x)$ on wave functions $\psi\in \mathcal{H}=L^2(\mathbb{R})$.
The operator $X$ has to map a wave function into a wave function. 
That is, we have to require that  $X \psi$ is again square-integrable, i.e.\ $x\psi(x)\in L^2(\mathbb{R})$.
This imposes a condition on the possible input states of $X$, which restricts its domain $\mathcal{D}(X)=\{\psi \in L^2(\mathbb{R}) : x\psi(x)\in L^2(\mathbb{R})\}$. 
Obviously, $\mathcal{D}(X)$ is a strict subset of the entire $L^2(\mathbb{R})$. 
For example, the square-integrable function $\psi(x)= 1/(1+|x|)$ does not belong to $\mathcal{D}(X)$, because $x\psi(x)$ is not a wave function ($\|X\psi\|=\infty$).

Notice that all Hamiltonians on \emph{finite}-dimensional Hilbert spaces are automatically bounded.
Bounded operators are defined on the entire Hilbert space so that we do not have to worry about their domains.
Therefore, the domain conditions in our results are trivially satisfied and can be dropped for all \emph{bounded} Trotter pairs, in particular for finite-dimensional systems.
Conversely, for unbounded self-adjoint operators, the domain is always a strict subspace of the Hilbert space and the domain conditions are mandatory.

We now present our state-dependent bound for the first-order Trotter product formula.
For simplicity, we initially consider eigenstates and lift these bounds to general states later.
To this end, we define the eigenstate-dependent error
\begin{equation}
	\xi_N(t;\varphi)\equiv
	\left\Vert \left[\left(\rme^{-\rmi \frac{t}{N}H_2}\rme^{-\rmi \frac{t}{N}H_1}\right)^{N}-\rme^{-\rmi th}\right]\varphi\right\Vert,\label{eq:definition_standard_error}
\end{equation}
where $h$ is an eigenstate of $H_1+H_2$, i.e.\ $(H_1+H_2)\varphi=h\varphi$.
Theorem~\ref{thm:trotter_thm} gives a bound on $\xi_N(t;\varphi)$ and is a generalization of Ref.~\cite[Main result 1]{Burgarth2023} to unbounded Hamiltonians.
\begin{thm}\label{thm:trotter_thm}
Let $H_{1}$ be self-adjoint on $\mathcal{D}(H_{1})$ and $H_{2}$
be self-adjoint on $\mathcal{D}(H_{2})$. 
Let $\varphi$ be an eigenstate of $H_{1}+H_{2}$ with eigenvalue $h$, i.e.\ $(H_{1}+H_{2})\varphi=h\varphi$.
If $\varphi\in\mathcal{D}(H_{1}^{2})\cap\mathcal{D}(H_{2}^{2})$,
then 
\begin{equation*}
\xi_N(t;\varphi)\le\frac{t^{2}}{2N}\,\Bigl(\Vert(H_{1}-g)^{2}\varphi\Vert+\Vert(H_{2}-h+g)^{2}\varphi\Vert\Bigr),
\end{equation*}
for all $t,g\in\mathbb{R}$, and the Trotter product formula converges on $\varphi$.
\end{thm}
\begin{proof}
This is proved in Appendix~\ref{sec:proof_trotter_thm}\@.
The idea is to shift the energy of $H_1+H_2$ in $\varphi$ to zero.
By this approach, the target evolution applied to $\varphi$ becomes the identity. 
The Trotter error is bounded by $N$ times the error in a single step, which can be bounded under the above domain conditions on $\varphi$ by using a second-order Taylor expansion for the short-time evolution
\begin{equation}
\rme^{-\rmi \frac{t}{N} H_j}\varphi= \varphi -\rmi \frac{t}{N} H_j \varphi + \mathcal{O}\!\left(\frac{t^2}{N^2}\|H_j^2\varphi\|\right),
\label{eq:Taylor2t}
\end{equation}
for $j=1,2$.
\end{proof}
One might wonder why Thm.~\ref{thm:trotter_thm} does not give rise to a commutator scaling.
However, one can show that general state-dependent Trotter bounds with  a commutator scaling cannot exist:
\begin{thm}\label{thm:no-go_commutator_scaling}
There exist self-adjoint operators $H_{1}$ and $H_{2}$, and vectors $\varphi\in\mathcal{D}(H_1H_2)\cap\mathcal{D}(H_2H_1)$, such that a  Trotter error bound of the form
\begin{equation}
\xi_N(t;\varphi)\leq \omega\Vert [H_1,H_2]\varphi\Vert,
\label{eq:no-go_commutator_scaling}
\end{equation}
with $\omega>0$, does not hold.
\end{thm}
\begin{proof}
Essentially, the reason why a general bound~\eqref{eq:no-go_commutator_scaling} cannot exist is because $[H_1, H_2]\varphi$ does not carry enough information.
Therefore, we can easily construct counter-examples, where such a bound fails to apply.
To this end, we have to find an example with $[H_1, H_2]\varphi=0$ but $[\rme^{-\rmi t H_1},\rme^{-\rmi t H_2}]\varphi\neq0$.
Then, the right-hand side in Eq.~\eqref{eq:no-go_commutator_scaling} is zero, even though the Trotter formula is not exact.
We now give two examples of this situation.
\emph{Example 1} --- Consider the two $3\times3$ Hamiltonians
	\begin{equation*}
		H_1 = \begin{pmatrix}
			0 & 1 & 0 \\
			1 & 0 & 1 \\
			0 & 1 & 1
		\end{pmatrix},
		\quad
		H_2 = \begin{pmatrix}
			0 & -1 & 0 \\
			-1 & 0 & -1 \\
			0 & -1 & 0
		\end{pmatrix},
	\end{equation*}
	and the input state
	\begin{equation*}
		\varphi = \begin{pmatrix}
			1 \\
			0 \\
			0
		\end{pmatrix}.
	\end{equation*}
Notice that $\varphi$ is an eigenstate of $H_1+H_2$ with eigenvalue zero.
Furthermore, we have $\Vert [H_1,H_2]\varphi\Vert=0$, but $\xi_N(t;\varphi)>0$ if $N<\infty$ and $t\neq0$.
\emph{Example 2} --- While Example 1 only goes wrong for certain input states, it is even possible to construct examples where $\Vert [H_1,H_2]\varphi\Vert=0$ but $\xi_N(t;\varphi)>0$ on a \emph{dense} set of states.
For this, consider the model from Ref.~\cite{Burgarth2021}: $H_1=-P^2$ and $H_2=P^2+P$, where $P$ is the one-dimensional momentum operator defined on the real half-line $L^2(\mathbb{R}_+)$.
The operators $H_1$ and $H_2$ are self-adjoint on a common dense domain $D$~\cite[Eq.~(16)]{Burgarth2021}, and  $[H_1,H_2]\varphi=0$ for all $\varphi\in \mathcal{D}(0,+\infty)\subset\mathcal{D}(H_1H_2)\cap\mathcal{D}(H_2H_1)$, which is dense in $L^2(\mathbb{R}_+)$.
Thus, the right-hand side of Eq.~\eqref{eq:no-go_commutator_scaling} is zero.
However, as shown in Ref.~\cite{Burgarth2021}, the sum $H_1+H_2=P$ is not essentially self-adjoint on $D$ and Trotter does not converge strongly on $D$.
Therefore, the actual Trotter error is larger than zero even in the limit $N\rightarrow\infty$.
\end{proof}
Notice that there might not always be eigenstates, i.e.\ if $H_1+H_2$ only admits a purely continuous spectrum.
However, this does not happen in chemistry problems, where the states of bounded electrons (``orbitals'') constitute a set of eigenstates for the atomic or molecular target Hamiltonian.
Note also that Thm.~\ref{thm:trotter_thm} naturally extends to Trotter products with more than two operators.
See the discussion around Eq.~(\ref{eq:first_order_more_operators}).

To treat generic input states, we can apply Thm.~\ref{thm:trotter_thm} to superpositions of eigenstates.
\begin{corol}\label{cor:superposition_states}
For states that are superpositions of eigenstates of $H_1+H_2$,
\begin{equation*}
\psi=\sum_{\ell=1}^L c_{\ell}\varphi_{\ell},
\end{equation*}
with $L$ finite or $L=\infty$ and $(H_1+H_2)\varphi_\ell=h_\ell\varphi_\ell$,
 we obtain for all $g_\ell\in\mathbb{R}$
	\begin{align}
		&\xi_N (t;\psi)
		\le\sum_{\ell=1}^L |c_\ell|\xi_N(t;\varphi_\ell)
		\le\left(\sum_{\ell=1}^L  \xi_N(t;\varphi_\ell)^2\right)^{1/2}
		\nonumber\\
		&\quad\le\frac{t^{2}}{2N}\left[\sum_{\ell=1}^L\Bigl(
		\Vert H_{1}(g_\ell)^{2}\varphi_{\ell}\Vert
		+\Vert H_{2}(h_\ell-g_\ell)^{2}\varphi_{\ell}\Vert\Bigr)^2\right]^{1/2},
		\label{eq:superposition_states}
	\end{align}
where $H_j(g)=H_j-g$ ($j=1,2$), and the right-hand side is defined as $+\infty$ if it diverges~\cite{note3}. 
\end{corol}
\begin{proof}
	This follows by the triangle and Cauchy-Schwarz inequalities.
\end{proof}

Theorem~\ref{thm:trotter_thm} requires the input state to be in the domain intersection $\varphi\in\mathcal{D}(H_{1}^{2})\cap\mathcal{D}(H_{2}^{2})$, so that
$\|H_j^2\varphi\|<\infty$ ($j=1,2$) and the bound does not diverge.
However, this might not always be the case.
In particular, we will see in Sec.~\ref{sec:bound} that this domain condition is violated by the ground state of the hydrogen atom.
For such a case, we establish the following theorem, which works under relaxed domain conditions.
It comes at the price of a slower Trotter convergence speed.
The scaling is ruled by the high-energy distribution of the input state $\varphi$,
\begin{equation*}
\operatorname{\mathrm{Prob}}_{\varphi}(|H_j|\geq \Lambda)
=\|\theta(|H_j|-\Lambda)\varphi\|^2 
=\mu_{j,\varphi}(\{|\lambda|\geq\Lambda\}),
\end{equation*}
where $\theta$ is the Heaviside step function, and
$\mu_{j,\varphi}(\Omega)$ 
with $\Omega\subset\mathbb{R}$ are
the spectral measures of the Hamiltonians $H_j$ ($j=1,2$) at the input state $\varphi$.
\begin{thm}\label{thm:first_order_alpha}
Let $H_1$ be self-adjoint on $\mathcal{D}(H_1)$ and $H_2$ be self-adjoint on $\mathcal{D}(H_2)$. 
Let $\varphi$ be an eigenstate of $H_1+H_2$ with eigenvalue $h$, i.e.\ $(H_1+H_2)\varphi=h\varphi$. 
Furthermore, let the tails of the spectral measures $\mu_{j,\varphi}$ of $H_j$ at $\varphi$ decay as
\begin{equation}
\mu_{j,\varphi}(\{|\lambda|\geq\Lambda\})
=\mathcal{O}\!\left(\frac{1}{\Lambda^{2\delta}}\right),\quad j=1,2,
\label{eq:tail_bound_1}
\end{equation}
for $\Lambda>0$ and some $\delta>1$. 
Then,
\begin{equation*}
\xi_N(t;\varphi)
=\begin{cases}
\medskip
\displaystyle
\mathcal{O}\!\left(\frac{t^\delta}{N^{\delta-1}}\right),&1<\delta<2,\\
\medskip
\displaystyle
\mathcal{O}\!\left(\frac{t^2}{N}\sqrt{\log (N/t)}\right),&\delta=2,\\
\displaystyle
\mathcal{O}\!\left(\frac{t^2}{N}\right),&\delta>2.
\end{cases}
\end{equation*}
\end{thm}
\begin{proof}
The proof for $\delta\in(1,2]$ is provided in Appendix~\ref{sec:proof_trotter_thm}\@.
We use the same strategy of Thm.~\ref{thm:trotter_thm}\@. 
However, the bound on the single-step error will now require a finer analysis than the second-order Taylor expansion~\eqref{eq:Taylor2t}, since $\|H_j^2\varphi\|$ is no longer assured to be bounded. 
By using the spectral theorem, one can show that the high-energy decay of the spectral measures implies instead the following error on the short-time evolution:
\begin{equation*}
\rme^{-\rmi \frac{t}{N} H_j}\varphi
=\varphi -\rmi \frac{t}{N} H_j \varphi + \mathcal{O}\!\left(\frac{t^\delta}{N^\delta}\right),
\end{equation*}
for $\delta<2$, and an analogous error with an additional logarithmic factor for $\delta=2$, whence the theorem follows. 
Finally, the case $\delta>2$ is covered by Thm.~\ref{thm:trotter_thm}, since in this case the decay~\eqref{eq:tail_bound_1} implies that $\varphi\in\mathcal{D}(H_1^2)\cap\mathcal{D}(H_2^2)$. 
\end{proof}
Note that the conditions on the tails of the spectral distributions in Eq.~(\ref{eq:tail_bound_1}) imply that $\varphi\in\mathcal{D}(|H_1|^{\gamma_1})\cap\mathcal{D}(|H_2|^{\gamma_2})$ for all $\gamma_j\in[0,\delta)$.
If $\delta\le2$, it is not guaranteed that $\varphi\in\mathcal{D}(H_1^2)\cap\mathcal{D}(H_2^2)$, and one obtains a convergence slower (by a logarithmic factor) than the convergence $\mathcal{O}(1/N)$ given by Thm.~\ref{thm:trotter_thm}\@.

Notice that the assumption that $\varphi$ is an eigenstate of $H_1+H_2$ and thus $\varphi\in\mathcal{D}(H_1)\cap\mathcal{D}(H_2)$ implies the decay~\eqref{eq:tail_bound_1} of the spectral distributions with $\delta=1$. 
Therefore, the condition~\eqref{eq:tail_bound_1} contains no additional information and is immaterial for $\delta\leq 1$.
As a matter of fact, one can prove the convergence of the Trotter product formula without assuming~\eqref{eq:tail_bound_1} (or equivalently for $\delta\leq 1$). 
However, in such a case the rate of convergence $\xi_N(t;\varphi)\to 0$ as $N\to\infty$ can be arbitrarily slow.

If Eq.~\eqref{eq:superposition_states} converges and $\psi$ is fully supported in the part of the Hilbert space corresponding to the point spectrum of $H_1+H_2$, then the Trotter error scaling is determined by the domain conditions in Thm.~\ref{thm:trotter_thm} and Thm.~\ref{thm:first_order_alpha}\@.
Notice that this is always the case for any bounded $H_1+H_2$, for which the condition~\eqref{eq:tail_bound_1} holds for any arbitrary large $\delta$.
Moreover, this is true for chemistry simulations, where the input state of quantum phase estimation or the Hartree state aims to approximate the target eigenstate well. 
In particular, one is usually interested in the space of bounded electron states of a chemical Hamiltonian, which corresponds to its point spectrum.

Together, Thm.~\ref{thm:trotter_thm} and Thm.~\ref{thm:first_order_alpha} allow us to study the Trotter error for arbitrary eigenstates of the target Hamiltonian.
The bounds can be lifted to generic states employing Cor.~\ref{cor:superposition_states}\@.
Our bound from Thm.~\ref{thm:trotter_thm} is tight in the sense that it is saturated by the example provided in Eq.~\eqref{eq:example}:
If $\epsilon=0$, then the state $\varphi$ in Eq.~\eqref{eq:state} has eigenvalue $h=0$, i.e.\ $(H_1+H_2)\varphi=0$.
Furthermore, $H_1^2\varphi=H_2^2\varphi=0$.
Therefore, choosing $g=0$ leads to the bound $\xi_N(t;\varphi)=0$, which is exact.
Thus, our bounds allow for more efficient quantum simulation due to an explicit and tight state dependency.
In the following, we apply them to the chemistry simulation of the hydrogen atom.
We emphasise that---even though the main example considered here is the hydrogen atom---our bounds are applicable much more broadly.
For example, they could be used to study the error in symplectic integration for eigenstates of smoothly perturbed Hamiltonians $H=H_0+\epsilon H_\mathrm{smooth}$ ($\epsilon\ll1$).
See for instance Ref.~\cite{McLachlan2002}\@.

\section{Application of the first-order bounds to the hydrogen atom}
\label{sec:bound}
The Hamiltonian of the hydrogen atom is given by
\begin{equation*}
H_\mathrm{hydrogen} = H_1 + H_2, 
\end{equation*}
with (temporarily re-establishing the Planck constant $\hbar$)
\begin{equation}
H_{1}=-\frac{\hbar^{2}}{2m_{\rme}}\Delta
\label{eq:hamiltonian_kin}
\end{equation}
and 
\begin{equation}
H_{2}=-\frac{\hbar^{2}}{m_{\rme}a_{0}r},
\label{eq:hamiltonian_pot}
\end{equation}
where
$m_\rme$ is the electron mass (or equivalently the reduced electron mass), $e$ is the elementary (electron) charge, and 
\begin{equation*}
a_{0}=\frac{4\pi\varepsilon_{0}\hbar^{2}}{m_{\rme}e^{2}}
\end{equation*}
is the Bohr radius, with $\varepsilon_0$ being the vacuum permittivity.
Furthermore, $\Delta=\nabla^2$ is the Laplace operator, and $r=\vert \mathbf{r} \vert$ denotes the distance from the nucleus, which is placed at $\mathbf{r}=(0, 0, 0)$.
The eigenfunctions $\Psi_{n\ell m}(\mathbf{r})$
of the hydrogen atom have eigenenergies
\begin{equation*}
E_n=-\frac{\hbar^{2}}{2m_{\rme}a_{0}^{2}n^2},
\end{equation*}
and are labeled by the principal quantum number $n=1,2,\ldots$, as well as $\ell=0,1,\ldots,n-1$ and $m=-\ell,-\ell+1,\ldots,\ell-1,\ell$, which determine the orbital angular momentum.

To implement the dynamics under the Hamiltonian $H_\mathrm{hydrogen}=H_1+H_2$ of the hydrogen atom, one Trotterizes between the kinetic energy $H_1$ and the potential energy $H_2$.
This approach is called \emph{split-step method} and is commonly used in the numerical analysis of quantum chemistry, see e.g.\ Refs.~\cite{Chan2023, Kosloff1988, Choi2019}.
For the advantages of this approach over other methods, see Refs.~\cite{Kosloff1988, Choi2019, Roulet2019, Tannor2008}.
If we want to know the speed of Trotter convergence, we can compute the distance between the Trotterized dynamics and the actual dynamics. 
For finite-dimensional systems, this can be bounded by Eq.~\eqref{eq:operator_norm_bound}.
However, it is easy to see that for the hydrogen atom, we have 
\begin{equation}
\Vert [H_{1},H_{2}]\Vert
\propto\left\Vert \left[\frac{\Delta}{2},\frac{1}{r}\right]\right\Vert 
=\left\Vert \frac{1}{r^2}\frac{\rmd}{\rmd r}\right\Vert
=\infty.
\label{eq:commutator_diverge}
\end{equation}
Physically, this comes from the fact that the commutator is a differential operator (involving radial momentum) and has a singularity at $r=0$: 
There are wave functions in the Hilbert space with arbitrary large radial momentum and/or large support (high electron density) close to this singularity at the nucleus.
From a mathematical perspective, this is not a surprise, since both $H_1$ and $H_2$ are unbounded and we need state-dependent bounds.
See, for instance, Refs.~\cite{Simon2005, Ichinose2004, Burgarth2023}.
Since the task of quantum chemistry simulation is to find the spectrum of a Hamiltonian, it is perfectly reasonable to look at the convergence speed
on its eigenstates.
That is, the figure of merit for the hydrogen atom is
\begin{equation*}
	\xi_N(t;\Psi_{n\ell m}) = \left\Vert 
	\left[
	\left(\rme^{-\rmi \frac{t}{N}H_2}\rme^{-\rmi \frac{t}{N}H_1}\right)^{N} -\rme^{-\rmi t E_n}
	\right]\Psi_{n\ell m}
	\right\Vert.
\end{equation*}
It is well known that the eigenfunctions of the hydrogen atom are given by the product of the radial wave function $R_{n\ell}$ and the spherical harmonics $Y_{\ell m}$.
The spherical part $Y_{\ell m}$ only accounts for the degeneracies in the spectrum of the hydrogen atom, so different $m$ values for fixed $n$ and $\ell$ lead to the same Trotter dynamics.
Therefore, we can restrict our attention to the radial part $R_{n\ell}$ when computing $\xi_N$.
Now, we apply our bound from Thm.~\ref{thm:trotter_thm} to this Trotter scenario.
For example, we could consider the excited state $\Psi_{320}$ of the hydrogen atom, for which we can bound
\begin{equation*}
	\xi_N(t;\Psi_{320})\leq \frac{1}{40 t_0^2}\frac{t^2}{N}, \qquad t_0=\frac{m_\mathrm{e}a_0^2}{\hbar}.
\end{equation*}
As opposed to the standard bound from Eq.~\eqref{eq:operator_norm_bound}, which just gives the trivial bound for the hydrogen atom due to $\|[H_1, H_2]\|=\infty$, this bound is nontrivial and has the expected scaling of $\mathcal{O}(t^2/N)$.

The ground state is of particular importance for quantum chemistry applications.
However, our bound from Thm.~\ref{thm:trotter_thm} diverges in this case as the domain conditions are not satisfied by the ground state $\Psi_{100}$ of the hydrogen atom.
Therefore, we need to perform a more refined analysis through Thm.~\ref{thm:first_order_alpha}\@.
This indicates that the Trotter product formula converges \emph{slower} than $\mathcal{O}(N^{-1})$ for the ground state $\Psi_{100}$ of the hydrogen atom.
In fact, we find a slower scaling of the Trotter error for input states with low angular momenta.
That is, all s-orbitals only admit the $N^{-1/4}$ scaling and all p-orbitals lead to the $N^{-3/4}$ scaling.
Importantly, we emphasize that the quantum number $\ell$ of the orbital angular momentum alone determines the state-dependent Trotter scaling.
We summarize our findings in Table~\ref{tab:order_results}\@.
\begin{table}
	\caption{\label{tab:order_results}Scalings of the analytical first-order Trotter error bounds for the eigenfunctions $\Psi_{n\ell m}$ of the hydrogen atom.
	See Sec.~\ref{sec:trotter_thm} and specifically Eqs.~\eqref{eq:bound_R_n0}--\eqref{eq:bound_R_nl}.
	We find different scalings depending on the input states, which are solely determined by the orbital angular momentum $\ell$.
	In particular, the scalings for the s- and p-orbitals are slower than the expected $N^{-1}$ scaling.}
	\begin{center}
		\begin{tabular}{|c|c|}
		\hline
		\begin{tabular}[c]{@{}c@{}}  \bf{Orbital angular momen-} \\ \bf{tum quantum number} $\bm{\ell}$ \end{tabular} &  \begin{tabular}[c]{@{}c@{}} \bf{Scaling of bound on} \\ \bf{Trotter error} $\bm{\xi_N(t;\Psi_{n\ell m})}$ \end{tabular} \\
			\hhline{:==:}
			$\ell = 0$ (s-orbitals) & $\mathcal{O}(N^{-1/4})$\\
			$\ell = 1$ (p-orbitals) & $\mathcal{O}(N^{-3/4})$\\
			$\ell \geq 2$ (d-orbitals and higher)& $\mathcal{O}(N^{-1})$\\
			\hline
		\end{tabular}
	\end{center}
\end{table}
We also provide explicit error bounds in Sec.~\ref{sec:trotter_thm} of the Methods section. 
See Eqs.~\eqref{eq:bound_R_n0}--\eqref{eq:bound_R_nl}.
These bounds allow us to compute the necessary resources to achieve a chemical accuracy for an algorithm for a chemistry simulation.
For instance, assume one would like to implement the dynamics under the hydrogen atom Hamiltonian for the ground state $\Psi_{100}$ for a total evolution time $t$ of the order of microseconds.
To reach a chemical accuracy~\cite{note4}, 
one would need Trotter steps $N$ of the order $10^{3}$--$10^{4}$\@.

Since the results in Table~\ref{tab:order_results} are just upper bounds, it makes sense to convince ourselves numerically that we indeed obtain a slower scaling for the ground state than for higher excited states.
To compute the Trotter error, we numerically integrate the Schr\"{o}dinger equation using XMDS2~\cite{Dennis2013}, an open-source package for solving multidimensional partial differential equations.
For these simulations, we use a grid-based method and work in the Hartree atomic units, so that $\hbar=m_\rme=a_0=1$ and the Hamiltonians become $H_1=-\Delta/2$ and $H_2=-1/r$.
We diagonalize the kinetic energy $H_2=-\Delta/2$ by representing it in its spectral basis of spherical Bessel functions.
The number of Bessel functions used corresponds to the number of grid points in position space as well as the number of modes in the spectral space, resulting in a discretized system suitable for numerical integration.
More details on the numerical simulation methods can be found in Sec.~\ref{sec:methods_numerical} of the Methods section.

We simulated the system multiple times, each using a different number of modes while keeping the radial cutoff constant.
The results are shown in Fig.~\ref{figNumericsGroundState}\@.
The ground-state Trotter error initially scales as $N^{-1/4}$, but asymptotes to $N^{-1}$.
This transition occurs later and later (at a larger number of the Trotter steps $N$) as the number of Bessel modes (numerical cutoff dimension) increases.
Notice that the eigenstate Trotter error always asymptotically scales as $N^{-1}$ for finite-dimensional systems~\cite{Burgarth2023}.
Therefore, for a finite number of Bessel modes, we eventually see the $N^{-1}$ scaling for large enough $N$.
For this reason, a later transition point to the $N^{-1}$ scaling with increased system dimension is the only way to numerically observe slower convergence rates for the full unbounded Hamiltonians.
\begin{figure}
\centering
	\begin{tikzpicture}[mark size={1.5}, scale=1]
			\begin{axis}[
			xmode=log,
			ymode=log,
			xlabel={Trotter steps $N$},
			ylabel=First-order Trotter error $\xi_N(t;\Psi_{100})$,
			label style={font=\footnotesize},
			tick label style={font=\footnotesize},
			x post scale=1,
			y post scale=1,
			legend style={at={(0.4,-0.2)},anchor=north},
			legend columns=3,
			legend cell align={left},
			xmin=2, xmax=3000,
			]
			\addplot[color=1, thick, mark=*, only marks] table[x=N, y=error, col sep=comma] {Trotter_Error_1st_order_groundstate_100.csv};
			\addplot[color=2, mark=*, only marks] table[x=N, y=error, col sep=comma] {Trotter_Error_1st_order_groundstate_200.csv};
			\addplot[color=3, mark=*, only marks] table[x=N, y=error, col sep=comma] {Trotter_Error_1st_order_groundstate_300.csv};
			\addplot[color=4, mark=*, only marks] table[x=N, y=error, col sep=comma] {Trotter_Error_1st_order_groundstate_400.csv};
			\addplot[color=5, mark=*, only marks] table[x=N, y=error, col sep=comma] {Trotter_Error_1st_order_groundstate_800.csv};
			\addplot[color=0, ultra thick, domain=2:3000, densely dashed]{0.3*x^(-1)};
			\addplot[color=6, ultra thick, domain=2:3000]{2.04665*x^(-1/2) + 0.832107/x + 1.34365/x^(1/4)};
			\addplot[color=0, ultra thick, domain=2:3000]{0.31/x + 0.1/x^(1/2) + 0.1/x^(1/4)};
			\legend{
			\footnotesize $100$ modes,
			\footnotesize $200$ modes,
			\footnotesize $300$ modes,
			\footnotesize $400$ modes,
			\footnotesize $800$ modes,
			\footnotesize $N^{-1}$,
			\footnotesize analytic bound,
			\footnotesize $\infty$ mode scaling,
			};
			\end{axis}
			\node[color=6,rotate=-8] at (4,4.75) {\footnotesize upper bound};
			\node[color=0,rotate=-8] at (4,3.5) {\footnotesize scaling for $\infty$ modes};
			\draw [-stealth, ultra thick] (3,2) node[below, align=left] {\footnotesize increasing\\  \footnotesize \# of modes} -- (4.2,3);
			\node[color=0,rotate=-25] at (0.7,3.15) {\footnotesize $\mathcal{O}(N^{-1})$};
		\end{tikzpicture}
		\caption{\label{figNumericsGroundState}The Trotter error at time $t=1$ in the Hartree atomic units $\hbar=m_\mathrm{e}=a_0=1$, as a function of the Trotter steps for the ground state $\Psi_{100}$ of the hydrogen atom.
		The radial cutoff in the simulations is $R=30$.
		We show five different levels of discretization characterized by the number of Bessel modes.
		For reference, we show the slopes of $N^{-1}$ (grey dashed line) and our bound (brown solid line), which scales as $N^{-1/4}$.
		See Eqs.~\eqref{eqErrorBoundhydrogenGround1}--\eqref{eqErrorBoundhydrogenGround2}.
		We see that the asymptotics for any finite discretization initially start in a consistent way with our bound, and eventually go as $N^{-1}$, but $N$ at which this transition happens becomes larger with increasing number of modes.
		This provides an evidence that the scaling in the infinite-mode limit is indeed slower than $N^{-1}$.
		To indicate a potential curve for the infinite-mode limit, we show a slope (gray solid line) with the same scaling as our analytic bound.}
\end{figure}
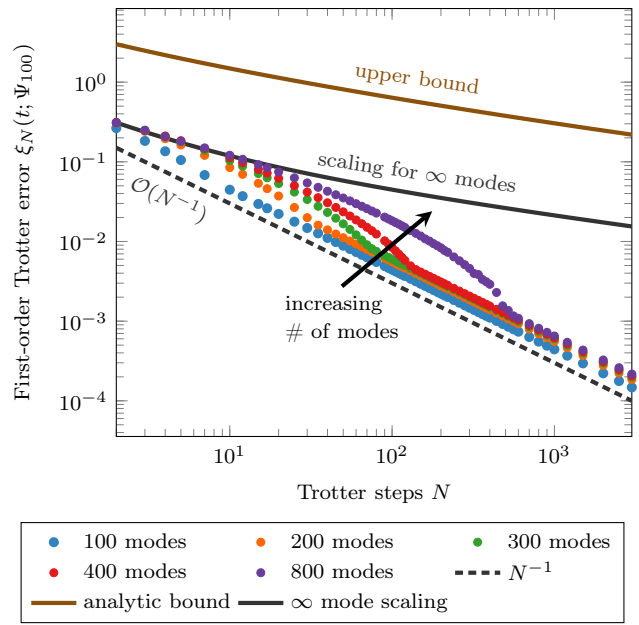
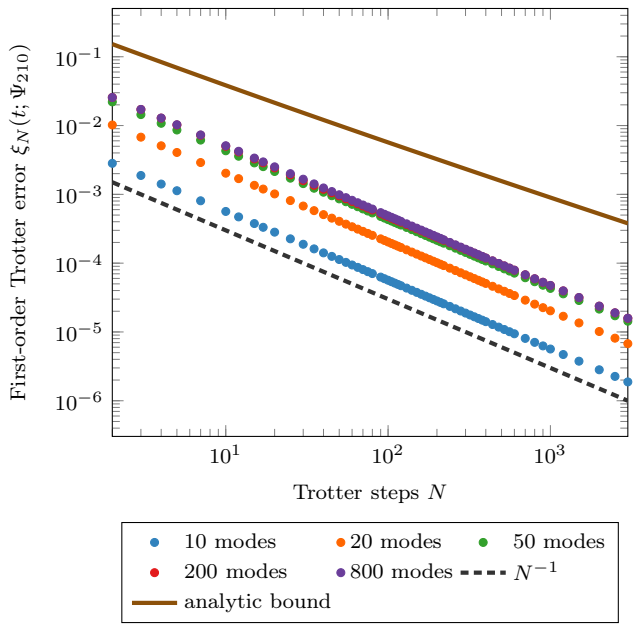
\begin{figure}
	\centering
		\begin{tikzpicture}[mark size={1.5}, scale=1]
			\begin{axis}[
			xmode=log,
			ymode=log,
			xlabel={Trotter steps $N$},
			ylabel={First-order Trotter error $\xi_N(t;\Psi_{210})$},
			label style={font=\footnotesize},
			tick label style={font=\footnotesize},
			x post scale=1,
			y post scale=1,
			legend style={at={(0.5,-0.2)},anchor=north},
			legend columns=3,
			legend cell align={left},
			xmin=2, xmax=3000,
			]
			\addplot[color=1, mark=*, only marks] table[x=N, y=error, col sep=comma] {Trotter_Error_1st_order_n2_l1_10.csv};
			\addplot[color=2, mark=*, only marks] table[x=N, y=error, col sep=comma] {Trotter_Error_1st_order_n2_l1_20.csv};
			\addplot[color=3, mark=*, only marks] table[x=N, y=error, col sep=comma] {Trotter_Error_1st_order_n2_l1_50.csv};
			\addplot[color=4, mark=*, only marks] table[x=N, y=error, col sep=comma] {Trotter_Error_1st_order_n2_l1_200.csv};
			\addplot[color=5, mark=*, only marks] table[x=N, y=error, col sep=comma] {Trotter_Error_1st_order_n2_l1_800.csv};
			\addplot[color=0, ultra thick, domain=2:3000, densely dashed]{0.003*x^(-1)};
			\addplot[color=6, ultra thick, domain=2:3000]{0.145959/x + 0.133831/x^(3/4)};
			\legend{
			\footnotesize $10$ modes,
			\footnotesize $20$ modes,
			\footnotesize $50$ modes,
			\footnotesize $200$ modes,
			\footnotesize $800$ modes,
			\footnotesize $N^{-1}$,
			\footnotesize analytic bound,
			}
			\end{axis}
		\end{tikzpicture}
		\caption{\label{figNumericsn2l1State}The Trotter error at time $t=1$ in the Hartree atomic units $\hbar=m_\mathrm{e}=a_0=1$, as a function of the Trotter steps for the state $\Psi_{210}$ of the hydrogen atom.
		The radial cutoff in the simulations is $R=30$.
		We show five different levels of discretization characterized by the number of Bessel modes; however, beyond 200 modes (red) the results are indistinguishable.
		The grey dashed line shows the slope of the $N^{-1}$ scaling, and we see that the Trotter error scales as $N^{-1}$ in all cases, even though the state $\Psi_{210}$ does not satisfy the domain conditions for Thm.~\ref{thm:trotter_thm}\@.
		Our analytic bound, which scales as $N^{-3/4}$, is depicted in brown.
		}
\end{figure}
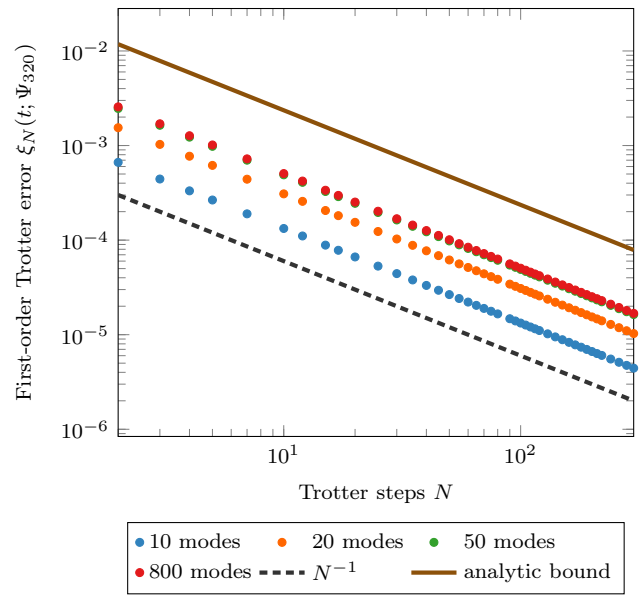
\begin{figure}
	\centering
		\begin{tikzpicture}[mark size={1.5}, scale=1]
			\begin{axis}[
			xmode=log,
			ymode=log,
			xlabel={Trotter steps $N$},
			ylabel={First-order Trotter error $\xi_N(t;\Psi_{320})$},
			label style={font=\footnotesize},
			tick label style={font=\footnotesize},
			x post scale=1,
			y post scale=1,
			legend style={at={(0.5,-0.2)},anchor=north},
			legend columns=3,
			legend cell align={left},
			xmin=2, xmax=300,
			]
			\addplot[color=1, mark=*, only marks] table[x=N, y=error, col sep=comma] {Trotter_Error_n3_l2_10.csv};
			\addplot[color=2, mark=*, only marks] table[x=N, y=error, col sep=comma] {Trotter_Error_n3_l2_20.csv};
			\addplot[color=3, mark=*, only marks] table[x=N, y=error, col sep=comma] {Trotter_Error_n3_l2_50.csv};
			\addplot[color=4, mark=*, only marks] table[x=N, y=error, col sep=comma] {Trotter_Error_n3_l2_800.csv};
			\addplot[color=0, ultra thick, domain=2:300, densely dashed]{0.0006*x^(-1)};
			\addplot[color=6, ultra thick, domain=2:300]{0.023591/x};
			\legend{
			\footnotesize $10$ modes,
			\footnotesize $20$ modes,
			\footnotesize $50$ modes,
			\footnotesize $800$ modes,
			\footnotesize $N^{-1}$,
			\footnotesize analytic bound,
			}
			\end{axis}
		\end{tikzpicture}
		\caption{\label{figNumericsn3l2State}The Trotter error at time $t=1$ in the Hartree atomic units $\hbar=m_\mathrm{e}=a_0=1$, as a function of the Trotter steps for the state $\Psi_{320}$ of the hydrogen atom.
		The radial cutoff in the simulations is $R=40$.
		We use four different levels of discretization characterized by the number of Bessel modes; however, beyond 50 modes (green) the results are indistinguishable.
		The grey dashed line shows the slope of the $N^{-1}$ scaling, and we see that the Trotter error scales as $N^{-1}$ in all cases as expected.
		For reference, we plot our analytic bound in brown.
		It also shows an $N^{-1}$ behavior.}
\end{figure}

We also investigate the Trotter error scaling for the excited state $\Psi_{210}$ of the hydrogen atom with $n=2$, $\ell=1$, and $m=0$ in Fig.~\ref{figNumericsn2l1State}, as well as for the excited state $\Psi_{320}$ with $n=3$, $\ell=2$, and $m=0$ in Fig.~\ref{figNumericsn3l2State}\@.
In both cases, we find a $N^{-1}$ scaling independent of the number of Bessel modes.
This behavior shows that the Trotter errors for the ground and excited states scale differently in the full system limit of infinitely many Bessel modes.
More specifically, excited states admit $N^{-1}$ scaling, whereas the ground state only scales as $N^{-1/4}$.
These results are rather remarkable as the ground state is usually the target of quantum chemistry simulations.
Its Trotter error scales slower than expected, indicating that more resources (quantum gates) are needed for accurate quantum chemistry simulation.
Ultimately, this leads to a polynomially increased runtime of ground-state quantum chemistry algorithms.
Furthermore, the increased gate count can intensify the errors of the quantum computation due to environmental decoherence or gate imperfections.

Physically, the slower scaling of the Trotter error in the ground state, and in all s-orbitals, can be explained by the fact that the electron is close to the $1/r$ singularity.
Due to this, it has very high fluctuations in both potential and kinetic energies.
In particular, we will see in Appendix~\ref{sec:hydrogen_bounds} that the fourth moments, i.e.\ $\|H_j^2\varphi\|$, are diverging.
When performing the Trotterization, in each period the state evolves under the bare potential and kinetic Hamiltonians.
Due to its huge energy fluctuations, there is a significant probability for the electron to be kicked out of the space of the bound states by this evolution.
In turn, the hydrogen atom gets ionized.
As we increase the number of Trotter steps $N$, each Trotter cycle becomes shorter and we approximate the target dynamics under the Hamiltonian of the hydrogen atom better.
Therefore, the ionization effect decreases in $N$.
If we truncate the system to make it amenable to numerical simulations, we will always cut off certain modes corresponding to free electron states.
As a result, these modes are not attained from the Trotter evolution and a smaller ionization effect is observed numerically.
Going to higher and higher truncation dimensions allows for more room in the space of free electron states.
In turn, we can numerically see a stronger ionization effect, which leads to a worse approximation of the target dynamics.
This explains the scaling behavior we see numerically in Fig.~\ref{figNumericsGroundState}\@.
Ultimately, in the limit of infinitely many modes, there is always a significant probability of leaking out of the bound space no matter how large the Trotter number $N$ is.
Therefore, the full Trotter problem converges slower than $N^{-1}$.
We numerically observe this ionization behavior and show the results in Fig.~\ref{fig:ionization}\@.
It is noticeable that the s-orbitals admit particularly strong ionization behaviors.
This is consistent with our analytical bounds as well as the numerical results in Figs.~\ref{figNumericsGroundState}--\ref{figNumericsn3l2State}\@.
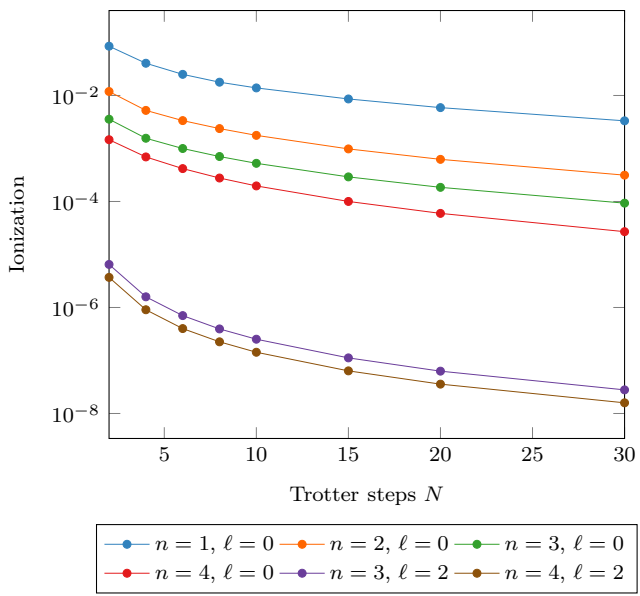
\begin{figure}
	\centering
		\begin{tikzpicture}[mark size={1.5}, scale=1]
			\begin{axis}[
			ymode=log,
			xlabel={Trotter steps $N$},
			ylabel={Ionization},
			label style={font=\footnotesize},
			tick label style={font=\footnotesize},
			x post scale=1,
			y post scale=1,
			legend style={at={(0.5,-0.2)},anchor=north},
			legend columns=3,
			legend cell align={left},
			xmin=2, xmax=30,
			]
			\addplot[color=1, mark=*, ultra thin] table[x=N, y=error, col sep=comma] {Ionization_n1_l0.csv};
			\addplot[color=2, mark=*, ultra thin] table[x=N, y=error, col sep=comma] {Ionization_n2_l0.csv};
			\addplot[color=3, mark=*, ultra thin] table[x=N, y=error, col sep=comma] {Ionization_n3_l0.csv};
			\addplot[color=4, mark=*, ultra thin] table[x=N, y=error, col sep=comma] {Ionization_n4_l0.csv};
			\addplot[color=5, mark=*, ultra thin] table[x=N, y=error, col sep=comma] {Ionization_n3_l2.csv};
			\addplot[color=6, mark=*, ultra thin] table[x=N, y=error, col sep=comma] {Ionization_n4_l2.csv};
			\legend{
			\footnotesize {$n=1$, $\ell=0$},
			\footnotesize {$n=2$, $\ell=0$},
			\footnotesize {$n=3$, $\ell=0$},
			\footnotesize {$n=4$, $\ell=0$},
			\footnotesize {$n=3$, $\ell=2$},
			\footnotesize {$n=4$, $\ell=2$},
			}
			\end{axis}
		\end{tikzpicture}
		\caption{\label{fig:ionization}Numerical simulations of the ionization probability as a function of the number of Trotter steps. 
		The initial state $\varphi$ is prepared in various energy eigenstates $\Psi_{n\ell m}$ of the hydrogen atom, and then we let it evolve to $\varphi(t)=U_N(t)\varphi$ by the Trotter evolution $U_N(t)$, with the Hamiltonians $H_1$ and $H_2$ given by Eqs.~\eqref{eq:hamiltonian_kin} and \eqref{eq:hamiltonian_pot}, for a total time $t=1$ in the Hartree atomic units $\hbar=m_\mathrm{e}=a_0=1$. 
		The ionization probability is given by $1-\|P_{\mathrm{bd}}\varphi(t)\|^2$, where $P_{\mathrm{bd}}=\sum_{n\ell m} |\Psi_{n\ell m}\rangle \langle\Psi_{n\ell m}|$ is the projection on the space of bound states of the hydrogen atom.
		For the evolution under the non-Trotterized Hamiltonian $H = H_1 + H_2$, the ionization fraction was zero up to numerical errors $(<10^{-10})$ as expected. 
		For the Trotterized evolution, $\varphi(t)$ acquires a nonzero component out of the space spanned by the bound states, resulting in a nonzero ionization probability, as shown in the figure. 
		The $\ell=0$ eigenstates ionize much more heavily than the $\ell=2$ eigenstates, but in both cases, the ionization rate decreases with the number of Trotter steps as the Trotter approximation approaches the true evolution.}
\end{figure}

Mathematically, the slower convergence of the ground state can be explained by the fact that both $H_1$ and $H_2$ are unbounded operators.
It is well known in mathematical physics that, for unbounded operators, the Trotter product formula can potentially converge slower (see e.g.\ Ref.~\cite{Ichinose2004,Jahnke2000}) or even not at all (see e.g.\ Refs.~\cite{Zhu1993, Arenz2018}).
In our context, slower convergence arises from the circumstance that the ground state of the hydrogen atom is not in the domain of $H_1^2$ and $H_2^2$, see Appendix~\ref{sec:hydrogen_bounds}\@.
On the contrary, the state $\Psi_{320}$ with $n=3$, $\ell=2$, and $m=0$ considered in Fig.~\ref{figNumericsn3l2State} satisfies this domain condition.
This allowed us to derive analytical bounds that scale as $N^{-1}$, which is consistent with the numerical simulation.

For the simulation of non-eigenstates, we can use Cor.~\ref{cor:superposition_states} to bound the Trotter error.
Here, the bound follows from the triangle and Cauchy-Schwarz inequalities after representing the state in the basis of eigenstates of $H_1+H_2$.
In this case, the contributing eigenstate with the slowest scaling in the superposition determines the overall error scaling for the state.
Therefore, we expect the slower scaling to be a generic behaviour even if the exact decomposition into eigenstates is unknown.
This suspicion is bolstered numerically in Fig.~\ref{figNumericsSTO-3G}, where we study the Trotter error for the state  $\Psi_{100}^\text{STO-3G}(r)=0.44\,\rme^{-0.11 r^2} + 0.53\,\rme^{-0.41 r^2} + 0.15\, \rme^{-2.23 r^2}$.
$\Psi_{100}^\text{STO-3G}$ approximates the ground state $\Psi_{100}$ of the hydrogen atom, where three primitive Gaussians are fitted via the least squares method.
This is a standard basis set truncation in computational chemistry known as STO-3G, see, e.g.\ Ref.~\cite{Young2001}\@.
Notice that $\Psi_{100}^\text{STO-3G}\not\in\mathcal{D}(H_1^2)\cap\mathcal{D}(H_2^2)$ so that Thm.~\ref{thm:trotter_thm} does not apply and we expect a slower error scaling (according to Thm.~\ref{thm:first_order_alpha} and Cor.~\ref{cor:superposition_states}).
Indeed, we find exactly the same Trotter error behavior for $\Psi_{100}^\text{STO-3G}$ as for the ground state $\Psi_{100}$ of the hydrogen atom shown in Fig.~\ref{figNumericsGroundState}\@.
That is, the error initially starts with a slower scaling and eventually goes as $N^{-1}$.
This transition shifts to larger Trotter numbers $N$ when the numerical truncation dimension is increased.
Thus, in the full infinite-dimensional system limit, the transition will not occur indicating an overall slower scaling for $\Psi_{100}^\text{STO-3G}$, see Fig.~\ref{figNumericsSTO-3G}\@.

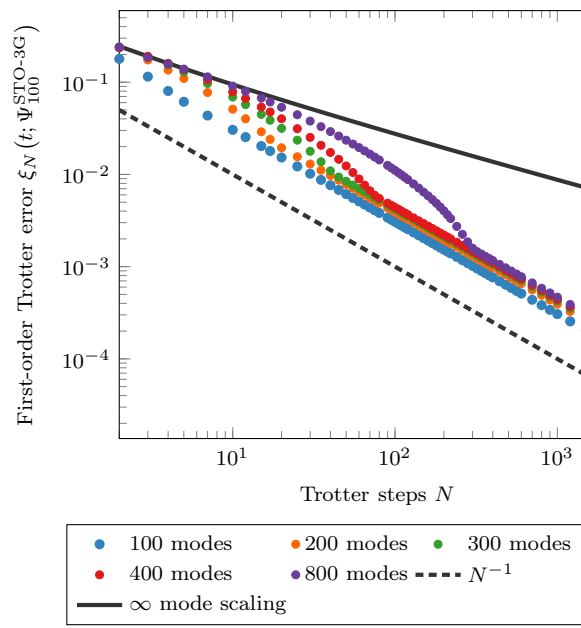
\begin{figure}
\centering
	\begin{tikzpicture}[mark size={1.5}, scale=1]
			\begin{axis}[
			xmode=log,
			ymode=log,
			xlabel={Trotter steps $N$},
			ylabel=First-order Trotter error $\xi_N\big(t;\Psi_{100}^\text{STO-3G}\big)$,
			label style={font=\footnotesize},
			tick label style={font=\footnotesize},
			x post scale=1,
			y post scale=1,
			legend style={at={(0.4,-0.2)},anchor=north},
			legend columns=3,
			legend cell align={left},
			xmin=2, xmax=3000,
			]
			\addplot[color=1, thick, mark=*, only marks] table[x=N, y=error, col sep=comma] {Trotter_Error_100_STO3G.csv};
			\addplot[color=2, mark=*, only marks] table[x=N, y=error, col sep=comma] {Trotter_Error_200_STO3G.csv};
			\addplot[color=3, mark=*, only marks] table[x=N, y=error, col sep=comma] {Trotter_Error_300_STO3G.csv};
			\addplot[color=4, mark=*, only marks] table[x=N, y=error, col sep=comma] {Trotter_Error_400_STO3G.csv};
			\addplot[color=5, mark=*, only marks] table[x=N, y=error, col sep=comma] {Trotter_Error_800_STO3G.csv};
			\addplot[color=0, ultra thick, domain=2:3000, densely dashed]{0.1*x^(-1)};
			\addplot[color=0, ultra thick, domain=2:3000]{0.13/x + 0.25/x^(1/2) + 0.004/x^(1/4)};
			\legend{
			\footnotesize $100$ modes,
			\footnotesize $200$ modes,
			\footnotesize $300$ modes,
			\footnotesize $400$ modes,
			\footnotesize $800$ modes,
			\footnotesize $N^{-1}$,
			\footnotesize $\infty$ mode scaling,
			};
			\end{axis}
		\end{tikzpicture}
		\caption{\label{figNumericsSTO-3G}The Trotter error at time $t=1$ in the Hartree atomic units $\hbar=m_\mathrm{e}=a_0=1$, as a function of the Trotter steps for an approximation $\Psi_{100}^\text{STO-3G}(r)=0.44\,\rme^{-0.11 r^2} + 0.53\,\rme^{-0.41 r^2} + 0.15\, \rme^{-2.23 r^2}$ of the ground state of the hydrogen atom in the STO-3G representation.
		The radial cutoff in the simulations is $R=40$.
		We show five different levels of discretization characterized by the number of Bessel modes.
		For reference, we show the slopes of $N^{-1}$ (grey dashed line) and $N^{-1/4}$ (grey line).
		We see that the asymptotics for any finite discretization initially start in a consistent way with the grey $N^{-1/4}$ line, and eventually go as $N^{-1}$, but $N$ at which this transition happens becomes larger with increasing number of modes.
		This provides an evidence that the scaling in the infinite-mode limit is indeed slower than $N^{-1}$.
		}
\end{figure}

\section{Second-order Trotter bounds}
\label{sec:second_order}
A more refined method to simulate quantum systems that is commonly used in both classical and quantum simulations is the second-order Trotterization.
It is obtained by symmetrizing the Trotter product as
\begin{equation}
	\mathcal{S}_N^{(2)}(t)=\left(\rme^{-\rmi \frac{t}{2N}H_1}\rme^{-\rmi \frac{t}{N}H_2}\rme^{-\rmi \frac{t}{2N}H_1}\right)^N,
	\label{eq:second_order}
\end{equation}
or alternatively,
\begin{equation}
	\tilde{\mathcal{S}}_N^{(2)}(t)=\left(\rme^{-\rmi \frac{t}{2N}H_2}\rme^{-\rmi \frac{t}{N}H_1}\rme^{-\rmi \frac{t}{2N}H_2}\right)^N.\label{eq:second_order_reversed}
\end{equation}
Through this approach, one obtains a tighter error scaling of $\mathcal{O}(t^3/N^2)$ in the case of bounded operators.
In particular~\cite[Prop.~10]{Childs2021},
\begin{align}
	\Vert\mathcal{S}_N^{(2)}(t) - \rme^{-\rmi t(H_1+H_2)}\Vert \le{}& \frac{t^3}{12N^2}\Vert [H_2, [H_2, H_1]]\Vert\nonumber\\
		&{}+\frac{t^3}{24N^2}\Vert [H_1, [H_1, H_2]]\Vert,
		\label{eq:norm_bound_2nd_order}
\end{align}
and analogously for $\tilde{\mathcal{S}}_N^{(2)}(t)$.
In the following, we will focus on the $H_1H_2H_1$ scheme given by the product formula $\mathcal{S}_N^{(2)}(t)$.
Nevertheless, we can always derive equivalent results for the $H_2H_1H_2$ scheme $\tilde{\mathcal{S}}_N^{(2)}(t)$ by exchanging $H_1\leftrightarrow H_2$.
In fact, both product formulas only differ by boundary terms.
As in the case of the first-order Trotterization, the bound in Eq.~\eqref{eq:norm_bound_2nd_order} diverges for unbounded operators such as the kinetic and potential energies of the hydrogen atom.
See also Ref.~\cite{Ichinose2004}.
Therefore, we develop explicitly state-dependent bounds for the second-order Trotterization, which enable us to study such problems.
For this, we define
\begin{equation}
	\xi_N^{(2)}(t;\varphi)\equiv \Vert [\mathcal{S}_N^{(2)}(t)-\rme^{-\rmi ht}]\varphi\Vert,
	\label{eq:second_order_error}
\end{equation}
where again $(H_1+H_2)\varphi=h\varphi$.
We then find the following state-dependent bound on the second-order Trotter product.
\begin{thm}
\label{thm:2nd_order}
Let $H_1$ and $H_2$ be self-adjoint on their respective domains $\mathcal{D}(H_1)$ and $\mathcal{D}(H_2)$, and let $(H_1+H_2)\varphi=h\varphi$.
For $g\in\mathbb{R}$, define ${H}_1(g)=H_1-g$ and ${H}_2(g)=H_2-h+g$.
Then, the state-dependent Trotter error of $\mathcal{S}_N^{(2)}(t)$ can be bounded for all $g\in\mathbb{R}$ and $t\geq 0$ by
	\begin{align*}
		\xi_N^{(2)}(t;\varphi)\leq \frac{t^3}{N^2}\,\biggl(&
		\frac{1}{24}\Vert {H}_1(g)^3\varphi\Vert + \frac{1}{8}\Vert {H}_2(g){H}_1(g)^2\varphi\Vert\nonumber\\
			&\qquad\qquad\qquad\quad\ \
			{}+\frac{1}{12}\Vert {H}_2(g)^3\varphi\Vert\biggr),
	\end{align*}
provided that $\varphi\in\mathcal{D}(H_1^3)\cap\mathcal{D}(H_2H_1^2)\cap\mathcal{D}(H_2^3)$. 
\end{thm}
\begin{proof}
	This is proved in Appendix~\ref{sec:higher_order_derivation}\@.
	For the proof, we again shift the energy of $H_1+H_2$ with respect to $\varphi$ to zero to turn the target evolution on $\varphi$ to the identity $I$.
	Then, we use integration by parts to rewrite the error operator $\mathcal{S}_N^{(2)}(t)-I$.
	By noticing that the boundary terms disappear, we are only left with a remainder, that admits the desired scaling of $\mathcal{O}(t^3/N^2)$.
	Explicitly bounding this remainder gives the theorem.
\end{proof}
As in the case of the first-order Trotterization, it is possible to state a no-go theorem for the existence of the second-order Trotter bounds with commutator scaling by making the same argument as in Thm.~\ref{thm:no-go_commutator_scaling}\@.
In fact, this reasoning can be extended to arbitrary higher-order product formulas.
Furthermore, the bound from Thm.~\ref{thm:2nd_order} can be extended to generic input states by considering superpositions of eigenstates.
See Cor.~\ref{cor:superposition_states}\@.

Notice that the domain conditions on the input state $\varphi$ in Thm.~\ref{thm:2nd_order} are even stronger than the ones for the first-order Trotterization in Thm.~\ref{thm:trotter_thm}\@.
For this reason, the ground state $\Psi_{100}$ of the hydrogen atom does not satisfy them and we expect a slower scaling behavior also for the second-order Trotterization.
As in the case of the first-order Trotterization (Thm.~\ref{thm:first_order_alpha}), we derive domain conditions and bounds for a slower fractional scaling.
\begin{thm}\label{thm:second_order_alpha}
Let $H_1$ be self-adjoint on $\mathcal{D}(H_1)$ and $H_2$ be self-adjoint on $\mathcal{D}(H_2)$.
Let $\varphi$ be an eigenstate of $H_1+H_2$ with eigenvalue $h$, i.e.\ $(H_1+H_2)\varphi=h\varphi$. 
Let the spectral measures $\mu_{j,\varphi}$ of $H_j$ at $\varphi$ decay as
\begin{equation}
	\mu_{j,\varphi}(\{|\lambda|\geq \Lambda\}) = \mathcal{O}\!\left(\frac{1}{\Lambda^{2\delta}}\right),\quad  j=1,2,
	\label{eq:tail_bound_2}
\end{equation}
for $\Lambda>0$ and some $\delta\in(1,2]$.
Then,
\begin{equation*}
\xi_N^{(2)}(t;\varphi)
=\begin{cases}
\medskip
\displaystyle
\mathcal{O}\!\left(\frac{t^\delta}{N^{\delta-1}}\right),
&\delta\in(1,2),\\
\displaystyle
\mathcal{O}\!\left(\frac{t^2}{N}\sqrt{\log (N/t)}\right), &\delta=2.
\end{cases}
\end{equation*}
Furthermore, if the decay~(\ref{eq:tail_bound_2}) holds for $\delta\geq 2$, $\varphi\in \mathcal{D}(H_1^2)$, and
the spectral measure $\mu_{2,H_1^2\varphi}$ of  $H_2$ at the vector $H_1^2\varphi$ decays as
\begin{equation*}
	\mu_{2,H_1^2\varphi}(\{|\lambda|\geq \Lambda\})=\mathcal{O}\!\left(\frac{1}{\Lambda^{2(\delta-2)}}\right),
\end{equation*}
for  $\Lambda>0$, then,
\begin{equation*}
\xi_N^{(2)}(t;\varphi)=\begin{cases}
\medskip
\displaystyle
\mathcal{O}\!\left(\frac{t^\delta}{N^{\delta-1}}\right),
&\delta\in[2,3),\\
\medskip
\displaystyle
\mathcal{O}\!\left(\frac{t^3}{N^2}\sqrt{\log (N/t)}\right), &\delta=3,\\
\displaystyle
\mathcal{O}\!\left(\frac{t^3}{N^2}\right), &\delta>3.
\end{cases}
\end{equation*}
\end{thm}

\begin{proof}
The proof is similar to the proof of Thm.~\ref{thm:first_order_alpha}\@.
The case $\delta>3$ is covered by Thm.~\ref{thm:2nd_order}, since in this case we are sure that $\varphi\in\mathcal{D}(H_1^3)\cap\mathcal{D}(H_2H_1^2)\cap\mathcal{D}(H_2^3)$ and the domain condition for Thm.~\ref{thm:2nd_order} is fulfilled.
The remaining nontrivial parts of the theorem are proved by the method described in  Appendix~\ref{sec:higher_order_derivation}\@.
\end{proof}

We now apply our second-order Trotter bounds to the problem of simulating the hydrogen atom.

\section{Application of the second-order bounds to the hydrogen atom}
\label{sec:second_order_hydrogen}
Applied to the eigenfunctions $\Psi_{n\ell m}$ of the hydrogen atom, Thm.~\ref{thm:second_order_alpha} results in scalings slower than $\mathcal{O}(N^{-2})$ for all $\ell\leq 3$.
As in the case of the first-order Trotterization, the scaling is again entirely determined by the quantum number $\ell$ of the orbital angular momentum.
We summarize our findings in Table~\ref{tab:scalings_2nd_order}\@.
\begin{table}
	\caption{\label{tab:scalings_2nd_order}Scalings of the analytical second-order Trotter error bounds for the  eigenfunctions $\Psi_{n\ell m}$ of the hydrogen atom.
	The fractional scalings are derived from Thm.~\ref{thm:second_order_alpha}\@.
	We find different scalings depending on the input states, which are solely determined by the quantum number $\ell$ of the orbital angular momentum.
	In particular, the scalings for all s- to f-orbitals with $\ell\le3$ are slower than the expected $N^{-2}$ scaling.
	Interestingly, the scalings for the f-orbitals slightly differ depending on whether we use the $H_1H_2H_1$ product formula $\mathcal{S}_N^{(2)}(t)$ [see the $\xi_N^{(2)}(t;\Psi_{n\ell m})$ column] or the $H_2H_1H_2$ product formula $\tilde{\mathcal{S}}_N^{(2)}(t)$ [see the $\tilde{\xi}_N^{(2)}(t;\Psi_{n\ell m})$ column].
	}
	\begin{center}
		\begin{tabular}{|c|c|c|}
		\hline
		\begin{tabular}[c]{@{}c@{}}  \bf{Orbital angular momen-} \\ \bf{tum quantum number} $\bm{\ell}$ \end{tabular} &
		\begin{tabular}[c]{@{}c@{}} \bf{Scaling of} \\ $\bm{\xi_N^{(2)}(t;\Psi_{n\ell m})}$ \end{tabular} &
		\begin{tabular}[c]{@{}c@{}} \bf{Scaling of} \\ $\bm{\tilde{\xi}_N^{(2)}(t;\Psi_{n\ell m})}$ \end{tabular}
		\\
			\hhline{:===:}
			$\ell = 0$ (s-orbitals) & $\mathcal{O}(N^{-1/4})$ & $\mathcal{O}(N^{-1/4})$\\
			$\ell = 1$ (p-orbitals) & $\mathcal{O}(N^{-3/4})$ & $\mathcal{O}(N^{-3/4})$\\
			$\ell = 2$ (d-orbitals)& $\mathcal{O}(N^{-5/4})$ & $\mathcal{O}(N^{-5/4})$\\
			$\ell = 3$ (f-orbitals)& $\mathcal{O}(N^{-3/2})$ & $\mathcal{O}(N^{-7/4})$\\
			$\ell \geq 4$ (g-orbitals and higher)& $\mathcal{O}(N^{-2})$ & $\mathcal{O}(N^{-2})$\\
			\hline
		\end{tabular}
	\end{center}
\end{table}

We numerically investigate the ground-state Trotter error in Fig.~\ref{fig:second_order}\@.
Interestingly, we observe the same $N^{-1/4}$ scaling behavior as in the case of the first-order Trotterization.
This shows that the second-order Trotterization is not advantageous in this case.
See also Fig.~\ref{fig:order_comparison} below.
\begin{figure}
\centering
	\begin{tikzpicture}[mark size={1.5}, scale=1]
			\begin{axis}[
			xmode=log,
			ymode=log,
			xlabel={Trotter steps $N$},
			ylabel=Second-order Trotter error $\xi_N^{(2)}(t;\Psi_{100})$,
			label style={font=\footnotesize},
			tick label style={font=\footnotesize},
			x post scale=1,
			y post scale=1,
			legend style={at={(0.5,-0.2)},anchor=north},
			legend columns=3,
			legend cell align={left},
			xmin=2, xmax=3000,
			]
			\addplot[color=1, mark=*, only marks] table[x=N, y=error, col sep=comma] {Trotter_Error_2nd_order_groundstate_100.csv};
			\addplot[color=2, mark=*, only marks] table[x=N, y=error, col sep=comma] {Trotter_Error_2nd_order_groundstate_200.csv};
			\addplot[color=3, mark=*, only marks] table[x=N, y=error, col sep=comma] {Trotter_Error_2nd_order_groundstate_300.csv};
			\addplot[color=4, mark=*, only marks] table[x=N, y=error, col sep=comma] {Trotter_Error_2nd_order_groundstate_400.csv};
			\addplot[color=5, mark=*, only marks] table[x=N, y=error, col sep=comma] {Trotter_Error_2nd_order_groundstate_800.csv};
			\addplot[color=0, ultra thick, densely dashed, domain=2:3000]{0.15*x^(-2)};
			\addplot[color=6, ultra thick, domain=2:3000]{0.5*(2.04665*x^(-1/2) + 0.832107/x + 1.34365/x^(1/4))};
			\legend{
			\footnotesize $100$ modes,
			\footnotesize $200$ modes,
			\footnotesize $300$ modes,
			\footnotesize $400$ modes,
			\footnotesize $800$ modes,
			\footnotesize $N^{-2}$,
			\footnotesize analytic bound,
			};
			\end{axis}
		\end{tikzpicture}
		\caption{\label{fig:second_order}The error of the second-order Trotter product formula at time $t=1$ in the Hartree atomic units $\hbar=m_\mathrm{e}=a_0=1$, as a function of the Trotter steps for the ground state $\Psi_{100}$ of the hydrogen atom. See Eq.~\eqref{eq:second_order_error}. The radial cutoff in the simulations is $R=30$. We show five different levels of discretization characterized by the number of Bessel modes. For reference, we show the slopes of $N^{-2}$ (grey dashed line) and our analytic bound (brown solid line), which scales as $N^{-1/4}$. Notice that the second-order bound just computes to half the first-order bound. We see that the asymptotics for any finite discretization initially start in alignment with our bound, and eventually go as $N^{-2}$, but $N$ at which this transition happens becomes larger with increasing number of modes. This provides an evidence that the scaling in the infinite-mode limit is indeed slower than $N^{-2}$. In particular, we expect no improvement in the asymptotic scaling from the second-order Trotterization over the first-order Trotterization, see Fig.~\ref{figNumericsGroundState}\@.}
\end{figure}
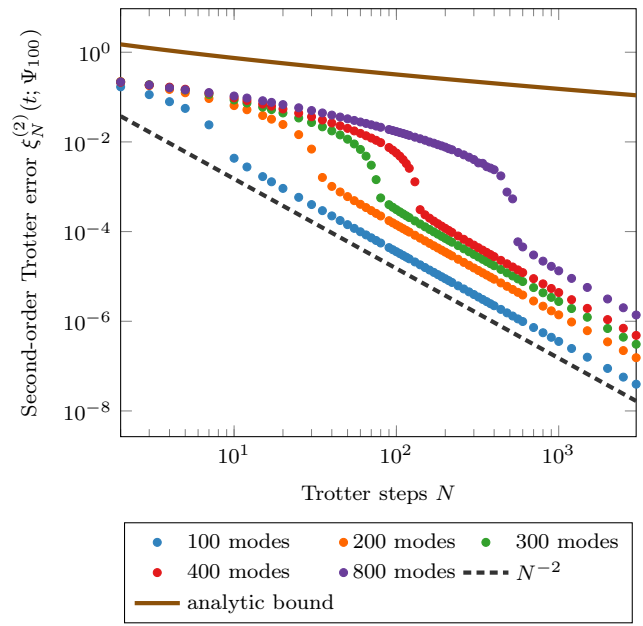

An interesting observation about the bounds in Thm.~\ref{thm:2nd_order} is that they are asymmetric under the exchange of $H_1$ and $H_2$.
This means that also the domain conditions on the state $\varphi$ change if we implement this change in the second-order Trotter product formula~\eqref{eq:second_order}.
Indeed, we find an asymmetry in the scaling for the eigenstate $\Psi_{430}$ of the hydrogen atom, see Table~\ref{tab:scalings_2nd_order}\@.
However, we do not see this effect numerically and both product formulas seem to admit the same scaling.
We simulated the Trotter error for $\Psi_{430}$ for the $H_1H_2H_1$ scheme $\mathcal{S}_N^{(2)}(t)$ and performed the same simulation for the $H_2H_1H_2$ scheme $\tilde{\mathcal{S}}_N^{(2)}(t)$.
The results are shown in Fig.~\ref{fig:second_order_n4_l3}.
Both curves are indistinguishable; nevertheless, $\Psi_{430}$ exhibits a scaling which is slower than $N^{-2}$.
However, we conjecture that there exist quantum systems, which indeed admit an asymmetric scaling for certain input states.

\begin{figure}
\centering
	\begin{tikzpicture}[mark size={1.5}, scale=1]
			\begin{axis}[
			xmode=log,
			ymode=log,
			xlabel={Trotter steps $N$},
			ylabel=Second-order Trotter error $\mathord{\mathop{\xi}\limits^{\scriptscriptstyle(\sim)}}_N^{(2)}(t;\Psi_{430})$,
			label style={font=\footnotesize},
			tick label style={font=\footnotesize},
			x post scale=1,
			y post scale=1,
			legend style={at={(0.5,-0.2)},anchor=north},
			legend columns=2,
			legend cell align={left},
			xmin=2, xmax=3000,
			]
			\addplot[color=1, mark=*, only marks] table[x=N, y=error, col sep=comma] {Trotter_Error_2nd_order_ABA_n4_l3_200.csv};
			\addplot[color=1, mark=x, only marks, mark size=4] table[x=N, y=error, col sep=comma] {Trotter_Error_2nd_order_BAB_n4_l3_200.csv};
			\addplot[color=2, mark=*, only marks] table[x=N, y=error, col sep=comma] {Trotter_Error_2nd_order_ABA_n4_l3_400.csv};
			\addplot[color=2, mark=x, only marks, mark size=4] table[x=N, y=error, col sep=comma] {Trotter_Error_2nd_order_BAB_n4_l3_400.csv};
			\addplot[color=3, mark=*, only marks] table[x=N, y=error, col sep=comma] {Trotter_Error_2nd_order_ABA_n4_l3_800.csv};
			\addplot[color=3, mark=x, only marks, mark size=4] table[x=N, y=error, col sep=comma] {Trotter_Error_2nd_order_BAB_n4_l3_800.csv};
			\addplot[color=4, mark=*, only marks] table[x=N, y=error, col sep=comma] {Trotter_Error_2nd_order_ABA_n4_l3_1200.csv};
			\addplot[color=4, mark=x, only marks, mark size=4] table[x=N, y=error, col sep=comma] {Trotter_Error_2nd_order_BAB_n4_l3_1200.csv};
			\addplot[color=0, mark=*, ultra thick, mark options={line width=0.4pt}] table[x=N, y=error, col sep=comma] {Scaling_n4_l3_Three_Half.csv};
			\addplot[color=0, mark=x, ultra thick, mark options={line width=0.4pt}, mark size=4] table[x=N, y=error, col sep=comma] {Scaling_n4_l3_Seven_Quarter.csv};
			\addplot[color=0, ultra thick, domain=2:3000, densely dashed]{0.01*x^(-2)};
			\legend{
			\footnotesize $200$ modes $\mathcal{S}_N^{(2)}$,
			\footnotesize $200$ modes $\tilde{\mathcal{S}}_N^{(2)}$,
			\footnotesize $400$ modes $\mathcal{S}_N^{(2)}$,
			\footnotesize $400$ modes $\tilde{\mathcal{S}}_N^{(2)}$,
			\footnotesize $800$ modes $\mathcal{S}_N^{(2)}$,
			\footnotesize $800$ modes $\tilde{\mathcal{S}}_N^{(2)}$,
			\footnotesize $1200$ modes $\mathcal{S}_N^{(2)}$,
			\footnotesize $1200$ modes $\tilde{\mathcal{S}}_N^{(2)}$,
			\footnotesize $N^{-3/2}$,
			\footnotesize $N^{-7/4}$,
			\footnotesize $N^{-2}$
			};
			\end{axis}
		\end{tikzpicture}
		\caption{\label{fig:second_order_n4_l3}The error of the second-order Trotter product formula at time $t=1$ in the Hartree atomic units $\hbar=m_\mathrm{e}=a_0=1$, as a function of the Trotter steps for the eigenstate $\Psi_{430}$ of the hydrogen atom.
		The radial cutoff in the simulations is $R=20$.
		We show four different levels of discretization characterized by the number of Bessel modes.
		For reference, we show the slope of $N^{-2}$ (grey dashed line).
		The results of the $H_1H_2H_1$ scheme $\mathcal{S}_N^{(2)}(t)$ of the second-order Trotterization in Eq.~\eqref{eq:second_order} are shown by dots $\bullet$.
		The results of the $H_2H_1H_2$ scheme $\tilde{\mathcal{S}}_N^{(2)}(t)$ of the second-order Trotterization in Eq.~\eqref{eq:second_order_reversed} are shown by crosses $\times$.
		Both schemes show indistinguishable results.
		However, we see again that the asymptotics for any finite discretization initially start slower, and eventually go as $N^{-2}$, but $N$ at which this transition happens becomes larger with increasing number of modes.
		This provides an evidence that the scaling in the infinite-mode limit is indeed slower than $N^{-2}$.
		For reference, we show the predicted scaling of $N^{-3/2}$ for $\mathcal{S}_N^{(2)}(t)$ (grey line with dots) and the predicted scaling of $N^{-7/4}$ for $\tilde{\mathcal{S}}_N^{(2)}(t)$ (grey line with crosses).
		}
\end{figure}
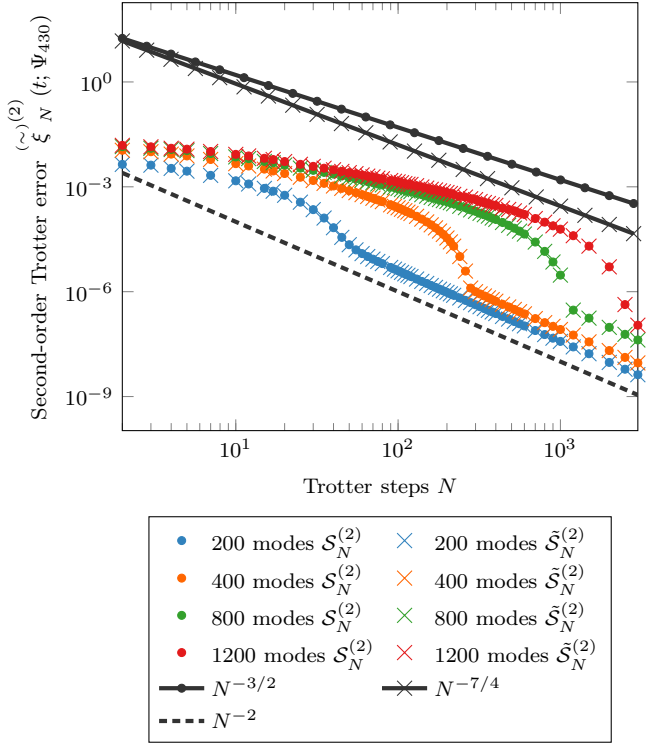

From the previous discussion, it seems like higher-order Trotter schemes are more restricted in terms of their domain conditions than the lower-order schemes.
To investigate this further, we develop a method to obtain state-dependent higher-order Trotter bounds for arbitrary order $p$.
This is done by generalizing the approach for the second-order results.
In the next section, we explain how to derive such bounds and explicitly study the example of a fourth-order Trotter product formula.

\section{Higher-order Trotterization}
\label{sec:higher_order}
A common practice in quantum simulation is to employ higher-order Trotter schemes~\cite{Childs2021}.
These are obtained by symmetrizing the Trotter product, which results in the cancellation of low-order terms in the Taylor expansion of the exponentials~\cite{Morales2022}.
For two \emph{bounded} Hamiltonians $H_1$ and $H_2$, the error of the $p$th-order Trotter product formula $\mathcal{S}^{(p)}(t)$ scales as $\mathcal{O}(t^{p+1}/N^p)$~\cite{Childs2021}.
Schemes with a higher order $p$ might have many different solutions for the short switching times between $H_1$ and $H_2$~\cite{Morales2022}.
An example of fourth-order scheme ($p=4$) is~\cite{Suzuki1991}
\begin{align}
\mathcal{S}_N^{(4)}(t)=\Bigl(&\rme^{-\rmi\frac{t}{N}\frac{\tau}{2}H_1}\rme^{-\rmi \frac{t}{N}\tau H_2}\rme^{-\rmi \frac{t}{N}\frac{1-\tau}{2}H_1}\rme^{-\rmi \frac{t}{N}(1-2\tau)H_2}\nonumber\\
		&\qquad\qquad\
		{}\times \rme^{-\rmi\frac{t}{N}\frac{1-\tau}{2}H_1}\rme^{-\rmi\frac{t}{N}\tau H_2}\rme^{-\rmi\frac{t}{N}\frac{\tau}{2}H_1}\Bigr)^N,
		\label{eq:fourth_order}
\end{align}
where $\tau=1/(2-2^{1/3})$.
If the two Hamiltonians $H_1$ and $H_2$ are bounded, the Trotter error scales as
\begin{equation*}
	\Vert\mathcal{S}_N^{(4)}(t) - \rme^{-\rmi t(H_1+H_2)}\Vert = \mathcal{O}\!\left(\frac{t^5}{N^4}\right).
\end{equation*}
Notice that the better scalings of the higher-order schemes in the number of Trotter steps $N$ simultaneously go along with having to implement more unitaries per Trotter cycle.
A recent study~\cite{Morales2022} suggests that an 8th-order expansion could lead to optimal results requiring minimal resources.

Nevertheless, the results mentioned above are usually not applicable to unbounded operators, see e.g.\ Refs.~\cite{Ichinose2004, Burgarth2023}.
As before, we need to perform a state-dependent analysis in this case.
In Sec.~\ref{sec:methods_higher_order} of the Methods section, and in Appendix~\ref{sec:higher_order_derivation} in more detail, we present a method to obtain such state-dependent bounds systematically for general $p$th-order product formulas through repeated integration by parts.
This method is inspired by Refs.~\cite{Burgarth2022, Burgarth2023}.
By generalizing these ideas, we are able to carry them over to unbounded operators provided the input state satisfies certain domain conditions.

In general, a $p$th-order Trotter product is given in the form
\begin{equation}
\mathcal{S}_N^{(p)}(t)=\left(\rme^{-\rmi\frac{t}{N}\tau_MH_{\sigma(M)}}\cdots \rme^{-\rmi\frac{t}{N}\tau_2H_{\sigma(2)}}\rme^{-\rmi\frac{t}{N}\tau_1H_{\sigma(1)}}\right)^N,
\label{eq:higher_order_evolution}
\end{equation}
where
\begin{equation}
	\sigma(j)=\begin{cases}
	\medskip
	\displaystyle
		1,&\text{if } j \text{ is odd},\\
	\displaystyle
		2,&\text{if } j \text{ is even},
	\end{cases}
\label{eq:parity}
\end{equation}
and $\sum_{i=1}^{\lceil M/2\rceil}\tau_{2i-1}=\sum_{i=1}^{\lfloor M/2\rfloor}\tau_{2i}=1$, where $\lceil x\rceil$ is the ceiling and $\lfloor x\rfloor$ is the floor function.
This product is to be compared with $\rme^{-\rmi t(H_1+H_2)}$.
We then define the $p$th-order state-dependent Trotter error by
\begin{equation}
	\xi_N^{(p)}(t;\varphi)=\Vert[\mathcal{S}_N^{(p)}(t)-\rme^{-\rmi th}]\varphi\Vert,
	\label{eq:higher_order_error}
\end{equation}
where again $(H_1+H_2)\varphi=h\varphi$.

Note that some $\tau_j$ in the product formula~(\ref{eq:higher_order_evolution}) can be negative.
Indeed, in the fourth-order Trotter product in Eq.~(\ref{eq:fourth_order}), $\tau_3=\tau_5=1-\tau$ and $\tau_4=1-2\tau$ are negative.
In fact, it has been shown by Suzuki that product formulas with $p>2$ necessarily have to involve some negative $\tau_j$ \cite[Thm.~3]{Suzuki1991}\@.
Many examples of $p$th-order product formulas can be found in the literature, see e.g.\ Refs.~\cite{Suzuki1990, Yoshida1990,Suzuki1991, Morales2022, Ostmeyer2023}\@.
For instance, in the case of Suzuki's first fractal method for generating higher-order product formulas~\cite{Suzuki1990, Suzuki1991, Morales2022}, the number of exponentials involved in the product formula of an even order $p$ is given by $M=2\cdot3^{p/2-1}+1$~\cite{Morales2022}.
In any case, our method for bounding $\xi_N^{(p)}(t;\varphi)$ works for any $p$th-order product formula.
Although our method allows us to set up and solve the equations for the switching times between $H_1$ and $H_2$ in a $p$th-order Trotter product, we assume that these are known here.

A Mathematica script that explicitly computes the Trotter bounds according to our procedure for given switching times is outlined in Appendix~\ref{sec:Mathematica}\@.
This method produces bounds with small prefactors that can be used for practical state-dependent Hamiltonian-simulation purposes.
Nevertheless, we here provide a loose state-dependent bound, which only aims to show the scaling.
\begin{thm}\label{thm:loose_bound_higher_order}
Let $H_1$ be self-adjoint on $\mathcal{D}(H_1)$ and $H_2$ be self-adjoint on $\mathcal{D}(H_2)$.
Let $\varphi$ be an eigenstate of $H_1+H_2$ with eigenvalue $h$, i.e.\ $(H_1+H_2)\varphi=h\varphi$.
For $g\in\mathbb{R}$, define ${H}_1(g)=H_1-g$ and ${H}_2(g)=H_2-h+g$.
Then, the state-dependent error of a $p$th-order Trotter product $\mathcal{S}_N^{(p)}(t)$, consisting of $M$ exponentials of $H_1$ and $H_2$ per Trotter cycle,
can be bounded for all $g\in\mathbb{R}$ and $t\geq 0$ by
\begin{equation*}
\xi_N^{(p)}(t;\varphi)\leq\frac{1}{(p+1)!}\frac{(\tau_{*}t)^{p+1}}{N^p}K_{p+1}(\varphi),
\end{equation*}
where $\tau_*=\sum_{j=1}^M|\tau_j|$ and
$$
{K}_p(\varphi)=\max_{1\leq j_1\leq\dots\leq j_p\leq M}\|H_{\sigma(j_p)}(g)\cdots H_{\sigma(j_1)}(g)\varphi\|,
$$
provided that $\varphi$ is in the domains of all the involved operator products.
\end{thm}
\begin{proof}
This is proved in Appendix~\ref{sec:higher_order_derivation}\@.
\end{proof}
The bound in Thm.~\ref{thm:loose_bound_higher_order} admits the desired scaling of
\begin{equation*}
	\xi_N^{(p)}(t;\varphi)=\mathcal{O}\!\left(\frac{t^{p+1}}{N^p}\right),
\end{equation*}
and is directly applicable to any $p$th-order Trotter scheme.
However, as pointed out before, it can be very loose.
To obtain tighter bounds, we follow the method outlined in Sec.~\ref{sec:methods_higher_order}\@, which is a generalization of the strategy we developed to obtain the second-order bounds in Sec.~\ref{sec:second_order}\@.
For the first-order Trotterization, this method leads to the same result as Thm.~\ref{thm:trotter_thm}\@.
Applied to the fourth-order Trotter product presented in Eq.~\eqref{eq:fourth_order}, it results in the following theorem.
\begin{thm}\label{thm:4th_order}
Let $H_1$ and $H_2$ be self-adjoint on their respective domains $\mathcal{D}(H_1)$ and $\mathcal{D}(H_2)$, and let $\varphi$ be an eigenstate of $H_1+H_2$ with eigenvalue $h$, i.e.\ $(H_1+H_2)\varphi=h\varphi$.
For $g\in\mathbb{R}$, define ${H}_1(g)=H_1-g$ and ${H}_2(g)=H_2-h+g$.
Then, the fourth-order state-dependent Trotter product formula $\mathcal{S}_N^{(4)}(t)$ in Eq.~(\ref{eq:fourth_order}) is bounded by
\begingroup
\allowdisplaybreaks
\begin{align}
	\xi_N^{(4)}(t;\varphi)\leq\frac{t^5}{N^4}\,\biggl(&
	\frac{a_0}{{17\,280}}\Vert{H}_1(g)^5 \varphi\Vert
		\nonumber\\
		&{}+ \frac{a_1}{1\,152}\Vert {H}_1(g)^2{H}_2(g){H}_1(g)^2\varphi\Vert
		\nonumber\\
		&{}+ \frac{a_1}{1\,728}\Vert {H}_1^2(g){H}_2(g)^3\varphi\Vert
		\nonumber\\
		&{}+ \frac{a_2}{1\,728} \Vert {H}_1(g){H}_2(g){H}_1(g)^3\varphi\Vert
		\nonumber\\
		&{}+ \frac{a_2}{576} \Vert {H}_1(g){H}_2(g)^2{H}_1(g)^2\varphi\Vert
		\nonumber\\
		&{}+ \frac{a_2}{864} \Vert {H}_1(g){H}_2(g)^4\varphi\Vert
		\nonumber\\
		&{}+ \frac{a_3}{6\,912} \Vert {H}_2(g){H}_1(g)^4\varphi\Vert
		\nonumber\\
		&{}+ \frac{a_4}{576} \Vert {H}_2(g){H}_1(g){H}_2(g){H}_1(g)^2\varphi\Vert
		\nonumber\\
		&{}+ \frac{a_4}{864} \Vert {H}_2(g){H}_1(g){H}_2(g)^3\varphi\Vert
		\nonumber\\
		&{}+ \frac{a_5}{48} \Vert {H}_2(g)^2{H}_1(g)^3\varphi \Vert
		\nonumber\\
		&{}+ \frac{a_6}{1\,728} \Vert {H}_2(g)^3{H}_1(g)^2\varphi \Vert
		\nonumber\\
		&{}+ \frac{a_7}{4\,320} \Vert {H}_2(g)^5\varphi\Vert
		\biggr),
		\label{eq:bound_4th_order}
\end{align}
where
\begin{align*}
	a_0&=23+19\cdot2^{1/3}+17\cdot 2^{2/3},\\
	a_1&=4+3\cdot2^{1/3}+2\cdot 2^{2/3},\\
	a_2&=14+11\cdot2^{1/3}+9\cdot 2^{2/3},\\
	a_3&=68+55\cdot2^{1/3}+44\cdot 2^{2/3},\\
	a_4&=18+14\cdot2^{1/3}+11\cdot 2^{2/3},\\
	a_5&=\frac{1+2^{1/3}}{(2-2^{1/3})^5},\\
	a_6&= 226+180\cdot2^{1/3}+143\cdot 2^{2/3},\\
	a_7&= 330+263\cdot2^{1/3}+209\cdot 2^{2/3},
\end{align*}
\endgroup
for all $t,g\in\mathbb{R}$, if $\varphi$ is in the domains of all the involved operator products in Eq.~(\ref{eq:bound_4th_order})\@.
\end{thm}
\begin{proof}
	This is computed by using the Mathematica script provided in Appendix~\ref{sec:Mathematica}\@.
	It follows exactly the method outlined in Sec.~\ref{sec:methods_higher_order} of the Methods section.
	For the numerical values of the coefficients, see Appendix~\ref{sec:Mathematica}\@.
\end{proof}

We also computed error bounds for a $6$th-order product formula numerically with our Mathematica script.
See Appendix~\ref{sec:Mathematica}\@.

\begin{figure}
\centering
	\begin{tikzpicture}[mark size={1.5}, scale=1]
			\begin{axis}[
			xmode=log,
			ymode=log,
			xlabel={Trotter steps $N$},
			ylabel=Fourth-order Trotter error $\xi_N^{(4)}(t;\Psi_{100})$,
			label style={font=\footnotesize},
			tick label style={font=\footnotesize},
			x post scale=1,
			y post scale=1,
			legend pos=south west,
			legend columns=1,
			legend cell align={left},
			xmin=2, xmax=2500,
			]
			\addplot[color=1, mark=*, only marks] table[x=N, y=error, col sep=comma] {Trotter_Error_4th_order_groundstate_100.csv};
			\addplot[color=2, thick, mark=*, only marks] table[x=N, y=error, col sep=comma] {Trotter_Error_4th_order_groundstate_200.csv};
			\addplot[color=3, mark=*, only marks] table[x=N, y=error, col sep=comma] {Trotter_Error_4th_order_groundstate_300.csv};
			\addplot[color=4, mark=*, only marks] table[x=N, y=error, col sep=comma] {Trotter_Error_4th_order_groundstate_400.csv};
			\addplot[color=5, mark=*, only marks] table[x=N, y=error, col sep=comma]{Trotter_Error_4th_order_groundstate_800.csv};
			\addplot[color=0, ultra thick, domain=2:2500]{2*x^(-1/4)};
			\addplot[color=0, ultra thick, domain=2:2500, densely dashed]{2*x^(-4)};
			\legend{
			\footnotesize $100$ modes,
			\footnotesize $200$ modes,
			\footnotesize $300$ modes,
			\footnotesize $400$ modes,
			\footnotesize $800$ modes,
			\footnotesize $N^{-1/4}$,
			\footnotesize $N^{-4}$,
			};
			\end{axis}
		\end{tikzpicture}
		\caption{\label{fig:fourth_order}The error of the fourth-order Trotter product formula at time $t=1$ in the Hartree atomic units $\hbar=m_\mathrm{e}=a_0=1$, as a function of the Trotter steps for the ground state $\Psi_{100}$ of the hydrogen atom.
		See Eq.~\eqref{eq:higher_order_error}.
		The radial cutoff in the simulations is $R=30$.
		We show five different levels of discretization characterized by the number of Bessel modes.
		For reference, we show the slopes for $N^{-4}$ (grey dashed line) and $N^{-1/4}$ (grey solid line).
		We see that the Trotter error for any finite discretization initially starts as $N^{-1/4}$, and eventually goes as $N^{-4}$, but $N$ at which this transition happens becomes larger with increasing number of modes.
		This provides an evidence that the scaling in the infinite-mode limit is indeed slower than $N^{-4}$.
		In particular, we expect no improvement in the asymptotic scaling from the fourth-order Trotterization over the first- or second-order Trotterization, see Figs.~\ref{figNumericsGroundState} and~\ref{fig:second_order}\@.
		}
\end{figure}
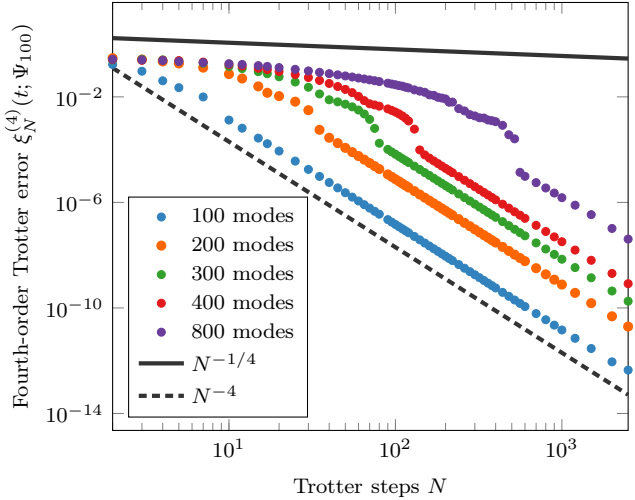

As  expected, the domain conditions in Thm.~\ref{thm:4th_order} on the input state $\varphi$ for the fourth-order Trotter product are even stronger than the conditions in Thm.~\ref{thm:2nd_order} for the second-order Trotterization.
Again, the ground state of the hydrogen atom does not satisfy these conditions and a slower Trotter scaling is expected.
We numerically investigate the fourth-order Trotterization in Fig.~\ref{fig:fourth_order}\@.
As before, we observe an error scaling of $\mathcal{O}(N^{-1/4})$ in the full system limit, which coincides with the first- and second-order Trotterizations.
In summary, we find
\begin{equation*}
	\xi^{(p)}_N(t;\Psi_{100})
	=\mathcal{O}(N^{-1/4}), \quad p=1,2,4.
\end{equation*}
This is a strong indication that the entire higher-order Trotter hierarchy breaks down for the ground state $\Psi_{100}$ of the hydrogen atom.
Indeed, we show in Appendix~\ref{sec:higher_order_derivation} how to obtain error bounds that admit fractional scalings for any $p$th-order product formulas.
This is a generalization of Thm.~\ref{thm:first_order_alpha} and Thm.~\ref{thm:second_order_alpha}\@.
We summarize the implications of our results as follows.
Assume that we are given a $p$th-order product formula $\mathcal{S}_N^{(p)}(t)$ and an input state $\varphi$, which is an eigenstate of $H_1+H_2$.
On one hand, if $\varphi$ satisfies the domain conditions of the product formula $\mathcal{S}_N^{(q)}(t)$ with $q\leq p$, then the Trotter error scales at least as $\xi_N^{(p)}(t;\varphi)=\mathcal{O}(N^{-q})$.
On the other hand, if $\varphi$ does \emph{not} satisfy the domain conditions of the product formula $\mathcal{S}_N^{(r)}(t)$ with $r\leq p$, the Trotter error $\xi_N^{(p)}(t;\varphi)$ scales \emph{slower} than $\mathcal{O}(N^{-r})$.
In total,
\begin{equation*}
\xi_N^{(p)}(t;\varphi) = \mathcal{O}(N^{-\delta}),\quad q\leq\delta<r.
\end{equation*}
Appendix~\ref{sec:higher_order_derivation} provides the proof for this statement including explicit error bounds for $q=1,2$, but the described procedure can be easily iterated to obtain error bounds for $q\ge3$.
Furthermore, the reasoning from Cor.~\ref{cor:superposition_states} to treat superpositions of eigenstates can be applied to $p$th-order product formulas as well.

Since the ground state $\Psi_{100}$ of the hydrogen atom does not even satisfy the domain conditions for the first-order Trotterization ($q=0$), any higher-order method only admits error bounds with $N^{-1/4}$ scaling.
From our numerical investigations, it seems like these error bounds reflect the true scaling.
See Fig.~\ref{fig:fourth_order}\@.
That is, the Trotter error for the ground state $\Psi_{100}$ of the hydrogen atom always scales the same and slower than $\mathcal{O}(N^{-1})$ irrespectively of the Trotter order $p$.
This makes the first-order Trotterization the most favorable one for quantum chemistry simulations, as it involves the smallest number of unitaries per Trotter cycle.
A comparison of the first-, second-, and fourth-order Trotter errors as functions of the total number of unitaries to implement can be found in Fig.~\ref{fig:order_comparison}\@.
Since the evolutions of the first- and second-order Trotterizations only differ by two boundary unitaries, they admit similar performances.
The fourth-order product formula has a worse error-to-resource ratio though.
See the white region of Fig.~\ref{fig:order_comparison}\@.
\begin{figure}
	\centering
	\begin{tikzpicture}[mark size={1.5}, scale=1]
		\begin{axis}[
			xmode=log,
			ymode=log,
			xlabel={Total number of unitaries},
			ylabel=Ground-state Trotter error $\xi_N^{(p)}(t;\Psi_{100})$,
			label style={font=\footnotesize},
			tick label style={font=\footnotesize},
			x post scale=1,
			y post scale=1,
			legend pos=south west,
			legend columns=1,
			legend cell align={left},
			xmin=4, xmax=18001,
			enlarge y limits=false,
			]
			\addplot[color=1, mark=*, only marks] table[x=N, y=error, col sep=comma] {Order_Comparison_1st_order_groundstate_800.csv};
			\addplot[color=2, thick, mark=*, only marks] table[x=N, y=error, col sep=comma] {Order_Comparison_2nd_order_groundstate_800.csv};
			\addplot[color=3, mark=*, only marks] table[x=N, y=error, col sep=comma] {Order_Comparison_4th_order_groundstate_800.csv};
			\addplot[ultra thick, samples=50, smooth, color=0, dashed] coordinates {(881,0.00000001) (881,1)};
			\addplot[fill=lightgray,draw=none] coordinates {(881,0.00000001) (18001,0.00000001) (18001,1) (881,1)} \closedcycle;
			\legend{
				\footnotesize First-order ($p=1$),
				\footnotesize Second-order ($p=2$),
				\footnotesize Fourth-order ($p=4$),
			};
		\end{axis}
	\end{tikzpicture}
	\caption{\label{fig:order_comparison}Comparison of the Trotter errors for different orders $p$ of Trotterization for the ground state $\Psi_{100}$ of the hydrogen atom.
		For the simulation, we fixed the total evolution time to $t=1$ and worked in the Hartree atomic units $\hbar=m_\mathrm{e}=a_0=1$.
		Furthermore, we truncated the potential and kinetic energies at $800$ Bessel modes and chose a radial cutoff of $R=30$.
		We then compare the Trotter errors~\eqref{eq:higher_order_error} for the first-order ($p=1$, blue dots), second-order ($p=2$, orange dots), and fourth-order ($p=4$, green dots) Trotterizations as functions of the total number of unitaries in the Trotter product.
		If boundary unitaries of two adjacent Trotter cycles are generated by the same Hamiltonian, they are combined and only counted as one unitary.
		In the regime where all Trotter errors scale as $N^{-1/4}$ (white region), the first- and second-order Trotterizations perform almost the same and are both better than the fourth-order product formula.
		Once the total number of unitaries exceeds a threshold (dashed line), the Trotter errors start to scale faster.
		In this region (light gray), product formulas with higher $p$ will eventually outperform product formulas with lower $p$.
		However, this behavior does not reflect the true error scaling, due to finite-size effects.
	}
\end{figure}
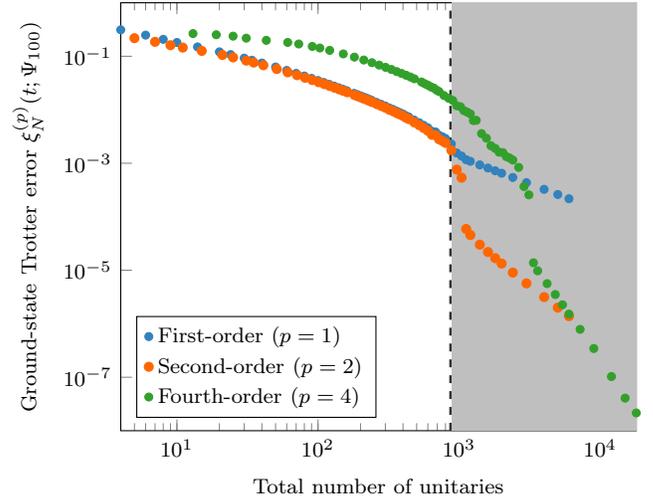

\section{Methods}
\label{sec:methods}
In this section, we describe the methods we used to obtain our results.
This includes both details on the numerical simulations and the derivations of the analytical error bounds.
The technical details and lengthy calculations can be found in the appendices.

\subsection{Numerical methods for the grid-based hydrogen simulations}
\label{sec:methods_numerical}
To perform the grid-based simulation described in Sec.~\ref{sec:bound}, we numerically solve the Schr{\"{o}}dinger equation
\begin{equation}
\rmi  \frac{\rmd \psi({\mathbf{r}},t)} {\rmd t} = H(t) \psi({\mathbf{r}},t),  \quad \psi({\mathbf{r}}, 0) = \Psi_{n\ell m}(\mathbf{r})
\label{eqSEhydrogenNaturalUnits}
\end{equation}
over some time $t$, with $H(t)$ being alternately taken as $H_1$ and $H_2$ for the Trotter step evolution, or $H_1 + H_2$ for the actual evolution.
The numerical integration of Eq.~\eqref{eqSEhydrogenNaturalUnits} was carried out using XMDS2~\cite{Dennis2013}, an open-source package for solving multidimensional partial differential equations.

The time-dependent wave functions are given by $\psi_{\text{Trot}}({\mathbf{r}}, t) = U_N(t) \psi({\mathbf{r}}, 0) $ and $\psi({\mathbf{r}}, t) = V(t) \psi({\mathbf{r}}, 0) $, with the initial state chosen to be an eigenfunction $\psi({\mathbf{r}}, 0) = \Psi_{n\ell m}(\mathbf{r})$ of the hydrogen atom.
Here, $U_N(t)$ is the dynamics given by the Trotter product and $V(t)$ is the target dynamics under $H_1+H_2$.
In terms of these wave functions, the Trotter error~\eqref{eq:definition_standard_error} at time $t$ becomes
\begin{equation}
\xi_N(t;\Psi_{n\ell m}) = \left( \int_{\mathbb{R}^3}  |\psi_{\text{Trot}}({\mathbf{r}}, t) - \psi({\mathbf{r}},t)|^2 \, \rmd^3{\mathbf{r}} \right)^{1/2}. 
\label{eq:schroedinger_integral}
\end{equation}

As is well known, Eq.~\eqref{eqSEhydrogenNaturalUnits} is separable for the hydrogen atom, and since all the terms in the Hamiltonian are spherically symmetric, the evolution conserves the angular momentum, meaning only the differential equation for the radial coordinate is relevant. 
In spherical coordinates, the Laplacian (corresponding to $H_1$) acting on an eigenstate $\Psi_{n\ell m}$ is given by
\begin{equation*}
	\Delta = \frac{\rmd^2}{\rmd r^2} + \frac{2}{r} \frac{\rmd}{\rmd r} - \frac{\ell(\ell+1)}{r^2}.
\end{equation*} 
This form of the Laplacian allows the use of the spectral methods to numerically solve Eq.~\eqref{eqSEhydrogenNaturalUnits}, with the spherical Bessel functions as the spectral basis.
The number of Bessel functions used corresponds to the number of grid points in position space as well as the number of modes in the spectral space, resulting in a discretized system suitable for numerical integration.
The utility of the spherical Bessel modes is that, in the spectral basis, the action of the Laplacian is purely multiplicative.

The simulations were carried out on a radial domain $[0, R]$, where the radial cutoff $R$ was chosen to ensure that the wave function remained within the interval in both position and spectral space.
The integration of Eq.~\eqref{eq:schroedinger_integral} itself is performed exactly through quadrature integration and uses an adaptive fourth-fifth order Runge-Kutta integrator.
By taking the zeroes of the spherical Bessel functions as mesh points, this exact integration method avoids the occurrence of a spatial error.
The numerical stability was tested by ensuring that the norm was stable, the wave function remained far from the boundary in both position and spectral bases, and the results were robust to changing the relative error tolerance of the stepper.

\subsection{Bounds on the first-order Trotter convergence for the hydrogen atom}
\label{sec:trotter_thm}
To obtain the scalings presented in Table~\ref{tab:order_results} in Sec.~\ref{sec:bound}, we apply Thm.~\ref{thm:trotter_thm}, which poses conditions on the input eigenstate.
Since not all the eigenstates of the hydrogen atom satisfy these conditions, some particular states need a special treatment.
An example of such a state is the ground state of the hydrogen atom.

Before applying Thm.~\ref{thm:trotter_thm} to the hydrogen atom, we first notice that $H_\mathrm{hydrogen}$ is self-adjoint on $\mathcal{D}(H_{1})$ and $\mathcal{D}(H_{1})\subset\mathcal{D}(H_{2})$~\cite[Chap.~18.2]{Hall2013}.
By Trotter's theorem~\cite[Thm.~1.2]{Simon2005}, we thus know that the Trotter product formula converges, i.e.\ for all $\psi\in\mathcal{H}$,
\begin{equation*}
\lim_{N\rightarrow\infty}\left[\left(\rme^{-\rmi \frac{t}{N}H_2}\rme^{-\rmi \frac{t}{N}H_1}\right)^{N}-\rme^{-\rmi t\left(H_{1}+H_{2}\right)}\right]\psi=0.
\end{equation*}
Theorem~\ref{thm:trotter_thm} does not necessarily demand Trotter convergence on all input states.
It only requires an eigenstate of interest $\varphi$ to be in the domains of $H_1^2$ and $H_2^2$.
In the case of the eigenfunctions $\Psi_{n\ell m}$ of the hydrogen atom, we find this condition to be satisfied by all states with $\ell\geq2$.
However, for the kinetic energy $H_1$, we have that all the states with the quantum number $\ell=0,1$ of the orbital angular momentum do not satisfy this condition, i.e.\  $\Psi_{n00},\Psi_{n1m}\not\in\mathcal{D}(H_1^2)$.
Furthermore, it is broken by the states with $\ell=0$ in the case of the potential energy, $\Psi_{n00}\not\in\mathcal{D}(H_2^2)$.
We show in Appendix~\ref{sec:hydrogen_bounds} how to apply Thm.~\ref{thm:trotter_thm} to the Hamiltonian of the hydrogen atom and also derive bounds for $\Psi_{n00}$ and $\Psi_{n1m}$ separately.
This is done by applying Thm.~\ref{thm:first_order_alpha}, which only requires the state to be in the domains of $H_1$ and $H_2$ rather than those of $H_1^2$ and $H_2^2$, and is therefore suitable for treating $\Psi_{n00}$ and $\Psi_{n1m}$.

Succinctly, we obtain for the $\ell=0$ eigenstates (including the ground state $n=1$),
\begin{align}
	\xi_N(t;\Psi_{n00})
	\leq{}&
	\frac{\tilde{t}^{\,5/4}}{N^{1/4}}
	\frac{4}{\sqrt{5\sqrt{\pi}\,n^3}}
	+
	\frac{\tilde{t}^{\,3/2}}{N^{1/2}}
	\sqrt{
	\frac{4\pi}{3n^3}
	}\nonumber\\
	&{}+
	\frac{\tilde{t}^{\,2}}{2N}
	\left(
	\sqrt{\frac{2}{n^7}}
	+
	\frac{1}{4n^4}
	\right),\label{eq:bound_R_n0}
\end{align}
for the $\ell=1$ states,
\begin{align}
	\xi_N(t;\Psi_{n1m})
	\leq{}&
	\frac{\tilde{t}^{\,7/4}}{N^{3/4}}
	\sqrt{
	\frac{64(n^2-1)}{189\sqrt{\pi}\,n^{5}}
	}\nonumber\\
	&{}+
	\frac{\tilde{t}^{\,2}}{2N}
	\left(
	\sqrt{
	\frac{6-4/n^2}{15n^3}
	}
	+
	\sqrt{\frac{2}{3n^7}}
	+
	\frac{1}{4n^4}
	\right),
	\label{eq:bound_R_n1}
\end{align}
and finally for the $\ell\geq2$ states,
\begin{align}
	&\xi_N(t;\Psi_{n\ell m})\nonumber\\
	&\quad
	\leq
	\frac{\tilde{t}^{\,2}}{2N}
	\,\Biggl(
	\frac{1}{8}
	\sqrt{
	\frac{(n+\ell)!\beta(\ell-3/2,\ell+11/2)}{n^7(n-\ell-1)![\Gamma(\ell+3/2)]^2}
	}
	\nonumber\\
	&\qquad\qquad\quad
	{}+
	\sqrt{
	\frac{3-\ell(\ell+1)/n^2}{2n^3\ell(\ell+1)(\ell+1/2)(\ell+3/2)(\ell-1/2)}
	}
	\nonumber\\
	&\qquad\qquad\quad
	{}+
	\frac{1}{\sqrt{n^{7}(\ell+1/2)}}
	+
	\frac{1}{4n^4}
	\Biggr),
	\label{eq:bound_R_nl}
\end{align}
where 
\begin{equation}
\tilde{t}=\frac{t}{t_0}, \quad \text{with} \quad t_0=\frac{m_\mathrm{e}a_0^2}{\hbar}	
\end{equation}
is a reduced (dimensionless) time,
$\beta(x,y)=\Gamma(x)\Gamma(y)/\Gamma(x+y)$ is the Euler beta function and $\Gamma(x)$ is the gamma function.

\subsection{Method for the higher-order Trotter bounds}
\label{sec:methods_higher_order}
Here, we present our method to obtain a bound on the state-dependent Trotter error $\xi_N^{(p)}(t;\varphi)$ defined in Eq.~(\ref{eq:higher_order_error}) for an arbitrary order $p$.
The technical details and the proofs which form the basis of this section can be found in Appendix~\ref{sec:higher_order_derivation}\@.
The idea of the method is to iterate the integration-by-part strategy developed in Ref.~\cite{Burgarth2022} to bound evolutions.

We first focus on a single Trotter cycle 
\begin{equation}
	\mathcal{S}_1^{(p)}(t/N)=\rme^{-\rmi\frac{t}{N}\tau_MH_{\sigma(M)}}\cdots \rme^{-\rmi\frac{t}{N}\tau_2H_{\sigma(2)}}\rme^{-\rmi\frac{t}{N}\tau_1H_{\sigma(1)}},
\label{eq:single_Trotter_cycle}
\end{equation}
and regard the sequence of the $M$ exponentials $\rme^{-\rmi\frac{t}{N}\tau_jH_{\sigma(j)}}$ ($j=1,\ldots,M$) as the evolution under a piecewise-constant time-dependent Hamiltonian $H^{(p)}(s)$, defined by
\begin{equation}
H^{(p)}(s)
=\tau_j H_{\sigma(j)}\ \ \mathrm{for}\ \ s\in[T_{j-1},T_j),
\label{eq:piecewise-constant}
\end{equation}
where
$T_j = j t/N$ for $j=1,\ldots,M$.
Without loss of generality, we shift the Hamiltonians $H_1$ and $H_2$ such that $(H_1+H_2)\varphi=0$, and the target evolution is the identity $I$ on the eigenstate $\varphi$. 
Then, in Appendix~\ref{sec:higher_order_derivation}, we show that the difference $\mathcal{S}_1^{(p)}(t/N)-I$
can be written, via iterating the integration by parts $q$ times, as
\begin{align}
&\mathcal{S}_1^{(p)}(t/N)-I
\nonumber\\
&\quad
=\sum_{k=1}^q(-\rmi)^kS_k(T_M)
\nonumber\\
&\qquad{}
+(-\rmi)^{q+1}\int_0^{T_M}\rmd s\,U^{(p)}(T_M,s)H^{(p)}(s)S_q(s),
\label{eq:iterated_formula}
\end{align}
where
\begin{gather}
	S_k(s)=\int_0^s\rmd s_1\cdots \int_0^{s_{k-1}}\rmd s_k\,
	 H^{(p)}(s_1)\cdots H^{(p)}(s_k),
	 \label{eq:action}
\\
U^{(p)}(t,s)=\T\exp\!\left(
-\rmi\int_s^t\rmd u\,H^{(p)}(u)
\right),\nonumber
\end{gather}
and $\T$ denotes time-ordering.
Note that $\mathcal{S}_1^{(p)}(t/N)=U^{(p)}(T_M,0)$, and Eq.~(\ref{eq:iterated_formula}) is essentially the Dyson expansion of $U^{(p)}(T_M,0)$.

A proper $p$th-order Trotter product formula should satisfy
\begin{equation}
S_k(T_M)=\frac{1}{k!}\frac{t^k}{N^k}(H_1+H_2)^k,\quad k=1,\ldots,p.
\label{eq:vanishing_actions}
\end{equation}
In other words, we can find a $p$th-order Trotter product formula, i.e.~a set of times $\{\tau_j\}$, by imposing the conditions in Eq.~(\ref{eq:vanishing_actions}).
We hence have 
\begin{equation*}
S_k(T_M)\varphi=0,\quad k=1,\ldots,p,
\end{equation*}
and the Trotter error $\xi_1^{(p)}(t/N;\varphi)=\|[\mathcal{S}_1^{(p)}(t/N)-I]\varphi\|$ on the input state $\varphi$ can be estimated by the last remainder term of the Dyson expansion~(\ref{eq:iterated_formula}), with some $q\le p$.
If $\varphi\in\mathcal{D}\bm{(}H^{(p)}(s)S_p(s)\bm{)}$ for all $s\in[0,T_M)$ and the remainder operator with $q=p$ can be applied to the input state $\varphi$, the error of the single Trotter cycle is bounded by
\begin{equation}
\xi_1^{(p)}(t/N;\varphi)
\le\int_0^{T_M}\rmd s\,\Vert H^{(p)}(s) S_p(s) \varphi\Vert,
\label{eq:higher_order_bound_starting_point_single}
\end{equation}
and the overall error $\xi_N^{(p)}(t;\varphi)$ of the $p$th-order Trotter product $\mathcal{S}_N^{(p)}(t)$ is bounded by
\begin{equation}
\xi_N^{(p)}(t;\varphi)\le N\xi_1^{(p)}(t/N;\varphi),
\label{eq:higher_order_bound_starting_point}
\end{equation}
which follows from a standard telescoping sum identity, see e.g.\ Ref.~\cite[Lemma~22]{Hahn2022}\@.

In order to bound Eq.~\eqref{eq:higher_order_bound_starting_point_single}, we need to evaluate $S_p(s)$.
This can be done by explicitly calculating the integrals in Eq.~\eqref{eq:action}.
Recall that the single Trotter cycle $\mathcal{S}_1^{(p)}(t/N)$ consists of $M$ time slots, see Eq.~\eqref{eq:single_Trotter_cycle}.
In the $j$th time slot, the Hamiltonian is constant, $H^{(p)}(s)=\tau_j H_{\sigma(j)}$, for time $t/N$, as defined in Eq.~\eqref{eq:piecewise-constant}.
For $s\in[T_{j-1},T_j)$ in the $j$th time slot, we get
\begin{align}
S_k(s)
={}&
\mathop{\sum_{\bm{\beta}\in\mathbb{N}^j}}_{|\bm{\beta}|=k} \left(
\frac{t}{N}
\right)^{k-\beta_j}
\frac{1}{\bm{\beta}!}
\left[\left(s-\frac{(j-1)t}{N}\right) \tau_j H_{\sigma(j)}\right]^{\beta_j}
\nonumber\\[-2truemm]
&\qquad\qquad\quad\ \
{}\times
(\tau_{j-1} H_{\sigma(j-1)})^{\beta_j}\cdots 
(\tau_1 H_{\sigma(1)})^{\beta_1},
\label{eq:kth_action_explicit}
\end{align}
where $\bm{\beta}=(\beta_1,\dots,\beta_j)$ is a multi-index of nonnegative integers, $|\bm{\beta}|=\beta_1+\cdots+\beta_j$, and $\bm{\beta}!=\beta_1!\cdots\beta_j!$.
At the end of the Trotter cycle $s=T_M=Mt/N$, it reads
\begin{equation}
S_k(T_M)
=
\left(
\frac{t}{N}
\right)^k
\mathop{\sum_{\bm{\beta}\in\mathbb{N}^M}}_{|\bm{\beta}|=k}
\frac{1}{\bm{\beta}!}
(\tau_M H_{\sigma(M)})^{\beta_M}
\cdots
(\tau_1H_1)^{\beta_1}.
\label{eq:kth_action_explicit_TM}
\end{equation}
Now, as we have an explicit expression of $S_p(s)$, we can estimate a bound on Eq.~\eqref{eq:higher_order_bound_starting_point}.
By dividing the total time interval $[0, Mt/N)$ into $M$ time slots and by using the triangle inequality, we get
\begin{equation}
\xi_N^{(p)}(t;\varphi)\leq N \sum_{j=1}^M \int_{(j-1)t/N}^{j t/N}\rmd s\,\Vert \tau_j H_{\sigma(j)} S_p(s)\varphi\Vert.
\label{eq:bound_higher_order_integral}
\end{equation}
Equation~\eqref{eq:bound_higher_order_integral} gives a bound on a higher-order Trotterization.
Such a bound can be further simplified by using identities due to $(H_1+H_2)\varphi=0$, i.e.
\begin{equation*}
H_1\varphi=-H_2\varphi.
\end{equation*}

We provide a Mathematica script that performs this computation for given switching times $\{\tau_j\}$ in Appendix~\ref{sec:Mathematica}\@.
Furthermore, we give an example for $p=2$ in Appendix~\ref{sec:higher_order_derivation}, where we perform the calculation analytically.
The last step is to shift back the energy of the target Hamiltonian so that indeed $(H_1+H_2)\varphi=h\varphi$.
This is trivial if the input state $\varphi$ satisfies the corresponding domain conditions.
It is done by distributing the energy $h$ to the individual Hamiltonians $H_1$ and $H_2$.
Since it is arbitrary in which way $h$ is distributed, we can replace $H_1\rightarrow {H}_1(g)=H_1-g$ and $H_2\rightarrow {H}_2(g)=H_2-h+g$ with any $g\in\mathbb{R}$.
If $\varphi$ does not satisfy the domain conditions and admits a fractional scaling $\mathcal{O}(N^{-\delta})$, the energy shift will add additional terms that scale faster than $N^{-\delta}$.
They are computed explicitly in Eq.~\eqref{eq:first_order_with shift} of Appendix~\ref{sec:hydrogen_bounds}\@.
For example, for states $\varphi$ that yield scalings slower than $\mathcal{O}(N^{-1})$ in the case of the second-order Trotterization, combining Eq.~\eqref{eq:fractional_second_1} with Eq.~\eqref{eq:first_order_with shift} for $H_1$ gives
\begin{equation*}
	\xi_N^{(2)}(t;\varphi)\leq \frac{1}{2}b_N(t;\varphi),
\end{equation*}
where $b_N(t;\varphi)$ is our bound for such states on the first-order Trotterization.
That is, $\xi_N(t;\varphi)\leq b_N(t;\varphi)$, see Eq.~\eqref{eq:first_order_fractional_scaling}.

\section{Generalizations and outlook}\label{sec:outlook}
While the methods derived in Sec.~\ref{sec:methods_higher_order} are used to compute bounds for higher-order Trotterization, they are more general than that.
For this reason, they might have further applications in various other fields.
In particular, Eq.~\eqref{eq:kth_action_explicit} provides a closed form for the Dyson expansion of any unitary $U$ generated by a \emph{piecewise-constant} time-dependent Hamiltonian.
Since the Dyson expansion arises by repeatedly reinserting a differential equation (the Schr{\"o}dinger equation) into itself, our method might be useful for solving general (piecewise-constant) linear differential equations.
Of course, in the case of non-self-adjoint generators, $U$ will not be unitary and some adjustments have to be made.
However, for applications in quantum mechanics with a piecewise-constant Hamiltonian, Eq.~\eqref{eq:kth_action_explicit} can be used directly.
For this, one only has to omit the map $\sigma$, which accounts for the periodicity of the Trotter Hamiltonian $H^{(p)}(s)$.
A potential application could be quantum control theory when considering time-independent drifts and piecewise-constant controls.

In the same manner, our methods from Sec.~\ref{sec:methods_higher_order} can be used to compute state-dependent bounds on Trotter errors for more than two operators.
For example, for a first-order Trotter product ($\tau_j=1$ for all $j=0,\ldots,L-1$) with a target Hamiltonian $H=\sum_{j=0}^{L-1} H_j$, we can compute the integral action at the $j$th time step as
\begin{equation}
	S_1(s) = \left(s-\frac{jt}{N}\right) H_{j\ \text{mod}\ L} + \frac{t}{N}\sum_{i=0}^{j-1} H_{i\ \text{mod}\ L},
	\label{eq:first_order_more_operators}
\end{equation}
for $s\in[T_{j-1},T_j)$.
Similar considerations can be carried out for higher-order schemes with $L$ summands in the target Hamiltonian as well.

For quantum simulations, there is a crucial consequence from the slower Trotter error scaling we observe for the ground state of the hydrogen atom.
One advantage of a quantum approach is the ability to use the large Hilbert space to store a big, hence very accurate, approximation of the Hamiltonian.
However, for states with low angular momenta, the higher the truncation dimension is, the more Trotter steps are necessary to arrive at the desired $N^{-1}$ scaling regime.
This shows that there is an intricate interplay between the error due to truncation and the Trotter error.
See also Ref.~\cite{Burgarth2023}.
It would be interesting to see whether there is an optimal choice of parameters for the truncation dimension and the number of Trotter steps so that the total simulation error is minimal.
Some ideas from bosonic systems might potentially help tackle this problem~\cite{Tong2022, Peng2023}.

One might suspect that the slower scaling also occurs for other atoms and molecules than the hydrogen atom.
This is suggestive since the scaling arises from the unbounded nature of the involved Hamiltonians, such as the $1/r$ singularity of the Coulomb potential.
For this reason, one would expect that states with a large electron density close to Coulomb singularities admit slower Trotter scalings.
More specifically, for the eigenstates of many-electron systems, one usually makes an ansatz through a linear combination of atomic orbitals inspired by the eigenfunctions of the hydrogen atom~\cite{Huheey1973, Helgaker2014}.
As we have already seen in Eq.~\eqref{eq:superposition_states}, such states can suffer from a slower scaling as well.
This gives rise to the hypothesis that molecular states overlapping the atomic orbitals with a high probability distribution close to one of the Coulomb singularities might reveal a reduced Trotter error scaling.
In addition, further electron-electron singularities can occur in more complicated atomic and molecular systems.
These could potentially slow the error scaling even further.
This conjecture is bolstered by the fact that some first signs of slower Trotter convergence for low-energy molecular orbitals have already been observed in Ref.~\cite{Chan2023}.
The authors speculate that this issue could potentially be resolved by investigating higher-order Trotter schemes.
However, our results suggest that this is not the case.

Furthermore, one might expect similar effects of slower algorithmic convergence to appear in other quantum chemistry methods such as qubitization and linear combination of unitaries.
For example, a recent paper by Mukhopadhyay \textit{et~al.}~\cite{Mukhopadhyay2023} compares a qubitization algorithm with a Trotter algorithm.
In both cases, the approximation error diverges in the full-system limit of no truncation.
This indicates that the effects of unbounded operators play a role in this scheme as well.
It would be interesting to further analyze how exactly the truncation error enters here.

Finally, from a mathematical perspective, it would be interesting to study the Trotter error in the $H^1$-norm, which is the natural norm in terms of the system's energy.
In addition, a generalization to arbitrary input states would extend our results to cases where the series in Eq.~\eqref{eq:superposition_states} diverges.

\section{Conclusion}
\label{sec:conclusion}
In summary, we have provided
\begin{enumerate}
	\item state-dependent bounds for the first-order Trotterization valid for unbounded operators,
	\item state-dependent bounds for the higher-order Trotterizations valid for unbounded operators, and
	\item bounds for the first- and higher-order Trotterizations, which admit slower fractional scalings when the input state does not satisfy the respective domain conditions.
\end{enumerate}
Our methods naturally generalize to Trotter problems with more than two operators.
These results are very general and might have wide-ranging applications in different fields, where the Trotter product formula is used.
This includes quantum computing, quantum control, quantum field theory, open quantum systems, quantum optics, Floquet physics, as well as quantum many-body physics~\cite{Burgarth2023}.
We discussed direct implications for quantum chemistry problems by studying the hydrogen atom.
For those applications, a state-dependent analysis is necessary as the involved Hamiltonians are unbounded.
Our results show three implications.
\begin{enumerate}
	\item There are input states that give slower scalings than the standard $N^{-1}$ scaling. For the hydrogen atom, the scaling is solely determined by the quantum number $\ell$ of the orbital angular momentum. Most importantly, the $\ell=0$ states including the ground state yield the $N^{-1/4}$ scaling. This behavior is due to the singularity of the Coulomb potential. It can be explained by the high electron density of low-angular-momentum states close to that singularity. This leads to diverging fourth moments in the kinetic and potential energies, due to which the Trotter dynamics can ionize the hydrogen atom. We conjecture that this behavior is generic to quantum chemistry as the Coulomb potential is ubiquitous there. 
	\item Input states with slower Trotter convergence do not benefit from an enhanced scaling by the higher-order Trotter-Suzuki methods. For a $p$th-order product formula to give the scaling $\mathcal{O}(t^{p+1}/N^p)$, the moments of the kinetic and potential energies of order $2(p+1)$ must be finite. This condition does not hold for states with low $\ell$, such as the ground state.
	\item For states with finite higher moments of energy, both the first- and higher-order methods might lead to the desired $N^{-p}$ scaling. However, the order of the operators in the Trotter product can determine the provable scaling behavior. This is due to an asymmetry in the conditions for the higher moments.
\end{enumerate}
Our findings imply that the first-order Trotterization is the most efficient scheme for the ground-state quantum chemistry, as it involves the least number of unitaries per Trotter cycle.
Furthermore, some previous resource and runtime estimates based on the $N^{-1}$ Trotter error scaling are too optimistic.
In practice, this means that the timescale towards useful chemistry simulations on quantum computers is longer than expected.
Nevertheless, our findings only imply a polynomially increased Trotter error and do not rule out a potential exponential quantum advantage.
In fact, the slower scaling also carries over to numerical simulations on classical computers.
Overall, we treat errors in digital simulations in an analytical way.
We hope that our analytical insights will inspire the development of more efficient (quantum) simulation algorithms and provide a methodology for the proof of exponential quantum advantage.

\acknowledgments
The authors thank Lauritz van Luijk for helpful feedback and for pointing out the argument for the non-existence of general state-dependent Trotter bounds with commutator scaling from Thm.~\ref{thm:no-go_commutator_scaling}\@.
DB and AH would like to thank Elliot Eklund and Ivan Kassal for their interesting and helpful discussions.
DB further acknowledges Reinhard F.\ Werner for a fruitful dialogue.
In addition, AH thanks Dominic Berry and Nathan K.\ Langford for stimulating conversations.

DB acknowledges funding from the Australian Research Council (project numbers FT190100106, DP210101367, CE170100009) and the Munich Quantum Valley project K8\@.
PF acknowledges financial
support from the PNRR MUR project
CN00000013 - National Centre for HPC, Big Data and
Quantum Computing, as well as partial support by Istituto Nazionale di Fisica Nucleare (INFN) through
the project ``QUANTUM'' and by the Italian National Group of Mathematical Physics (GNFM-INdAM).
AH was partially supported by the Sydney Quantum Academy.
KY acknowledges support by the Top Global University Project from the Ministry of Education, Culture, Sports, Science and Technology (MEXT), Japan, and by JSPS KAKENHI Grants No. JP18K03470, No. JP18KK0073, and No. JP24K06904 from the Japan Society for the Promotion of Science (JSPS).

\begin{widetext}
\appendix

\section{First-order Trotter bounds (Thm.~\ref{thm:trotter_thm} and Thm.~\ref{thm:first_order_alpha})}
\label{sec:proof_trotter_thm}
In this appendix, we provide the proofs of Thm.~\ref{thm:trotter_thm} and Thm.~\ref{thm:first_order_alpha} presented in Sec.~\ref{sec:first_order}\@.
We start by proving a lemma.
In the following, we denote by $I$ the identity operator.
\begin{lemma}\label{lem:bound_refined}
Let $H_{1}$ be self-adjoint on $\mathcal{D}(H_{1})$ and $H_{2}$
be self-adjoint on $\mathcal{D}(H_{2})$. Let $\varphi\in\mathcal{D}(H_{1})\cap\mathcal{D}(H_{2})$
be an eigenstate of $H_{1}+H_{2}$ with eigenvalue $0$, i.e.\ $(H_{1}+H_{2})\varphi=0$.
Then, for all $t\in\mathbb{R}$,
\begin{equation}
\xi_N(t;\varphi)
\le N\left\Vert \left(\rme^{-\rmi \frac{t}{N}H_{1}}-I+\rmi\frac{t}{N}H_{1}\right)\varphi\right\Vert
		+N \left\Vert \left(\rme^{\rmi\frac{t}{N}H_{2}}-I-\rmi\frac{t}{N}H_{2}\right)\varphi\right\Vert .
\label{eq:lemma_bound1}
\end{equation}
Moreover, if $\varphi\in\mathcal{D}(H_{1}^2)\cap\mathcal{D}(H_{2}^2)$, it is further bounded by
\begin{equation}
\xi_N(t;\varphi)
\le\frac{t^2}{2N}\,\Bigl(
\Vert H_1^2\varphi\Vert
+
\Vert H_2^2\varphi\Vert
\Bigr).
\label{eq:lemma_bound2}
\end{equation}
\end{lemma}
\begin{proof}
The definition of $\xi_N(t;\varphi)$ is given in Eq.~(\ref{eq:definition_standard_error}).
Under the assumption $(H_1+H_2)\varphi=0$, it reads
	\begin{equation*}
		\xi_N(t;\varphi)=\left\Vert\left[\left(\rme^{-\rmi \frac{t}{N}H_2}\rme^{-\rmi \frac{t}{N}H_1}\right)^{N}-I\right]\varphi\right\Vert.
	\end{equation*}
By telescoping the sum, we have
\begin{equation*}
D_N(t)\equiv\left(\rme^{-\rmi \frac{t}{N}H_2}\rme^{-\rmi \frac{t}{N}H_1}\right)^{N}-I
=\sum_{\ell=1}^{N}\left(\rme^{-\rmi \frac{t}{N}H_2}\rme^{-\rmi \frac{t}{N}H_1}\right)^{N-\ell}\left(\rme^{-\rmi \frac{t}{N}H_2}\rme^{-\rmi \frac{t}{N}H_1}-I\right).
\end{equation*}
Then, using the invariance of the norm under unitaries and the assumption $(H_{1}+H_{2})\varphi=0$, we get
\begin{align}
\xi_N(t;\varphi)
=
\Vert D_N(t)\varphi\Vert
&\le
N\left\Vert \left(\rme^{-\rmi \frac{t}{N}H_2}\rme^{-\rmi \frac{t}{N}H_1}-I\right)\varphi\right\Vert
\nonumber\\
&=N\left\Vert \left(\rme^{-\rmi \frac{t}{N}H_{1}}-\rme^{\rmi \frac{t}{N}H_{2}}\right)\varphi\right\Vert
\nonumber\\
&=N\left\Vert \left(\rme^{-\rmi \frac{t}{N}H_{1}}-I+\rmi\frac{t}{N}H_{1}-\rme^{\rmi\frac{t}{N}H_{2}}+I+\rmi\frac{t}{N}H_{2}\right)\varphi\right\Vert \nonumber\\
& \le N\left\Vert \left(\rme^{-\rmi \frac{t}{N}H_{1}}-I+\rmi\frac{t}{N}H_{1}\right)\varphi\right\Vert +N \left\Vert \left(\rme^{\rmi\frac{t}{N}H_{2}}-I-\rmi\frac{t}{N}H_{2}\right)\varphi\right\Vert,
\label{eq:first_order_bound}
\end{align}
which proves the first bound~(\ref{eq:lemma_bound1}).
Under the additional domain condition $\varphi\in\mathcal{D}(H_{1}^2)\cap\mathcal{D}(H_{2}^2)$, we can further bound it by Eq.~(\ref{eq:lemma_bound2}), using
\begin{equation}
(\rme^{-\rmi sH_j}-I+\rmi sH_j)\varphi
=
-\int_0^s\rmd u\,\rme^{-\rmi (s-u)H_j}uH_j^2 \varphi,
\label{eq:Taylor_first}
\end{equation}
and hence
\begin{equation}
\|(\rme^{-\rmi sH_j}-I+\rmi sH_j)\varphi\|
\le 
\frac{1}{2}s^2\|H_j^2\varphi\|,
\label{eq:Taylor2}
\end{equation}
for all $s\in\mathbb{R}$ and $j=1,2$.
\end{proof}

\subsection{Proof of Thm.~\ref{thm:trotter_thm}}
We are now ready to prove Thm.~\ref{thm:trotter_thm}\@.
\begin{proof}[Proof of Thm.~\ref{thm:trotter_thm}]
We shift the eigenvalue of the target Hamiltonian $H_1+H_2$ with respect to $\varphi$ to zero and apply the second bound~(\ref{eq:lemma_bound2}) of Lemma~\ref{lem:bound_refined}:
\begingroup
\allowdisplaybreaks
\begin{align*}
\xi_N(t;\varphi)
={}&\left\Vert\left[\left(\rme^{-\rmi \frac{t}{N}H_2}\rme^{-\rmi \frac{t}{N}H_1}\right)^{N}-\rme^{-\rmi th}\right]\varphi\right\Vert 
\nonumber\\
={}&\left\Vert\left[\left(\rme^{-\rmi \frac{t}{N}(H_2+g-h)}\rme^{-\rmi \frac{t}{N}(H_1-g)}\right)^{N}-I\right]\varphi\right\Vert 
\nonumber\\
\le{}&
N\left\Vert \left(\rme^{-\rmi \frac{t}{N} (H_1-g)}-I+\rmi\frac{t}{N} (H_1-g)\right)\varphi\right\Vert
		+N \left\Vert \left(\rme^{\rmi\frac{t}{N}(H_2+g-h)}-I-\rmi\frac{t}{N}(H_2+g-h)\right)\varphi\right\Vert
		\nonumber\\
\le{}&\frac{t^{2}}{2N}\,\Bigl(\Vert (H_1-g)^{2}\varphi\Vert+\Vert (H_2+g-h)^{2}\varphi\Vert\Bigr).
\end{align*}
\endgroup
This proves Thm.~\ref{thm:trotter_thm}\@.
\end{proof}

\subsection{Proof of Thm.~\ref{thm:first_order_alpha}}
If the domain condition $\varphi\in\mathcal{D}(H_{1}^2)\cap\mathcal{D}(H_{2}^2)$ is not fulfilled, the bound~(\ref{eq:Taylor2}) is not useful and the second bound~(\ref{eq:lemma_bound2}) of Lemma~\ref{lem:bound_refined} does not hold.
In such a case, we need to perform a refined analysis to estimate the scaling of the first bound~(\ref{eq:lemma_bound1}) of Lemma~\ref{lem:bound_refined}\@.
Let us look at
\begin{equation}
\Delta_j(s)
=\|(\rme^{-\rmi sH_j}-I+\rmi sH_j)\varphi\|^2
=\int_{\mathbb{R}}|\rme^{-\rmi s\lambda}-1+\rmi s\lambda|^2
\,\rmd\mu_{j,\varphi}(\lambda),\quad j=1,2,
\label{eq:Delta1_def}
\end{equation}
for $s\in\mathbb{R}$, where 
\begin{equation*}
	\mu_{j,\varphi}(\Omega)= \| \mathbf{1}_{\Omega}(H_j)\,\varphi\|^2,
\end{equation*}
for a measurable $\Omega\subset \mathbb{R}$, is the spectral measure of $H_j$ at the state $\varphi$, with $\mathbf{1}_{\Omega}(x)$ being the characteristic function of $\Omega$, which equals 1 for $x\in\Omega$ and 0 for $x\notin\Omega$.
Note that the convergence of the integral in Eq.~\eqref{eq:Delta1_def} only requires $\varphi\in\mathcal{D}(H_j)$.
Since $f(x)= |\rme^{-\rmi x}-1+\rmi x|^2= f(-x)$, it suffices to study
\begin{equation*}
\Delta_j(s)
=\int_{\mathbb{R}}f(s\lambda)\,\rmd\mu_{j,\varphi}(\lambda)
=\int_{\mathbb{R}_+} f(s\lambda)\,\rmd\nu_{j,\varphi}(\lambda),
\end{equation*}
where $\mathbb{R}_+=[0,+\infty)$, and
\begin{equation*}
	\nu_{j,\varphi}(\Omega) = \mu_{j,\varphi}\bm{(}\Omega \cup (-\Omega) \bm{)},\quad \Omega \subset\mathbb{R}_+.
\end{equation*}
The asymptotic behavior of $\Delta_j(s)$ for small $s$ is ruled by the high-energy tail of the measure $\nu_{j,\varphi}$.
The following Lemma provides a bound on $\Delta_j(s)$, when we assume a suitable decay of $\nu_{j,\varphi}\bm{(}(\lambda,+\infty)\bm{)}$.
\begin{lemma}\label{lem:FractionalScaling}
Let $\mu_\varphi(\Omega)$ be the spectral measure of a self-adjoint operator $H$ at the vector $\varphi$. Assume that
\begin{equation*}
	 \mu_{\varphi}(\{|x|\geq \lambda \}) = \nu_{\varphi}\bm{(}(\lambda,+\infty)\bm{)} \leq \left(\frac{\lambda_0}{\lambda}\right)^{2\delta},
\end{equation*}
for some $\lambda_0>0$ and $\delta>1$, and for all $\lambda>0$.
Then, 
\begin{equation}
\Delta(s)
=\int_{\mathbb{R}_+}|\rme^{-\rmi s\lambda}-1+\rmi s\lambda|^2
\,\rmd\nu_{\varphi}(\lambda)
\label{eq:Delta(s)def}
\end{equation} 
is bounded by
\begin{equation}
\Delta(s)
\le
\begin{cases}
\medskip
\displaystyle
\frac{\pi}{\Gamma(2\delta-1)\sin[(\delta-1)\pi]}(\lambda_0 s)^{2\delta},
\quad&1<\delta<2,\\
\medskip
\displaystyle
\frac{1}{4}\Lambda^2\|H\varphi\|^2s^4
+\frac{\lambda_0^4}{\Lambda^4}
g(\Lambda s),
\quad&\delta=2,\\
\displaystyle
\frac{1}{4}\|H^2\varphi\|^2s^4,&
\delta>2,
\end{cases}
\label{eqn:fractional_bound}
\end{equation}
for $s>0$, where $\Lambda>0$ is an arbitrary energy scale, and
\begin{equation*}
g(x)
=
-x^4\Ci(x)+\left(1+\frac{1}{2}x^2\right)[(1-\cos x)^2+(x-\sin x)^2]-2x^3(x-\sin x)+\frac{3}{2}x^4,
\end{equation*}
with
$\Ci(x)=-\int_x^\infty\frac{\cos t}{t}\,\rmd t$.
Note that 
\begin{equation*}
g(x)=-x^4\log x+\left(\frac{7}{4}-\gamma\right)x^4+\mathcal{O}(x^6),
\end{equation*}
as $x\downarrow 0$, with Euler's constant $\gamma=\Gamma'(1)\approx 0.577$. 
In particular, one gets 
\begin{equation}
g(x) \leq -x^4\log x+ 1.2003 x^4,
\label{eq:g(x)smallx}	
\end{equation}
for $0\leq x\leq 1$.
\end{lemma}
\begin{proof}
For $1<\delta<2$, we use integration by parts.
Let 
\begin{equation*}
f(x)
=|\rme^{-\rmi x}-1+\rmi x|^2
=(1-\cos x)^2+(x-\sin x)^2 .
\end{equation*}
Note that $f\in L^1(\mathbb{R}_+,\rmd\nu_{\varphi})$, $f'(\lambda)/\lambda^{2\delta}\in L^1(\mathrm{R}_+)$, $f(0)=0$, and $f(\lambda)=o(\lambda^{2\delta})$ as $\lambda\to+\infty$.
Since $\rmd\nu_{\varphi}(\lambda)= - \rmd[\nu_{\varphi}\bm{(}(\lambda,+\infty)\bm{)}]$, we get
\begin{align*}
	\Delta(s)
	&=
	-\int_{\mathbb{R}_+} \rmd[f(s\lambda)\nu_{\varphi}\bm{(}(\lambda,+\infty)\bm{)}]
	+\int_0^{+\infty}sf'(s\lambda)\nu_{\varphi}\bm{(}(\lambda,+\infty)\bm{)}\,\rmd\lambda\\
	&=
	f(0) \nu_{\varphi}\bm{(}(0,+\infty)\bm{)} 
	- \lim_{\lambda\rightarrow+\infty} f(s\lambda)\nu_{\varphi}\bm{(}(\lambda,+\infty)\bm{)}
	+ \int_0^{+\infty}f'(x)\nu_{\varphi}\bm{(}(x/s,+\infty)\bm{)}\,\rmd x.
\end{align*}
Therefore,
\begingroup
\allowdisplaybreaks
\begin{align*}
	\Delta(s)
	={}&
	\int_0^{+\infty}f'(x)\nu_{\varphi}\bm{(}(x/s,+\infty)\bm{)}\,\rmd x\\
	\leq{}& (\lambda_0 s)^{2\delta} \int_0^{+\infty}\frac{f'(x)}{x^{2\delta}}\,\rmd x\\
	={}&
	\frac{\pi}{\Gamma(2\delta-1)\sin[(\delta-1)\pi]}(\lambda_0 s)^{2\delta}.
\end{align*}
\endgroup
For $\delta=2$, we fix $\Lambda>0$ and split the integration range as
\begin{equation*}
\Delta(s)
=\int_{[0,\Lambda]} f(s\lambda)\,\rmd\nu_{\varphi}(\lambda)
+\int_{(\Lambda,+\infty)} f(s\lambda)\,\rmd\nu_{\varphi}(\lambda)
=\Delta_j^{<}(s)+\Delta_j^{>}(s).
\end{equation*}
Since $f(x)\le x^4/4$, the first term is estimated as
\begingroup
\allowdisplaybreaks
\begin{align*}
\Delta^{<}(s)
={}&
\int_{[0,\Lambda]} f(s\lambda)\,\rmd\nu_{\varphi}(\lambda)
\nonumber\\
\le{}&
\frac{1}{4}s^4
\int_{[0,\Lambda]} \lambda^4\,\rmd\nu_{\varphi}(\lambda)
\nonumber\\
\leq{}&
\frac{1}{4}\Lambda^2 s^4
\int_{[0,\Lambda]}\lambda^2\,\rmd\nu_{\varphi}(\lambda)
\nonumber\\
\le{}&
\frac{1}{4}\Lambda^2 s^4
\int_{\mathbb{R}_+}\lambda^2\,\rmd\nu_{\varphi}(\lambda)
\nonumber\\
={}&
\frac{1}{4}\Lambda^2\|H\varphi\|^2 s^4.
\intertext{Since
$f'(x)= 2x (1-\cos x) \geq 0$ for $x\geq0$,
the second term is bounded by}
\Delta^{>}(s)
={}&
\int_{(\Lambda,+\infty)} f(s\lambda)\,\rmd\nu_{\varphi}(\lambda)
\nonumber\\
={}& {-}\int_{(\Lambda,+\infty)}f(s\lambda)\,\rmd[\nu_{\varphi}\bm{(}(\lambda,+\infty)\bm{)}]
\nonumber\\
={}& f(s\Lambda)\nu_{\varphi}\bm{(}(\Lambda,+\infty)\bm{)} + s \int_{\Lambda}^{+\infty} f'(s\lambda)\nu_{\varphi}\bm{(}(\lambda,+\infty)\bm{)}\,\rmd\lambda
\nonumber\\
\leq{}& f(s\Lambda)\frac{\lambda_0^4}{\Lambda^4} + s \int_{\Lambda}^{+\infty} f'(s\lambda)\frac{\lambda_0^4}{\lambda^4}\,\rmd\lambda
\nonumber\\
={}& 4 \lambda_0^4 \int_{\Lambda}^{+\infty} \frac{f(s\lambda)}{\lambda^5}\,\rmd\lambda
\nonumber\\
={}&		
4 (\lambda_0 s)^4\int_{s\Lambda}^{+\infty}\frac{f(x)}{x^5}\,\rmd x
\nonumber
\\
={}&
\frac{\lambda_0^4}{\Lambda^4}g(\Lambda s).
\end{align*}
\endgroup
Finally, if $\delta>2$ then $\varphi\in \mathcal{D}(H^2)$, and the bound is given by Eq.~(\ref{eq:Taylor2}) obtained in the proof of Lemma~\ref{lem:bound_refined}\@.
\end{proof}
Let us collect some comments on this lemma in the following remark.
\begin{remark}\label{rem:tail_bound}
	\begin{enumerate}[label=(\roman*)]
	\item $\lambda_0$ and $\Lambda$ have the dimension of energy.
			$\Lambda$ is an arbitrary energy scale, which can be tuned to optimize the bound for $\delta=2$.
	The optimal $\Lambda$ is the solution to the equation
	$$
	\frac{f(s\Lambda)}{(s\Lambda)^6}=\frac{\|H\varphi\|^2}{8\lambda_0^4} \frac{1}{s^2}.
	$$
	Since $f(x)\sim x^4/4$ as $x\to 0$, one gets that for $s\downarrow0$ the optimal value approaches $\Lambda \to \sqrt{2}\, \lambda_0^2 / \|H\varphi\|$, whence a quasi-optimal bound for small $s$ is given by
	$$
	\Delta(s) \leq \frac{1}{2}(\lambda_0 s)^4 + \frac{ \|H\varphi\|^4}{4 \lambda_0^4} g\!\left(\frac{\sqrt{2}\,\lambda_0}{\|H\varphi\|} \lambda_0 s\right).
	$$
	By using the upper bound~\eqref{eq:g(x)smallx}, one gets for  $s\leq \|H\varphi\|/ (\sqrt{2}\,\lambda_0^2)$
	\begin{equation*}
		\Delta(s) \leq (\lambda_0 s)^4 \left(-\log(\lambda_0 s) + 1.7003-\log\frac{\sqrt{2}\,\lambda_0}{\|H\varphi\|}\right) .
	\end{equation*}
		\item The lemma can be extended to the range $\delta\leq 1$. Indeed, if one only assumes that $\varphi\in\mathcal{D}(H)$, so that $\Delta(s)$ in Eq.~\eqref{eq:Delta(s)def} is finite, without any further assumptions on the decay of the spectral measure, one easily gets from Eq.~\eqref{eq:Delta(s)def}
	\begin{equation*}
\Delta(s)
=s^2 \int_{\mathbb{R}_+}|F(s\lambda)|^2
\lambda^2\,\rmd\nu_{\varphi}(\lambda),
\qquad F(x)= \rme^{\rmi x/2} \sinc(x/2)-1.
	\end{equation*}
Since $F(x)$ is bounded and $F(x)\to 0$ as $x\to 0$, by dominated convergence one has that 
\begin{equation*}
	\Delta(s) =s^2 \|F(s H)H\varphi\|^2 = o(s^2). 
\end{equation*}
Notice, however, that in such a case we have no control over the rate of convergence of $\Delta(s)/s^2$, which can be arbitrarily slow.
		\item	If $\delta>2$, then $\varphi\in\mathcal{D}(H^2)$, and the bound from Thm.~\ref{thm:trotter_thm} applies.
			However, in this case, there might be a possibility of gaining a better scaling by a higher-order product formula.
	\end{enumerate}
\end{remark}

\bigskip
Theorem~\ref{thm:first_order_alpha} is immediately proven, using the scaling analysis of Lemma~\ref{lem:FractionalScaling}\@.
\begin{proof}[Proof of Thm.~\ref{thm:first_order_alpha}]
By the same replacements as in the proof of Thm.~\ref{thm:trotter_thm} and using the first bound~(\ref{eq:lemma_bound1}) of Lemma~\ref{lem:bound_refined}, we have
\begin{equation}
\xi_N(t;\varphi)
\le
N\,\Bigl(
\sqrt{\Delta_1(t/N)}
+
\sqrt{\Delta_2(t/N)}
\Bigr).
\label{eq:first_order_fractional_scaling}
\end{equation}
Under the conditions for Thm.~\ref{thm:first_order_alpha}, bounds on $\Delta_j(t/N)$ scale as Eq.~(\ref{eqn:fractional_bound}) of Lemma~\ref{lem:FractionalScaling}, i.e.,
\begin{equation*}
\Delta_j(t/N)
=
\begin{cases}
\medskip
\displaystyle
\mathcal{O}\bm{(}(t/N)^{2\delta}\bm{)},
\quad&1<\delta<2,\\
\displaystyle
\mathcal{O}\bm{(}(t/N)^4\log(N/t)\bm{)},
\quad&\delta=2,
\end{cases}
\end{equation*}
and we get the scalings of $\xi_N(t;\varphi)$ in Thm.~\ref{thm:first_order_alpha}\@.
\end{proof}

\section{Analytical Trotter bounds for the hydrogen atom}
\label{sec:hydrogen_bounds}
In this section, we apply the Trotter error bounds from Thm.~\ref{thm:trotter_thm} and Lemma~\ref{lem:FractionalScaling} (Thm.~\ref{thm:first_order_alpha}) to the Hamiltonian of the hydrogen atom, where we Trotterize between the kinetic energy (temporarily reinstating the Planck constant $\hbar$)
\begin{equation*}
H_1=-\frac{\hbar^2}{2m_\mathrm{e}}\Delta
\end{equation*}
and the potential energy
\begin{equation*}
H_2=-\frac{e^2}{4\pi\varepsilon_0 r}.
\end{equation*}
This results in the explicit bounds in Eqs.~\eqref{eq:bound_R_n0}--\eqref{eq:bound_R_nl} presented in Sec.~\ref{sec:trotter_thm} of the Methods section. 
Recall that the Bohr radius $a_0$ is given by
\begin{equation*}
	a_0=\frac{4\pi\varepsilon_0\hbar^2}{m_\mathrm{e}e^2}.
\end{equation*}
Furthermore, the hydrogen atom $H_1+H_2$ has eigenenergies
\begin{equation*}
	E_n =-\frac{\hbar^2}{2m_\mathrm{e}a_0^2n^2},\quad n=1,2,\ldots,
\end{equation*}
corresponding to the radial hydrogen wave functions
$R_{n\ell}(r)=a_0^{-3/2}\mathcal{R}_{n\ell}(r/a_0)$ ($\ell=0,\ldots,n-1$) with
\begin{equation}
\mathcal{R}_{n\ell}(u)
=
\left(\frac{2}{n}\right)^{3/2}
\sqrt{
\frac{(n-\ell-1)!}{2n(n+\ell)!}
}
\rme^{-u/n}
(2u/n)^\ell
L_{n-\ell-1}^{(2\ell+1)}(2u/n),
\label{eq:hydrogen_eigenfunctions}
\end{equation}
where
\begin{equation*}
	L_{n}^{(\alpha)}(x)=\sum_{i=0}^{n}\frac{(-1)^{i}}{i!}\binom{n+\alpha}{n-i}x^{i}
\end{equation*}
are the generalized Laguerre polynomials~\cite[Chap.~22]{Abramowitz1972}. 
The full eigenstates $\Psi_{n\ell m}$ of the hydrogen atom including the degeneracies (characterized by the magnetic quantum number $m=-\ell,-\ell+1,\ldots,\ell-1,\ell$) are given by
\begin{equation*}
	\Psi_{n\ell m}(r, \theta, \phi) = R_{n\ell}(r) Y_{\ell m}(\theta, \phi),
\end{equation*}
where the spherical harmonics $Y_{\ell m}(\theta, \phi)$ are defined by
\begin{equation*}
	Y_{\ell m}(\theta, \phi) = \sqrt{\frac{2\ell +1}{4\pi}\frac{(\ell-m)!}{(\ell+m)!}} \rme^{\rmi m\phi} P_{\ell m}(\cos\theta),
\end{equation*}
with the associated Legendre polynomials
\begin{equation*}
	P_{\ell m}(x) = \frac{(-1)^m}{2^\ell \ell!} (1-x^2)^{m/2} \frac{\rmd^{\ell+m}}{\rmd x^{\ell+m}}(x^2-1)^\ell.
\end{equation*}
However, the spherical part $Y_{\ell m}$ only accounts for the degeneracies in the spectrum of the hydrogen atom.
Different values of $m$ for fixed $(n,\ell)$ lead to the same Trotter dynamics.
For this reason, it suffices to bound the Trotter convergence speed for the radial part $R_{n\ell}(r)$.
The momentum distributions of the eigenstates of the hydrogen atom are given by $(a_0/\hbar)\Xi_{n\ell}(a_0p/\hbar)$ with~\cite{PodolskyPauling}
\begin{equation}
\Xi_{n\ell}(u)
=
\frac{4^{2\ell+3}n^2(\ell!)^2(n-\ell-1)!}{2\pi(n+\ell)!}
\frac{(n^2u^2)^{\ell+1}}{(n^2u^2+1)^{2\ell+4}}
\left[
C_{n-\ell-1}^{(\ell+1)}\!\left(
\frac{n^2u^2-1}{n^2u^2+1}
\right)
\right]^2,
\label{eq:momentum_distribution}
\end{equation}
where $C_n^{(\alpha)}(x)$ is the Gegenbauer polynomial~\cite[Chap.~22]{Abramowitz1972}.
The radial wave functions and the momentum distributions are normalized as
\begin{equation*}
\int_0^\infty\rmd u\,u^2[\mathcal{R}_{n\ell}(u)]^2=1,\qquad
\int_0^\infty\rmd u\,\Xi_{n\ell}(u)=1.
\end{equation*}
We will use a bound on the generalized Laguerre polynomial~\cite{Koornwinder1977}
\begin{equation}
\rme^{-\frac{1}{2}x}|L_n^{(\alpha)}(x)|
\le L_n^{(\alpha)}(0)
=\frac{\Gamma(\alpha+1+n)}{\Gamma(\alpha+1)\Gamma(n+1)}\qquad
(x\ge0,\ \alpha\ge0),
\label{eqn:LaguerreBound}
\end{equation}
and a bound on the Gegenbauer polynomial~\cite[Thm.~7.33.1]{Szegoe1975}
\begin{equation}
|C_n^{(\alpha)}(x)|
\le C_n^{(\alpha)}(1)
=\frac{\Gamma(2\alpha+n)}{\Gamma(2\alpha)\Gamma(n+1)}\qquad
(-1\le x\le1,\ \alpha>0).
\label{eqn:GegenbauerBound}
\end{equation}

In this section, we set $g=0$ in Thm.~\ref{thm:trotter_thm} and define 
\begin{equation*}
\xi_{N,j}(t;\Psi_{n\ell m}) = N\left\|\left(\rme^{\frac{\rmi}{\hbar}\frac{t}{N}H_j} -I - \frac{\rmi}{\hbar}\frac{t}{N}H_j\right)\Psi_{n\ell m}\right\|
\le\frac{t^2}{2\hbar^2N}\|H_j^2 \Psi_{n\ell m}\|,\quad j=1,2,
\end{equation*}
for the kinetic and potential energies.
By computing these individually, we retrieve the Trotter bounds shown in Sec.~\ref{sec:trotter_thm}\@. 
Notice, however, that we shift the potential energy by the respective eigenvalue $h=E_n$ in $\xi_{N,2}(t;\Psi_{n\ell m})$ as we did in the proof of Thm.~\ref{thm:trotter_thm} in Appendix~\ref{sec:proof_trotter_thm}\@. 
We also define a characteristic energy scale and a time scale
\begin{equation}
\Lambda=\frac{\hbar^2}{m_\mathrm{e}a_0^2}, \qquad t_0=\frac{\hbar}{\Lambda}=\frac{m_\mathrm{e}a_0^2}{\hbar},
\label{eq:EnergyScale}
\end{equation}
and a reduced (dimensionless) time
\begin{equation*}
\tilde{t}=\frac{t}{t_0}.
\end{equation*}

\subsection{Kinetic term for $\ell\ge 2$}
For $\ell\ge2$, we have $\Psi_{n\ell m}\in\mathcal{D}(H_1^2)$ and we can just use Thm.~\ref{thm:trotter_thm}\@. 
That is, by setting $g=0$, we just have to bound $\|H_1^2 \Psi_{n\ell m}\|$ for
\begingroup
\allowdisplaybreaks
\begin{align*}
[\xi_{N,1}(t;\Psi_{n\ell m})]^2
&\le
\frac{t^4}{4\hbar^4N^2}\|H_1^2 \mathcal{R}_{n\ell}\|^2
=
\frac{\tilde{t}^{\,4}}{64N^2}
\int_0^\infty\rmd u\,
u^8
\Xi_{n\ell}(u),
\quad\ell\ge2.
\intertext{By bounding the Gegenbauer polynomial using Eq.~(\ref{eqn:GegenbauerBound}), we get}
[\xi_{N,1}(t;\Psi_{n\ell m})]^2
&\le{}
\frac{\tilde{t}^{\,4}}{64N^2}
\frac{4^{2\ell+3}n^2(\ell!)^2(n+\ell)!}{2\pi[(2\ell+1)!]^2(n-\ell-1)!}
\int_0^\infty\rmd u\,
u^8
\frac{(n^2u^2)^{\ell+1}}{(n^2u^2+1)^{2\ell+4}}
\nonumber\\
&={}
\frac{\tilde{t}^{\,4}}{16N^2}
\frac{(n+\ell)!\beta(\ell-3/2,\ell+11/2)}{n^7(n-\ell-1)![\Gamma(\ell+3/2)]^2},
\quad\ell\ge2,
\end{align*}
\endgroup
where $\beta(x,y)=\Gamma(x)\Gamma(y)/\Gamma(x+y)$ is the beta function.

\subsection{Kinetic term for $\ell=0,1$}
For $\ell=0,1$, we have diverging fourth moment $\|H_1^2 \Psi_{n\ell m}\|=\infty$, and $\Psi_{n\ell m}\not\in\mathcal{D}(H_1^2)$.
We hence need to estimate the finer bound, following Lemma~\ref{lem:FractionalScaling}\@. 
The spectral measure $\nu_{1,\Psi_{n\ell m}}$ of the kinetic-energy operator $H_1$ at $\Psi_{n\ell m}$ is absolutely continuous, 
$\rmd \nu_{1,\Psi_{n\ell m}}(\lambda)=\rho_{1,\Psi_{n\ell m}}(\lambda)\,\rmd\lambda$,
and is obtained from the momentum distribution of the eigenstate of the hydrogen atom~\eqref{eq:momentum_distribution} as
\begin{equation*}
\rho_{1,\Psi_{n\ell m}}(\lambda)\,\rmd\lambda
=
\Xi_{n\ell}(u)\,\rmd u
=
\frac{4^{2\ell+3}n^2(\ell!)^2(n-\ell-1)!}{2\pi(n+\ell)!}
\frac{(2n^2\lambda/\Lambda)^{\ell+1}}{(2n^2\lambda/\Lambda+1)^{2\ell+4}}
\left[
C_{n-\ell-1}^{\ell+1}\!\left(
\frac{2n^2\lambda/\Lambda-1}{2n^2\lambda/\Lambda+1}
\right)
\right]^2
\frac{\rmd\lambda}{\Lambda(2\lambda/\Lambda)^{1/2}},
\end{equation*}
where $\lambda=(\hbar^2/2m_\mathrm{e}a_0^2)u^2$ and
$\Lambda$ is the energy scale defined in Eq.~(\ref{eq:EnergyScale}).
By Eq.~\eqref{eqn:GegenbauerBound} and
$$
(2n^2\lambda/\Lambda)^{\ell+1}/(2n^2\lambda/\Lambda+1)^{2\ell+4}\le(2n^2\lambda/\Lambda)^{-\ell-3},
$$
we find that
\begin{equation}
\rho_{1,\Psi_{n\ell m}}(\lambda)
\le
\frac{2^{3\ell+3/2}(\ell!)^2(n+\ell)!}{\pi n^{2\ell+4}[(2\ell+1)!]^2(n-\ell-1)!}
\frac{\Lambda^{\ell+5/2}}{\lambda^{\ell+7/2}}.
\label{eqn:Tail1}
\end{equation}
This implies the decay condition of Lemma~\ref{lem:FractionalScaling} as
\begin{equation}
	\nu_{1,\Psi_{n\ell m}}\bm{(}(\lambda,+\infty)\bm{)} = \int_\lambda^\infty \rho_{1,\Psi_{n\ell m}}(x)\,\rmd x \leq \left(\frac{\lambda_1}{\lambda}\right)^{2\delta_1},
\label{eqn:lambda1}
\end{equation}
with
\begin{equation*}
\delta_1
=\frac{1}{2}\ell+\frac{5}{4}
=\begin{cases}
\medskip
\displaystyle
\frac{5}{4}, \quad&\ell=0,\\
\displaystyle
\frac{7}{4}, \quad&\ell=1.
\end{cases}
\end{equation*}
The bound~\eqref{eqn:fractional_bound} from Lemma~\ref{lem:FractionalScaling} then  yields
\begin{equation}
\Delta_1(t/N)
\le
\frac{2^{3\ell+1}(\ell!)^2(n+\ell)!}{n^{2\ell+4}[(2\ell+1)!]^2(n-\ell-1)!}
\frac{2\ell+3}{\Gamma(\ell+7/2)}
\left(
\frac{\tilde{t}}{N}
\right)^{\ell+5/2},
\label{eqn:HydrogenBound1}
\end{equation}
so that
\begin{equation*}
	\xi_{N,1}(t;\Psi_{n\ell m})
	\leq
	\frac{\tilde{t}^{\,\ell/2+5/4}}{N^{\ell/2+1/4}}
	\sqrt{
	\frac{2^{3\ell+1}(\ell!)^2(n+\ell)!}{n^{2\ell+4}[(2\ell+1)!]^2(n-\ell-1)!}
	\frac{2\ell+3}{\Gamma(\ell+7/2)}
	},
	\quad\ell=0,1.
\end{equation*}

\subsection{Potential term for $\ell\geq 1$}
For $\ell\geq 1$, we have $\Psi_{n\ell m}\in\mathcal{D}(H_2^2)$ and can therefore apply Thm.~\ref{thm:trotter_thm}\@.
We know from Ref.~\cite{Pauling1985} that for $\ell\ge0$
\begingroup
\allowdisplaybreaks
\begin{align*}
\langle\Psi_{n\ell m}|\frac{1}{r}|\Psi_{n\ell m}\rangle & =\frac{1}{a_{0}n^{2}},\\
\langle\Psi_{n\ell m}|\frac{1}{r^{2}}|\Psi_{n\ell m}\rangle & =\frac{1}{a_{0}^{2}(\ell+1/2)n^{3}},
\intertext{and for $\ell\geq1$,}
\langle\Psi_{n\ell m}|\frac{1}{r^{3}}|\Psi_{n\ell m}\rangle & =\frac{1}{a_{0}^{3}\ell(\ell+1)(\ell+1/2)n^{3}},\\
\langle\Psi_{n\ell m}|\frac{1}{r^{4}}|\Psi_{n\ell m}\rangle& =\frac{3n^2-\ell(\ell+1)}{2a_{0}^{4}\ell(\ell+1)(\ell+1/2)(\ell-1/2)(\ell+3/2)n^5}.
\end{align*}
\endgroup
It is worth remarking that the last two moments diverge when $\ell=0$.
This is because s-orbitals have high probabilities for the electron
to be close to the center, where the potential diverges. 
Using these moments, we can compute for $\ell\geq1$
\begin{align*}
\|(H_{2}-E_{n})^{2}\Psi_{n\ell m}\|^{2}
={}&\langle\Psi_{n\ell m}|(H_{2}-E_{n})^{4}|\Psi_{n\ell m}\rangle\\
={}&\langle\Psi_{n\ell m}|H_{2}^{4}|\Psi_{n\ell m}\rangle
-4E_{n}\langle\Psi_{n\ell m}|H_{2}^{3}|\Psi_{n\ell m}\rangle
\\
&{}
+6E_{n}^{2}\langle\Psi_{n\ell m}|H_{2}^{2}|\Psi_{n\ell m}\rangle
-4E_{n}^{3}\langle\Psi_{n\ell m}|H_{1}|\Psi_{n\ell m}\rangle+E_{n}^{4}
\end{align*}
to get
\begin{align*}
& \xi_{N,2}(t;\Psi_{n\ell m})\\
&\quad\le
\frac{\tilde{t}^{\,2}}{2N}\sqrt{\frac{16n^5-32(\ell-1/2)(\ell+3/2)n^3+24\ell(\ell+1)(\ell-1/2)(\ell+3/2) n-7\ell(\ell+1)(\ell+1/2)(\ell-1/2)(\ell+3/2)}{16\ell(\ell+1)(\ell+1/2)(\ell-1/2)(\ell+3/2)n^8}}.
\end{align*}

\subsection{Potential term for $\ell=0$}
For $\ell=0$, we have a diverging fourth moment $\|H_2^2\Psi_{n00}\|=\infty$, and hence $\Psi_{n00}\not\in\mathcal{D}(H_2^2)$.
Therefore, we need to estimate the finer bound from Lemma~\ref{lem:FractionalScaling}\@.
The spectral measure $\nu_{2,\Psi_{n\ell m}}$ of the potential-energy operator $H_2$ at $\Psi_{n\ell m}$ is absolutely continuous, 
$\rmd \nu_{2,\Psi_{n\ell m}}(\lambda)=\rho_{2,\Psi_{n\ell m}}(\lambda)\,\rmd\lambda$, and is obtained from the radial wave function~(\ref{eq:hydrogen_eigenfunctions}) as
\begin{equation*}
\rho_{2,\Psi_{n\ell m}}(\lambda)\,\rmd\lambda
=
[\mathcal{R}_{n\ell}(u)]^2u^2\,\rmd u
=
\left(\frac{2}{n}\right)^3
\frac{(n-\ell-1)!}{2n(n+\ell)!}
\rme^{-2\Lambda/n\lambda}
(2\Lambda/n\lambda)^{2\ell}
[L_{n-\ell-1}^{(2\ell+1)}(2\Lambda/n\lambda)]^2
(\Lambda/\lambda)^4\,\frac{\rmd\lambda}{\Lambda},
\end{equation*}
where $\lambda=\hbar^2/(m_\mathrm{e}a_0^2u)$, and $\Lambda$ is the energy scale defined in Eq.~(\ref{eq:EnergyScale}).
By bounding the Laguerre polynomial using Eq.~(\ref{eqn:LaguerreBound}),
we find
\begin{equation}
\rho_{2,\Psi_{n\ell m}}(\lambda)
\le
\frac{2^{2\ell+2}(n+\ell)!}{n^{2\ell+4}[(2\ell+1)!]^2(n-\ell-1)!}
\frac{\Lambda^{2\ell+3}}{\lambda^{2\ell+4}}.
\label{eqn:Tail2}
\end{equation}
This implies the decay condition of Lemma~\ref{lem:FractionalScaling} as 
\begin{equation}
	\nu_{2,\Psi_{n\ell m}}\bm{(}(\lambda,+\infty)\bm{)} = \int_\lambda^\infty \rho_{2,\Psi_{n\ell m}}(x)\,\rmd x \leq \left(\frac{\lambda_2}{\lambda}\right)^{2\delta_2},
\label{eqn:lambda2}
\end{equation}
with
\begin{equation*}
\delta_2
=\ell+\frac{3}{2}
=
\frac{3}{2},\quad\ell=0.
\end{equation*}
The bound~\eqref{eqn:fractional_bound} of Lemma~\ref{lem:FractionalScaling} then yields
\begin{equation}
\Delta_2(t/N)
\le
\frac{4\pi}{3n^3}
\left(
\frac{\tilde{t}}{N}
\right)^3,\quad\ell=0.
\label{eqn:HydrogenBound2}
\end{equation}
Thus, $\xi_{N,2}(t,\Psi_{n00})$ is bounded by
\begin{align*}
\xi_{N,2}(t;\Psi_{n00})
\le{}&
\sqrt{\frac{4\pi}{3n^3}} \frac{\tilde{t}^{\,3/2}}{N^{1/2}},\quad\ell=0.
\end{align*}

\subsection{Potential term with a shift}
According to Thm.~\ref{thm:trotter_thm}, we need to shift the potential Hamiltonian $H_2$ by the eigenvalue $E_n$ of the respective input eigenstate $\Psi_{n\ell m}$ to get the bound on the Trotter error.
This can be done as follows
\begingroup
\allowdisplaybreaks
\begin{align}
\xi_{N,j}(t;\Psi_{n\ell m})
={}& 
N\left\|\left(\rme^{\frac{\rmi}{\hbar}\frac{t}{N}(H_j-g_j)} -I - \frac{\rmi}{\hbar}\frac{t}{N}(H_j-g_j)\right)\Psi_{n\ell m}\right\|
\nonumber\\
={}& 
N\left\|
\left[
\left(
\rme^{\frac{\rmi}{\hbar}\frac{t}{N}H_j} -I 
-\frac{\rmi}{\hbar}\frac{t}{N}H_j
\right)
+
\left(
\rme^{\frac{\rmi}{\hbar}\frac{t}{N}(H_j-g_j)}-\rme^{\frac{\rmi}{\hbar}\frac{t}{N}H_j}
+
\frac{\rmi}{\hbar}\frac{t}{N}g_j
\right)
\right]
\Psi_{n\ell m}\right\|
\nonumber\\
={}&
N\left\|
\left[
\left(
\rme^{\frac{\rmi}{\hbar}\frac{t}{N}H_j} -I 
-\frac{\rmi}{\hbar}\frac{t}{N}H_j
\right)
+
\left(
\rme^{\frac{\rmi}{\hbar}\frac{t}{N}H_j}
(
\rme^{-\frac{\rmi}{\hbar}\frac{t}{N}g_j}
-1
)
+
\frac{\rmi}{\hbar}\frac{t}{N}g_j
\right)
\right]
\Psi_{n\ell m}\right\|
\nonumber\\
={}&
N\,\biggl\|
\left(
\rme^{\frac{\rmi}{\hbar}\frac{t}{N}H_j} -I 
-\frac{\rmi}{\hbar}\frac{t}{N}H_j
\right)
\Psi_{n\ell m}
{}+
\left[
\rme^{\frac{\rmi}{\hbar}\frac{t}{N}H_j}
\left(
\rme^{-\frac{\rmi}{\hbar}\frac{t}{N}g_j}
-1
+
\frac{\rmi}{\hbar}\frac{t}{N}g_j
\right)
-
\frac{\rmi}{\hbar}\frac{t}{N}g_j
(
\rme^{\frac{\rmi}{\hbar}\frac{t}{N}H_j}
-1
)
\right]
\Psi_{n\ell m}
\biggr\|
\nonumber\\
\le{}&
N\left\|
\left(
\rme^{\frac{\rmi}{\hbar}\frac{t}{N}H_j} -I 
-\frac{\rmi}{\hbar}\frac{t}{N}H_j
\right)
\Psi_{n\ell m}
\right\|
+
N
\left|
\rme^{-\frac{\rmi}{\hbar}\frac{t}{N}g_j}
-1
+
\frac{\rmi}{\hbar}\frac{t}{N}g_j
\right|
+
\frac{t}{\hbar}|g_j|
\|
(
\rme^{\frac{\rmi}{\hbar}\frac{t}{N}H_j}
-I
)
\Psi_{n\ell m}
\|
\nonumber\\
\le{}&
N\left\|
\left(
\rme^{\frac{\rmi}{\hbar}\frac{t}{N}H_j} -I 
-\frac{\rmi}{\hbar}\frac{t}{N}H_j
\right)
\Psi_{n\ell m}
\right\|
+
\frac{g_j^2t^2}{2N\hbar^2}
+
\frac{|g_j|t^2}{N\hbar^2}\|H_j\Psi_{n\ell m}\|.
\label{eq:first_order_with shift}
\end{align}
\endgroup
For instance, for the potential term $H_2$ with the shift $g_2=E_n=-\hbar^2/2m_\mathrm{e}a_0^2n^2$, we get
\begin{align*}
\xi_{N,2}(t;\Psi_{n\ell m})
={}&
N\left\|
\left(
\rme^{\frac{\rmi}{\hbar}\frac{t}{N}(H_2-E_n)} -I 
-\frac{\rmi}{\hbar}\frac{t}{N}(H_2-E_n)
\right)
\Psi_{n\ell m}
\right\|
\nonumber\\
\le{}&
\begin{cases}
\medskip
\displaystyle
\frac{\tilde{t}^{\,3/2}}{N^{1/2}}
\sqrt{
\frac{4\pi}{3n^3}
}
+
\frac{\tilde{t}^{\,2}}{2N}
\left(
\sqrt{\frac{2}{n^{7/2}}}
+
\frac{1}{4n^4}
\right),
&\ell=0,\\
\displaystyle
\frac{\tilde{t}^{\,2}}{2N}
\left(
\sqrt{
\frac{3-\ell(\ell+1)/n^2}{2n^3\ell(\ell+1)(\ell+1/2)(\ell+3/2)(\ell-1/2)}
}
+
\frac{1}{n^{7/2}\sqrt{\ell+1/2}}
+
\frac{1}{4n^4}
\right),
&\ell\ge1.
\end{cases}
\end{align*}
These yield the bounds~\eqref{eq:bound_R_n0}--\eqref{eq:bound_R_nl} presented in Sec.~\ref{sec:trotter_thm} of the Methods section.

\subsection{Hydrogen ground state}
A particularly interesting case is the ground state of the hydrogen atom, $n=1$, $\ell=0$, for which we have
\begin{equation}
	\xi_{N,1}(t;\Psi_{100})\le \frac{\tilde{t}^{\,5/4}}{N^{1/4}}
\frac{4}{\sqrt{5\sqrt{\pi}}},
\label{eqErrorBoundhydrogenGround1}
\end{equation}
and
\begin{equation}
	\xi_{N,2}(t;\Psi_{100})\le 
	\frac{\tilde{t}^{\,3/2}}{N^{1/2}}
	\sqrt{
\frac{4\pi}{3}
}
+
\frac{\tilde{t}^{\,2}}{2N}
\left(
\sqrt{2}
+
\frac{1}{4}
\right).
\label{eqErrorBoundhydrogenGround2}
\end{equation}
Notice that the total error is the sum of both. As we can see, the dominating part in the bound is of the order $\mathcal{O}(N^{-1/4})$.

\subsection{Remark on the tightness of the bounds obtained from Lemma~\ref{lem:FractionalScaling}}
For these hydrogen examples, the bounds~\eqref{eqn:HydrogenBound1} and~\eqref{eqn:HydrogenBound2} on $\Delta_j(s)$ are asymptotically tight for $s\to0$.
Indeed, the uniform bounds~\eqref{eqn:Tail1} and~\eqref{eqn:Tail2} on the spectral densities $\rho_{j,\Psi_{n\ell m}}(\lambda)$ allow us to apply Lebesgue's dominated convergence theorem to $\Delta_j(s)/s^{2\delta_j}$ as
\begin{align*}
\lim_{s\to0}\frac{1}{s^{2\delta_j}}\Delta_j(s)
&=
\lim_{s\to0}
\int_0^\infty f(x)
\frac{\rho_{j,\varphi}(x/s)}{s^{2\delta_j+1}} 
\,\rmd x
\nonumber\\
&=
\int_0^\infty f(x)
\lim_{s\to0}
\frac{\rho_{j,\varphi}(x/s)}{s^{2\delta_j+1}} 
\,\rmd x
\nonumber\\
&=
2\delta_j (\lambda_j)^{2\delta_j}\int_0^\infty 
\frac{f(x)}{x^{2\delta_j+1}} 
\,\rmd x
\nonumber\\
&=
\frac{\pi(\lambda_j)^{2\delta_j}}{\Gamma(2\delta_j-1)\sin[(\delta_j-1)\pi]},
\end{align*}
where $\lambda_j$ and $\delta_j$ ($j=1,2$) are introduced in Eqs.~(\ref{eqn:lambda1}) and~(\ref{eqn:lambda2}).
This shows that $\Delta_j(s)$ asymptotically approach the bounds.
Nevertheless, this does not necessarily lead to a tight bound on the convergence speed of the corresponding Trotter formula.

\section{Higher-order Trotter bounds and proofs of Thm.~\ref{thm:2nd_order}, Thm.~\ref{thm:second_order_alpha}, and Thm.~\ref{thm:loose_bound_higher_order}}
\label{sec:higher_order_derivation}
In this appendix, we show how to obtain state-dependent bounds for arbitrary higher-order Trotter product formulas.
We develop a general formalism, which we then apply to the first- and second-order Trotterizations (Thm.~\ref{thm:trotter_thm} and Thm.~\ref{thm:2nd_order}).
Furthermore, we show how to obtain a loose state-dependent error bound for a general $p$th-order product formula (Thm.~\ref{thm:loose_bound_higher_order}).
This follows the proofs of the error bounds with fractional scalings (Thm.~\ref{thm:first_order_alpha} and Thm.~\ref{thm:second_order_alpha}).
These proofs are based on the formalism presented before, and our proof method generalizes iteratively to the $p$th-order product formulas with errors scaling between $N^{-q+1}$ and $N^{-q}$, $q\in(1,p)$.
We explicitly carry out this computation for the cases $q=1$ and $q=2$.

\subsection{General formalism}
A general $p$th-order Trotter product, which aims to approximate the evolution $\rme^{-\rmi t(H_1+H_2)}$, is given in the form
\begin{equation*}
	\mathcal{S}_N^{(p)}(t)=\left(\rme^{-\rmi\frac{t}{N}\tau_MH_{\sigma(M)}}\cdots \rme^{-\rmi\frac{t}{N}\tau_2H_{\sigma(2)}}\rme^{-\rmi\frac{t}{N}\tau_1H_{\sigma(1)}}\right)^N,
\end{equation*}
with $\sum_{i=1}^{\lceil M/2\rceil}\tau_{2i-1}=\sum_{i=1}^{\lfloor M/2\rfloor}\tau_{2i}=1$ and $\sigma(j)$ defined in Eq.~(\ref{eq:parity}).
We wish to estimate the error of this approximation on an input eigenstate $\varphi$ of $H_1+H_2$.
Without loss of generality, we shift the Hamiltonians $H_1$ and $H_2$ so that we have $(H_1+H_2)\varphi=0$.
Then, the Trotter product $\mathcal{S}_N^{(p)}(t)$ is to be compared with the identity $I$ on the input eigenstate $\varphi$.
As in the proofs of Thm.~\ref{thm:trotter_thm} and Thm.~\ref{thm:first_order_alpha}, it suffices to study a single Trotter cycle, 
\begin{equation}
	\mathcal{S}_1^{(p)}(t/N)=\rme^{-\rmi\frac{t}{N}\tau_MH_{\sigma(M)}}\cdots \rme^{-\rmi\frac{t}{N}\tau_2H_{\sigma(2)}}\rme^{-\rmi\frac{t}{N}\tau_1H_{\sigma(1)}},
	\label{eq:1_trotter_cycle_appendix}
\end{equation}
instead of the full Trotter evolution.
This is because we can use a telescoping sum to obtain a bound on $\mathcal{S}_N^{(p)}(t)$ from a bound on $\mathcal{S}_1^{(p)}(t/N)$, that is
\begin{equation*}
	\mathcal{S}_N^{(p)}(t) - I =
	[\mathcal{S}_1^{(p)}(t/N)]^N - I = \sum_{k=0}^{N-1} [\mathcal{S}_1^{(p)}(t/N)]^k [\mathcal{S}_1^{(p)}(t/N) - I] ,
\end{equation*}
whence
\begin{equation}
\xi_N^{(p)}(t)=\Vert[\mathcal{S}_N^{(p)}(t)-I] \varphi\Vert \le N\Vert[\mathcal{S}_1^{(p)}(t/N)-I] \varphi\Vert
=N\xi_1^{(p)}(t/N)
	\label{eq:telescope_higher_order}
\end{equation}
follows from the triangle inequality and the unitary equivalence of the Euclidean norm.

The idea of our approach is to treat the Trotterized evolution $\mathcal{S}_1^{(p)}(t/N)$ as a unitary evolution
\begin{equation*}
	U^{(p)}(s)=\T\exp\!\left(
	-\rmi\int_0^s\rmd u\,H^{(p)}(u)
	\right)
\end{equation*}
generated by a piecewise-constant time-dependent Hamiltonian $H^{(p)}(s)$~\cite{Burgarth2022}.
For instance, in the case of the first-order Trotter product, we regard it as a unitary evolution generated by a time-dependent Hamiltonian $H^{(1)}(s)$ in which the Hamiltonian is switched as~\cite{Burgarth2023}
\begin{equation*}
	H^{(1)}(s)=\begin{cases}
	\medskip
	\displaystyle
		H_1,&s\in\left[0,\frac{t}{N}\right),\\
	\displaystyle
		H_2,&s\in\left[\frac{t}{N},\frac{2t}{N}\right).
	\end{cases}
\end{equation*}
In the case of the second-order Trotter product, 
\begin{equation*}
	H^{(2)}(s)=\begin{cases}
	\medskip
	\displaystyle
		\frac{1}{2} H_1,&s\in\left[0,\frac{t}{N}\right),\\
	\medskip
	\displaystyle
		H_2,&s\in\left[\frac{t}{N},\frac{2t}{N}\right),\\
	\displaystyle
		\frac{1}{2} H_1,&s\in\left[\frac{2t}{N},\frac{3t}{N}\right).
	\end{cases}
\end{equation*}
In the case of a general $p$th-order Trotter product, some of the $\tau_j$ can be negative.
In fact, Suzuki showed that every $p$th-order product formula of order $p\geq3$ has to have at least one negative $\tau_j$~\cite[Thm.~3]{Suzuki1991}.
If $\tau_j$ is negative, the Hamiltonian $H^{(p)}(s)$ in the $j$th time slot is
$\tau_j H_{\sigma(j)}= -|\tau_j| H_{\sigma(j)}$ for time duration
$t/N$, and we have
\begin{equation}
	H^{(p)}(s)=\begin{cases}
	\medskip
	\displaystyle
		\tau_1 H_{\sigma(1)},&s\in\left[0,\frac{t}{N}\right),\\
	\medskip
	\displaystyle
		\cdots&\cdots\\
	\medskip
	\displaystyle
		\tau_j H_{\sigma(j)},&s\in\left[\frac{(j-1)t}{N},\frac{jt}{N}\right),\\
	\medskip
	\displaystyle
		\cdots&\cdots\\
	\displaystyle
		\tau_M H_{\sigma(M)},&s\in\left[\frac{(M-1)t}{N},\frac{Mt}{N}\right).
	\end{cases}
\label{eq:piecesise-constant_Hamiltonian}
\end{equation}
Then, the evolution $U^{(p)}(s)$ generated by this piecewise-constant time-dependent Hamiltonian $H^{(p)}(s)$ yields  
$U^{(p)}(M t/N)=\mathcal{S}_1^{(p)}(t/N)$ at time $s=T_M=M t/N$.
The Hamiltonian $H^{(p)}(s)$ is extended periodically from one Trotter cycle to the full Trotter product by $H^{(p)}(s+T_M)=H^{(p)}(s)$.

To bound the state-dependent Trotter error $\xi_1^{(p)}(t/N)$, we use the following Dyson expansion up to the $q$th order.
\begin{lemma}[Remainder of Dyson expansion]\label{lem:error_rep_higher_order}
	Let  $t \mapsto H(t)$ be a piecewise-constant Hamiltonian in $[0,s]$.
	That is, there exists a partition $0=t_0<t_1<\dots<t_L=s$
of $[0,s]$, such that $H(t)=h_\ell$ for $t\in [t_{\ell-1},t_\ell)$, for $\ell=1,\ldots,L$ [and $H(s)=h_L$], with $h_\ell=h_\ell^\dag$. Let 
\begin{equation*}
U(u)=\T\exp\!\left(-\rmi\int_0^u\rmd s_1\,H(s_1)\right),
\end{equation*}
for $u\in[0,s]$, be the unitary propagator generated by $H(t)$. 

Let $q$ be a positive integer, and $\varphi\in\mathcal{D}(h_{\ell_1}h_{\ell_2} \cdots  h_{\ell_q} h_{\ell_{q+1}})$ for all $L\geq \ell_1\geq \ell_2 \geq \cdots \geq \ell_{q+1} \geq 1$.
Define the $k$th-order integral action on $\varphi$,
\begin{equation}
	S_k(u) \varphi =\int_0^u\rmd s_1\cdots\int_0^{s_{k-1}}\rmd s_k\, H(s_1)\cdots H(s_k) \varphi,\quad k=1,\ldots,q,
	\label{eq:action_appendix}
\end{equation}
for all $u\in[0,s]$.
Then, one has
\begin{align}
U(s)\varphi=\varphi+\sum_{k=1}^q(-\rmi)^kS_k(s)\varphi+(-\rmi)^{q+1}\int_0^s\rmd u\,U(s)U^\dag(u)H(u)S_q(u)\varphi.
\label{eq:lemma_error_operator}
\end{align}
\end{lemma}
\begin{remark}[Remainder of Taylor expansion]\label{rmk:TaylorExpansion}
For a constant Hamiltonian $H$, this lemma provides a remainder term of the Taylor expansion of the evolution operator, for $\varphi\in \mathcal{D}(H^{q+1})$:
\begin{equation}
\rme^{-\rmi sH}\varphi
=\varphi+\sum_{k=1}^q\frac{1}{k!}(-\rmi sH)^k\varphi
+(-\rmi)^{q+1}\int_0^s\rmd u\,\rme^{-\rmi(s-u)H}\frac{1}{q!}u^qH^{q+1}\varphi.
\label{eq:TaylorExpansion}
\end{equation}
It is already used in Eq.~(\ref{eq:Taylor_first}) and will be used later.
\end{remark}
We give two different derivations of the Dyson expansion~(\ref{eq:lemma_error_operator}).
The first is a generalization of the integration-by-part strategy developed in Ref.~\cite{Burgarth2022}.
\begin{proof}[Proof 1 of Lemma~\ref{lem:error_rep_higher_order}]
Notice first that, by assumption, $H(s_1)\cdots H(s_k) \varphi$ is well defined and piecewise continuous in all variables for all $k=1,\ldots,q$, whence the integral actions $S_k(u)\varphi$ are piecewise differentiable. Moreover, $\varphi \in \mathcal{D}\bm{(}H(u) S_k(u)\bm{)}$ for all $u\in[0,s]$.
The difference between $U(s)$ and the identity $I$ on $\varphi$ can be estimated by~\cite{Burgarth2022}
\begin{equation*}
[U(s)-I]\varphi
=-\biggl[
U(s)U^\dag(u)
\biggr]_{u=0}^{u=s}
\varphi
=-\int_0^s\rmd u\,
\frac{\partial}{\partial u}[
U(s)U^\dag(u)
]
\varphi
=-\rmi\int_0^s\rmd u\,
U(s)U^\dag(u)H(u)\varphi,
\end{equation*}
where we have used that $U^\dag(u)\varphi$ is piecewise continuously differentiable with  
\begin{equation*}
\frac{\rmd}{\rmd u}U^\dag(u)\varphi=\rmi U^\dag(u)H(u)\varphi.
\end{equation*}
Since $H(u)\varphi =\dot{S}_1(u)\varphi$ except at $u=t_k$, by integrating by parts one gets
\begin{align*}
[U(s)-I]\varphi
&=-\rmi\int_0^s\rmd u\,
U(s)U^\dag(u)\dot{S}_1(u)\varphi
\nonumber\\
&=-\rmi S_1(s)\varphi
+(-\rmi)^2\int_0^s\rmd u\,
U(s)U^\dag(u)H(u)S_1(u)\varphi.
\nonumber
\intertext{By iterating the integration by parts $(q-1)$ more times, we get}
[U(s)-I]\varphi
&=
-\rmi S_1(s)\varphi
+(-\rmi)^2\int_0^s\rmd u\,
U(s)U^\dag(u)\dot{S}_2(u)\varphi
\nonumber\\
&=
-\rmi S_1(s)\varphi
+(-\rmi)^2S_2(s)\varphi
+(-\rmi)^3\int_0^s\rmd u\,
U(s)U^\dag(u)H(u)S_2(u)\varphi
\nonumber\\
&=\cdots
\nonumber\\
&=
-\rmi S_1(s)\varphi
+\cdots
+(-\rmi)^qS_q(s)\varphi
+(-\rmi)^{q+1}\int_0^s\rmd u\,
U(s)U^\dag(u)H(u)S_q(u)\varphi.
\nonumber
\end{align*}
This proves the lemma.
\end{proof}
An alternate derivation of Lemma~\ref{lem:error_rep_higher_order} is to solve the backward Schr\"odinger equation
\begin{equation}
\frac{\partial}{\partial u}U(s,u) \varphi =\rmi U(s,u)H(u)\varphi,\qquad
U(s,s)=I,
\label{eq:BackwardSchrodingerEq}
\end{equation}
 for $U(s,u)=U(s)U^\dag(u)$, iteratively.
\begin{proof}[Proof 2 of Lemma~\ref{lem:error_rep_higher_order}]
The derivative of the unitary propagator on $\varphi$ in Eq.~\eqref{eq:BackwardSchrodingerEq} is piecewise continuous. By integrating the Schr\"odinger equation~\eqref{eq:BackwardSchrodingerEq}, one has
\begin{equation*}
U(s,u)\varphi=\varphi-\rmi\int_u^s\rmd u_1\,U(s,u_1)H(u_1)\varphi.
\end{equation*}
By iteration,
\begin{align*}
U(s,u)\varphi
={}&\varphi
-\rmi\int_u^s\rmd u_1\,H(u_1)\varphi
+(-\rmi)^2\int_u^s\rmd u_1\int_{u_1}^s\rmd u_2\,U(s,u_2)H(u_2)H(u_1)\varphi\\
={}&\cdots\\
={}&\varphi
+\sum_{k=1}^q
(-\rmi)^k\int_u^s\rmd u_1\cdots\int_{u_{k-1}}^s\rmd u_k\,H(u_k)\cdots H(u_1)\varphi
\nonumber\\
&\hphantom{\varphi}
{}+(-\rmi)^{q+1}\int_u^s\rmd u_1\cdots\int_{u_q}^s\rmd u_{q+1}\,U(s,u_{q+1})H(u_{q+1})\cdots H(u_1)\varphi.
\intertext{By reverting the order of the integrations,}
U(s,u)\varphi
={}&\varphi
+\sum_{k=1}^q
(-\rmi)^k\int_u^s\rmd u_k\int_u^{u_k}\rmd u_{k-1}\cdots\int_u^{u_2}\rmd u_1\,H(u_k)\cdots H(u_1)\varphi
\nonumber\\
&\hphantom{\varphi}
{}+(-\rmi)^{q+1}\int_u^s\rmd u_{q+1}\int_u^{u_{q+1}}\rmd u_q\cdots\int_u^{u_2}\rmd u_1\,U(s,u_{q+1})H(u_{q+1})\cdots H(u_1)\varphi.
\end{align*}
By setting $u=0$, we get Eq.~(\ref{eq:lemma_error_operator}).
\end{proof}
Notice that the use of the backward Schr\"odinger equation is essential in our argument. As a matter of fact, even the forward Schr\"odinger equation itself 
\begin{align*}
\frac{\rmd}{\rmd s}U(s)\varphi=-\rmi H(s)U(s)\varphi	
\end{align*} 
does not hold, unless one requires additional strong conditions on the operators $h_k$: that they have a common domain $D$, which is left invariant under their own unitary groups, and that $\varphi\in D$. 
Moreover, in order to iterate it one needs even stronger assumptions. In any case, one would finally obtain by iteration
\begin{align*}
U(s)\varphi={}&
\varphi-\rmi\int_0^s\rmd s_1\,H(s_1)U(s_1)\varphi
\\
={}&
\varphi
+\sum_{k=1}^q(-\rmi)^kS_k(s)\varphi
+(-\rmi)^{q+1}\int_0^s\rmd s_1\cdots\int_0^{s_q}\rmd s_{q+1}\,H(s_1)\cdots H(s_{q+1})U(s_{q+1})\varphi.
\end{align*}
This leads to a different expression for the remainder term, which is not convenient for our purpose.
Indeed, the presence of the unitary $U(s_{q+1})$ on $\varphi$ before the action of the Hamiltonians makes it  difficult to bound the remainder term.
On the other hand, if the unitary comes last in the remainder term as in the expansion~\eqref{eq:lemma_error_operator}, it is harmless in bounding the remainder term, as we will see below.

We use the Dyson expansion of Lemma~\ref{lem:error_rep_higher_order} for the Trotter evolution $U^{(p)}(s)$ generated by the piecewise-constant Hamiltonian $H^{(p)}(s)$ in Eq.~(\ref{eq:piecesise-constant_Hamiltonian}), and estimate the Trotter error $\xi_1^{(p)}(t/N)=\|[U^{(p)}(T_M)-I]\varphi\|$ for the single Trotter cycle.
\begin{lemma}\label{lem:bound_pth_order}
Let $U^{(p)}(s)$ be the unitary propagator generated by the piecewise-constant Hamiltonian $H^{(p)}(s)$ in Eq.~(\ref{eq:piecesise-constant_Hamiltonian}) representing a $p$th-order Trotterized evolution with $M$ time slots in a single Trotter cycle.
The unitary reaches
$U^{(p)}(Mt/N)=\mathcal{S}_1^{(p)}(t/N)$ in Eq.~\eqref{eq:1_trotter_cycle_appendix} at time $s=T_M=Mt/N$.

Let $q$ be a positive integer $q\le p$. Let the input state $\varphi$ satisfy $(H_1+H_2)\varphi=0$, and 
$\varphi\in\mathcal{D}\bm{(}H^{(p)}(s_1)H^{(p)}(s_2)\cdots H^{(p)}(s_{q+1})\bm{)}$ for all $T_M> s_1 \geq s_2 \geq \cdots \geq s_{q+1}\geq 0$.
Then, one has
\begin{equation}
[U^{(p)}(T_M)-I]\varphi
=R^{(p)}_q \varphi
=(-\rmi)^{q+1}\int_0^{T_M}\rmd s\,U^{(p)}(T_M)U^{(p)\dag}(s)H^{(p)}(s)S_q(s)\varphi,
\label{eqn:Dyson_error_q}
\end{equation}
where $S_q(s)\varphi$ is the $q$th-order integral action of the piecewise-constant Hamiltonian $H^{(p)}(s)$
 defined in Eq.~(\ref{eq:action_appendix}).

If $q=p$, the error $\xi_N^{(p)}(t)$ of the overall $p$th-order Trotter product $\mathcal{S}_N^{(p)}(t)$ on the input state $\varphi$ is bounded by
\begin{equation}
\xi_N^{(p)}(t)
\le
N\|[U^{(p)}(T_M)-I]\varphi\|
\le 
N\sum_{j=1}^M \int_{(j-1)t/N}^{jt/N}\rmd s\,\Vert \tau_j H_{\sigma(j)}S_p(s)\varphi\Vert,
\label{eqn:error_estimate_p}
\end{equation}
which is $\mathcal{O}(t^{p+1}/N^p)$.
\end{lemma}
\begin{proof}
Set $T_j=jt/N$. Notice first that the integral action $S_k(T_M)\varphi$ in Eq.~(\ref{eq:action_appendix}) for the piecewise-constant Hamiltonian $H^{(p)}(s)$ in Eq.~(\ref{eq:piecesise-constant_Hamiltonian}) is proportional to $t^k/N^k$.
Since a $p$th-order Trotter product $\mathcal{S}_1^{(p)}(t/N)=U^{(p)}(T_M)$ aims to approximate $\rme^{-\rmi\frac{t}{N}(H_1+H_2)}$ at $\mathcal{O}(t^{p+1}/N^{p+1})$, it formally requires that
\begin{equation}
S_k(T_M)=\frac{1}{k!}\frac{t^k}{N^k}(H_1+H_2)^k,\quad
k=1,\ldots,p.
\label{eq:action_zero}
\end{equation}
We hence have
\begin{equation*}
S_k(T_M)\varphi=0,\quad
k=1,\ldots,q,
\end{equation*}
and the Dyson expansion of Lemma~\ref{lem:error_rep_higher_order} with $q\le p$ yields Eq.~(\ref{eqn:Dyson_error_q}).
If $q=p$, then $\varphi\in\mathcal{D}\bm{(}H^{(p)}(s)S_p(s)\bm{)}$ for all $s\in[0,T_M)$, and the single-cycle Trotter error $\xi_1^{(p)}(t/N)$ is bounded by
\begin{equation*}
\xi_1^{(p)}(t/N)
=
\|[U^{(p)}(T_M)-I]\varphi\|
\le\int_0^{T_M}\rmd s\,\|H^{(p)}(s)S_p(s)\varphi\|
=
\sum_{j=1}^M \int_{T_{j-1}}^{T_j}\rmd s\, \Vert \tau_j H_{\sigma(j)}S_p(s)\varphi\Vert.
\end{equation*}
Thus, Eq.~\eqref{eq:telescope_higher_order} gives the bound~(\ref{eqn:error_estimate_p}).
\end{proof}

In the case of unbounded operators,
Eq.~(\ref{eqn:Dyson_error_q}) might only be well defined for some $q<p$.
In such a case, we have to perform a more sophisticated estimation of the error, instead of simply bounding Eq.~(\ref{eqn:Dyson_error_q}) by triangle inequality as in Eq.~(\ref{eqn:error_estimate_p}).
See the following subsections of this appendix.

In any case, to proceed, we need to explicitly compute the integral actions $S_k(s)$ defined in Eq.~\eqref{eq:action_appendix} for the piecewise-constant Hamiltonian $H^{(p)}(s)$ in Eq.~(\ref{eq:piecesise-constant_Hamiltonian}) for a $p$th-order Trotter product. We proceed formally, by temporarily assuming that all operators are bounded. For unbounded operators the expressions obtained will be valid when applied to a vector belonging to a suitable domain.

For $s\in[(j-1)t/N,jt/N)$ in the $j$th time slot, we have
	\begin{align*}
		S_k(s)
		&
		=\int_0^s\rmd s_1\cdots\int_0^{s_{k-1}}\rmd s_k\,
		H^{(p)}(s_1)\cdots H^{(p)}(s_k)
		\nonumber\\
		&
		=
		\sum_{j\geq j_1 \geq \dots \geq j_k \geq 1 }
		\int_{(j_1-1)t/N}^{\tilde{T}_{j_1}}\rmd s_1\cdots\int_{(j_k-1)t/N}^{\tilde{T}_{j_k}}\rmd s_k\,
		\tau_{j_1}H_{\sigma(j_1)}\cdots
		\tau_{j_k} H_{\sigma(j_k)},\quad
		s\in\left[\frac{(j-1)t}{N},\frac{j t}{N}\right),
\intertext{where $\tilde{T}_{j_\ell}=j_\ell t/N$ if $j_{\ell-1}>j_{\ell}$,
or $\tilde{T}_{j_\ell}=s_{\ell-1}$ if $j_{\ell-1}=j_\ell$,
 for $\ell=2,\ldots,k$, and $\tilde{T}_{j_1}=j_1 t/N$ if $j>j_1$, or $\tilde{T}_{j_1}=s$ if $j_1=j$.
By shifting the time integration in each slot, we get}
		S_k(s)
		&=
		\sum_{j\geq j_1 \geq \dots \geq j_k \geq 1 }
		\int_0^{s_{j_1}}\rmd s_1\cdots\int_0^{s_{j_k}}\rmd s_k\,
		\tau_{j_1} H_{\sigma(j_1)}\cdots
		\tau_{j_k} H_{\sigma(j_k)},\quad
		s\in\left[\frac{(j-1)t}{N},\frac{j t}{N}\right),
	\end{align*}
	where
	\begin{equation*}
		s_{j_\ell}
		=
		\begin{cases}
			\displaystyle
			\medskip
			s-\frac{(j-1)t}{N}, &\mathrm{if}\ \ell=1\ \mathrm{and}\ j_1=j,
			\\
			\displaystyle
			\medskip
			s_{\ell-1},&\mathrm{if}\ \ell\in\{2,\ldots,k\}\ \mathrm{and}\ j_\ell=j_{\ell-1},
			\\
			\displaystyle
			\frac{t}{N}, &\mathrm{otherwise}.
		\end{cases}
	\end{equation*}
	This yields
\begin{equation*}
	S_k(s)
		= \left(\frac{t}{N}\right)^k
		\sum_{j\geq j_1 \geq \dots \geq j_k \geq 1 } \frac{1}{m(j_1,\dots,j_k)}
				(\tilde{\tau}_{j_1}H_{\sigma(j_1)})\cdots
		(\tilde{\tau}_{j_k} H_{\sigma(j_k)}),\quad
		s\in\left[\frac{(j-1)t}{N},\frac{j t}{N}\right),
\end{equation*}
where $\tilde{\tau}_{i}=\tau_i$ if $i< j$ and $\tilde{\tau}_{j}=\tau_j [s-(j-1)t/N]/{(t/N)}$, and $m(j_1,\dots,j_k)$ is a multiplicity factor which for each group of equal indices $j_\ell=j_{\ell+1}=\dots=j_{\ell+\beta-1}$ counts the factorial of its cardinality $\beta!$. By gathering together each group of equal indices, we finally get
\begin{equation}	
		S_k(s)= \left(
		\frac{t}{N}
		\right)^k
		\mathop{\sum_{\bm{\beta}\in\mathbb{N}^j}}_{|\bm{\beta}|=k}
		\frac{1}{\bm{\beta}!}
		\left(\frac{s-(j-1)t/N}{t/N}\right)^{\beta_j}
		(\tau_{j}H_{\sigma(j)})^{\beta_j}
		\cdots
		(\tau_1H_1)^{\beta_1},
		\quad
		s\in\left[\frac{(j-1)t}{N},\frac{j t}{N}\right),
\label{eq:kth_action_explicit_expanded0}
\end{equation}
with $\bm{\beta}=(\beta_1,\dots,\beta_j)$ being a multi-index of nonnegative integers, $|\bm{\beta}|=\beta_1+\cdots+\beta_j$, and $\bm{\beta}!=\beta_1!\cdots\beta_j!$. This is Eq.~\eqref{eq:kth_action_explicit}. For $s=T_M$, we obtain
\begin{equation*}
S_k(T_M)
=
\left(
\frac{t}{N}
\right)^k
\mathop{\mathop{\sum}_{\bm{\beta}\in \mathbb{N}^M}}_{|\bm{\beta}|=k}
\frac{1}{\bm{\beta}!}
(\tau_M H_{\sigma(M)})^{\beta_M}
\cdots
(\tau_1H_1)^{\beta_1},
\end{equation*}
which is Eq.~(\ref{eq:kth_action_explicit_TM}).

\begin{remark}
	The usual strategy in the literature to obtain higher-order product formulas is by demanding that all terms up to order $p$ in the Taylor expansion of $U^{(p)}(T_M)$ vanish, see e.g.\ Refs.~\cite{Morales2022, Yoshida1990}.
	That is, one requires
	\begin{equation}
		U^{(p)}(T_M) = I + \sum_{k=1}^p \frac{(-\rmi)^k}{k!}\frac{t^k}{N^k}(H_1+H_2)^k +\mathcal{R}_{p+1}(T_M),
		\label{eq:Taylor_expansion}
	\end{equation}
	where the remainder $\mathcal{R}_{p+1}(T_M)$ is of order $\mathcal{O}(t^{p+1}/N^{p+1})$.
	This requirement gives a system of equations for the switching coefficients $\tau_k$ to give rise to a $p$th-order product formula.
	However, these equations are equivalent to the set of equations obtained by our approach through Eq.~\eqref{eq:action_zero}\@.
	This can be shown very easily.
	It is obvious that the requirement in Eq.~\eqref{eq:action_zero} gives a $p$th-order product formula.
	In fact, we showed that the overall Trotter error is of order $\mathcal{O}(t^{p+1}/N^p)$ under this assumption.
	Thus, it remains to show that Eq.~\eqref{eq:Taylor_expansion} implies Eq.~\eqref{eq:action_zero}\@.
	Combining Eq.~\eqref{eq:Taylor_expansion} with Lemma~\ref{lem:error_rep_higher_order} gives
	\begin{equation*}
		\sum_{k=1}^p(-\rmi)^kS_k(T_M)+(-\rmi)^{p+1}\int_0^{T_M}\rmd u\,U(s)U^\dag(u)H(u)S_p(u)
		=
		\sum_{k=1}^p \frac{(-\rmi)^k}{k!}\frac{t^k}{N^k}(H_1+H_2)^k +\mathcal{R}_{p+1}(T_M).
	\end{equation*}
	However, on the left-hand side, each $S_k(T_M)=\mathcal{O}(t^{k}/N^k)$.
	At the same time, each term $\frac{(-\rmi)^k}{k!}\frac{t^k}{N^k}(H_1+H_2)^k=\mathcal{O}(t^{k}/N^k)$ on the right-hand side.
	Therefore, these terms have to agree with each other order by order and we retrieve Eq.~\eqref{eq:action_zero}\@.
	Nevertheless, as opposed to a Taylor expansion, our approach allows us to explicitly compute state-dependent error bounds.
	This also allows us to compare how well different product formulas of the same order perform.
\end{remark}

In Appendix~\ref{sec:Mathematica}, we provide a Mathematica script that automatically computes the bound in Eq.~\eqref{eqn:error_estimate_p}, given the switching times $\tau_j$.
The fourth-order Trotter bound presented in Thm.~\ref{thm:4th_order} is computed by this script.

\subsection{First-order Trotter bound (Thm.~\ref{thm:trotter_thm})}
Let us see how the general formalism presented above works for the first-order Trotterization.
Without loss of generality, we shift the Hamiltonians $H_1$ and $H_2$ so that $(H_1+H_2)\varphi=0$.
The switching times for the first-order Trotterization are given by $\tau_1=1$ and $\tau_2=1$, and hence, $T_1=t/N$ and $T_2=2t/N$.
In this case, the formula~\eqref{eq:kth_action_explicit_expanded0} for the integral action gives
\begin{equation}
S_1(s) \varphi
=\begin{cases}
\medskip
\displaystyle
sH_1 \varphi,&s\in[0,\frac{t}{N}),\\
\displaystyle
\left[\left(s-\frac{t}{N}\right)H_2
+\frac{t}{N}H_1 \right]\varphi ,&
s\in[\frac{t}{N},\frac{2t}{N}).
\end{cases}
\label{eq:S1_first_order}
\end{equation}
It is easy to verify that we have
\begin{equation*}
S_1(2t/N) \varphi=\frac{t}{N}(H_1+H_2) \varphi =0.
\end{equation*}
Therefore, 
assuming that $\varphi \in\mathcal{D}(H_1^2)\cap \mathcal{D}(H_2^2)$,
the error of the first-order Trotterization is bounded by Eq.~(\ref{eqn:error_estimate_p}),
\begin{equation*}
\xi_N^{(1)}(t;\varphi)\leq N (I_1+I_2),
\end{equation*}
with $I_j$ ($j=1,2$) the integral over the $j$th time slot.
For the first time slot $s\in[0,\frac{t}{N})$, we have
\begin{equation*}
	I_1=\int_0^{t/N}\rmd s\,s\Vert H_1^2\varphi\Vert =\frac{t^2}{2N^2}\Vert H_1^2\varphi\Vert.
\end{equation*}
For the second time slot $s\in[\frac{t}{N},\frac{2t}{N})$,
\begin{equation*}
	I_2
	= \int_{t/N}^{2t/N}\rmd s\left\Vert\left[
	\left(s-\frac{t}{N}\right)H_2^2 
	+\frac{t}{N}H_2H_1
	\right]\varphi\right\Vert
	=\frac{t^2}{2N^2}\|H_2^2\varphi\|,
\end{equation*}
where the second step follows by the fact that $H_2H_1\varphi = - H_2^2\varphi$, which also shows that for a zero-energy eigenstate of $H_1+H_2$, the domain condition $\varphi \in\mathcal{D}(H_2 H_1)$ is equivalent to $\varphi\in \mathcal{D}(H_2^2)$.
Summing up these two terms reproduces the bound~(\ref{eq:lemma_bound2}) of Lemma~\ref{lem:bound_refined}\@.

\subsection{Second-order Trotter bound (Thm.~\ref{thm:2nd_order})}
\label{app:second_order_derivation}
The error bound on the second-order Trotterization in Thm.~\ref{thm:2nd_order} is obtained as an application of Lemma~\ref{lem:bound_pth_order}.
\begin{proof}[Proof of Thm.~\ref{thm:2nd_order}]
We shift the Hamiltonians $H_1$ and $H_2$ so that $(H_1+H_2)\varphi=0$, and apply the general formalism for bounding the error of a higher-order Trotterization described above.
The switching times for the second-order Trotterization are given by $\tau_1=1/2$, $\tau_2=1$, and $\tau_3=1/2$.
In this case, the formula~\eqref{eq:kth_action_explicit_expanded0} for the integral action gives
\begingroup
\allowdisplaybreaks
\begin{gather}
S_1(s)
=\begin{cases}
\medskip
\displaystyle
\frac{1}{2} s H_1,&s\in[0,\frac{t}{N}),\\
\medskip
\displaystyle
\left(s-\frac{t}{N}\right)H_2+\frac{t}{2N}H_1,&
s\in[\frac{t}{N},\frac{2t}{N}),\\
\displaystyle
\frac{1}{2}\left(s-\frac{t}{N}\right)H_1
+\frac{t}{N}H_2,&
s\in[\frac{2t}{N},\frac{3t}{N}),
\end{cases}
\label{eq:S1_second_order}
\\
S_2(s)
=\begin{cases}
\medskip
\displaystyle
\frac{1}{8}s^2H_1^2,&s\in[0,\frac{t}{N}),\\
\medskip
\displaystyle
\frac{1}{2}\left(s-\frac{t}{N}\right)^2H_2^2+\left(s-\frac{t}{N}\right)\frac{t}{2N}H_2H_1+\frac{t^2}{8N^2}H_1^2,&
s\in[\frac{t}{N},\frac{2t}{N}),\\
\displaystyle
\frac{1}{8}\left(s-\frac{t}{N}\right)^2H_1^2
+\left(s-\frac{2t}{N}\right)\frac{t}{2N}H_1H_2
+\frac{t^2}{2N^2}H_2^2+\frac{t^2}{2N^2}H_2H_1,&
s\in[\frac{2t}{N},\frac{3t}{N}).
\end{cases}
\label{eq:S2_second_order}
\end{gather}
\endgroup
It is easy to verify that we have
\begin{equation*}
S_1(3t/N)=\frac{t}{N}(H_1+H_2),\qquad
S_2(3t/N)=\frac{1}{2}\frac{t^2}{N^2}(H_1+H_2)^2,
\end{equation*}
and thus $S_1(3t/N)\varphi=0$ and $S_2(3t/N)\varphi=0$ for the eigenstate satisfying $(H_1+H_2)\varphi=0$.

Therefore, by assuming that 
\begin{equation*}
\varphi\in\mathcal{D}(H_1^3)\cap\mathcal{D}(H_2 H_1^2)\cap\mathcal{D}(H_2^3),	
\end{equation*}
the error of the second-order Trotterization is bounded by Eq.~(\ref{eqn:error_estimate_p}),
\begin{equation*}
\xi_N^{(2)}(t;\varphi)\leq N (I_1+I_2+I_3),
\end{equation*}
with $I_j$ ($j=1,2,3$) the integral over the $j$th time slot.
For the first time slot $s\in[0,\frac{t}{N})$, we have
\begin{equation*}
	I_1=\int_0^{t/N}\rmd s\,\frac{1}{16}s^2\Vert H_1^3\varphi\Vert  =\frac{t^3}{48N^3} \Vert H_1^3\varphi\Vert.
\end{equation*}
For the second time slot $s\in[\frac{t}{N},\frac{2t}{N})$,
\begin{align*}
	I_2&= \int_{t/N}^{2t/N}\rmd s\left\Vert\left[
	\frac{1}{2}\left(s-\frac{t}{N}\right)^2 H_2^3 
	+\left(s-\frac{t}{N}\right)\frac{t}{2N}H_2^2H_1
	+\frac{t^2}{8N^2}H_2H_1^2
	\right]\varphi\right\Vert\\
	&\leq \frac{t^3}{12N^3}\Vert H_2^3\varphi\Vert + \frac{t^3}{8N^3}\Vert H_2H_1^2\varphi\Vert,
\end{align*}
where the second step follows by the fact that $H_2^2H_1\varphi = - H_2^3\varphi$ [so that $\varphi\in\mathcal{D}(H_2^3)$ is equivalent to $\varphi\in\mathcal{D}(H_2^2 H_1)$] and the triangle inequality.
For the last time slot $s\in[\frac{2t}{N}, \frac{3t}{N})$, we get
\begin{align*}
I_3 &= \int_{2t/N}^{3t/N}\rmd s \left\Vert\left[
\frac{1}{16}\left(s-\frac{t}{N}\right)^2H_1^3
+\left(s-\frac{2t}{N}\right)\frac{t}{4N}H_1^2H_2 
+\frac{t^2}{4N^2}H_1H_2(H_1+H_2)
\right] \varphi\right\Vert\\
	&=\frac{t^3}{48N^3}\Vert H_1^3\varphi\Vert,
\end{align*}
where we have used $H_1^2H_2\varphi=-H_1^3\varphi$, so that $\varphi\in\mathcal{D}(H_1^3)$ is equivalent to $\varphi\in\mathcal{D}(H_1^2 H_2)$, and $\varphi\in\mathcal{D}\bm{(}H_1H_2 (H_1+H_2)\bm{)}$ with $H_1H_2(H_1+H_2)\varphi=0$.
Summing up these three terms results in the second-order Trotter bound presented in Thm.~\ref{thm:2nd_order}\@.
\end{proof}

\subsection{A general bound on the $p$th-order Trotter error (Thm.~\ref{thm:loose_bound_higher_order})}
\label{app:ProofThm5}
By bounding each term appearing in Eq.~\eqref{eqn:error_estimate_p}, we get Thm.~\ref{thm:loose_bound_higher_order}\@.
\begin{proof}[Proof of Thm.~\ref{thm:loose_bound_higher_order}]
We shift the Hamiltonians $H_1$ and $H_2$ so that $(H_1+H_2)\varphi=0$, and apply the general formalism for bonding the error of a higher-order Trotterization described above.
Bounding Eq.~(\ref{eqn:error_estimate_p}) with Eq.~(\ref{eq:kth_action_explicit_expanded0}) by triangle inequality, we have
\begingroup
\allowdisplaybreaks
\begin{align*}
\xi_N^{(p)}(t)
\le{}& N\sum_{j=1}^M
\mathop{\sum_{\bm{\beta}\in\mathbb{N}^j}}_{|\bm{\beta}|=p} \frac{1}{\bm{\beta}!} 
\left(\frac{t}{N}\right)^{p-\beta_j}
\int_0^{\frac{t}{N}}\rmd s\,
		s^{\beta_j} |\tau_{j}|^{\beta_{j}+1}
		{|\tau_{j-1}|^{\beta_{j-1}}\cdots|\tau_1|^{\beta_1}}
		\|H_{\sigma(j)}^{\beta_j+1}
		H_{\sigma(j-1)}^{\beta_{j-1}}
		\cdots
		H_{\sigma(1)}^{\beta_1}\varphi\|
\nonumber\\
={}&N\sum_{j=1}^M
\mathop{\sum_{\bm{\beta}\in\mathbb{N}^j}}_{|\bm{\beta}|=p} \frac{1}{\bm{\beta}!} 
\left(\frac{t}{N}\right)^{p+1}
		\frac{|\tau_{j}|^{\beta_{j}+1}|\tau_{j-1}|^{\beta_{j-1}}\cdots|\tau_1|^{\beta_1}}{\beta_j+1}
		\|H_{\sigma(j)}^{\beta_j+1}
		H_{\sigma(j-1)}^{\beta_{j-1}}
		\cdots
		H_{\sigma(1)}^{\beta_1}\varphi\|
\nonumber\\
={}&N
\mathop{\sum_{\bm{\beta}\in\mathbb{N}^M}}_{|\bm{\beta}|=p+1} \frac{1}{\bm{\beta}!} 
\left(\frac{t}{N}\right)^{p+1}
		{|\tau_M|^{\beta_M}\cdots|\tau_1|^{\beta_1}}
		\|H_{\sigma(M)}^{\beta_M}
		\cdots
		H_{\sigma(1)}^{\beta_1}\varphi\|
\nonumber\\
\le{}&N
\left(
\frac{t}{N}
\right)^{p+1}
\frac{1}{(p+1)!}
\mathop{\sum_{\bm{\beta}\in\mathbb{N}^M}}_{|\bm{\beta}|=p+1} \frac{|\bm{\beta}|!}{\bm{\beta}!} {|\tau_1|^{\beta_1}\cdots|\tau_M|^{\beta_M}}
{K}_{p+1}(\varphi)
\nonumber\\
={}&N
\left(
\frac{t}{N}
\right)^{p+1}
\frac{(|\tau_1|+\cdots+|\tau_M|)^{p+1}}{(p+1)!}
{K}_{p+1}(\varphi),
\nonumber\\
={}&N
\left(
\frac{\tau_*t}{N}
\right)^{p+1}
\frac{1}{(p+1)!}
K_{p+1}(\varphi),
\end{align*}
\endgroup
where $\tau_*=\sum_{j=1}^M|\tau_j|$ and ${K}_p(\varphi)=\max_{1\leq j_1\leq\dots\leq j_p\leq M}\|H_{\sigma(j_p)}\cdots H_{\sigma(j_1)}\varphi\|$.
The second last equality is the multinomial theorem, with ${|\bm{\beta}|!}/{\bm{\beta}!}$ being the multinomial coefficient.
\end{proof}

\subsection{Fractional scaling of the first-order Trotter bound (Thm.~\ref{thm:first_order_alpha})}
Theorem~\ref{thm:trotter_thm} requires $\varphi\in\mathcal{D}(H_1^2)\cap\mathcal{D}(H_2^2)$ and Thm.~\ref{thm:2nd_order} requires $\varphi\in\mathcal{D}(H_1^3)\cap\mathcal{D}(H_2H_1^2)\cap\mathcal{D}(H_2^3)$.
If these conditions are not met, we lose the $N^{-1}$ scaling of the first-order Trotter bound and the $N^{-2}$ scaling of the second-order Trotter bound.
To estimate the scalings in such cases, we perform refined analyses as done in Thm.~\ref{thm:first_order_alpha} for the first-order Trotterization.
The general formalism described above enables us to do it systematically for any general $p$th-order Trotterization.

Let us first look at the first-order Trotterization. In this case $\tau_1=\tau_2=1$,
and we have to compute the right-hand side of Eq.~(\ref{eqn:Dyson_error_q}), $R^{(p)}_q \varphi$, for $p=q=1$, where
\begin{equation*}
	R^{(1)}_1= (-\rmi)^{2}\int_0^{2t/N}\rmd s\,U^{(1)}(2t/N)U^{(1)\dag}(s)H^{(1)}(s)S_1(s).
\end{equation*}
Here, the integral action $S_1(s)$ is given by Eq.~(\ref{eq:S1_first_order}), and
the first-order Trotter evolution reads 
\begin{equation*}
U^{(1)}(2t/N)U^{(1)\dag}(s)
= \rme^{-\rmi\frac{t}{N}H_2}
\begin{cases}
\medskip
\displaystyle
\rme^{-\rmi(\frac{t}{N}-s)H_1},&s\in\left[0,\frac{t}{N}\right),\\
\displaystyle
\rme^{\rmi(s-\frac{t}{N})H_2},&
s\in\left[\frac{t}{N},\frac{2t}{N}\right).
\end{cases}
\end{equation*}
Thus, one has
\begin{equation*}
R^{(1)}_1
=
{-}\int_0^{t/N}\rmd s\,\rme^{-\rmi\frac{t}{N}H_2}\rme^{-\rmi(\frac{t}{N}-s)H_1}sH_1^2 
-\int_{t/N}^{2t/N}\rmd s\,\rme^{-\rmi(\frac{2t}{N}-s)H_2}\left[
\left(s-\frac{t}{N}\right)H_2^2+\frac{t}{N}H_2H_1
\right],
\end{equation*}
which can be integrated and gives
\begin{equation*}
	R^{(1)}_1 =
	\rme^{-\rmi\frac{t}{N}H_2}\left(
	\rme^{-\rmi\frac{t}{N}H_1}-I+\rmi\frac{t}{N}H_1
	\right)
		-
	\rme^{-\rmi\frac{t}{N}H_2}
	\left(
	\rme^{\rmi\frac{t}{N}H_2}-I-\rmi\frac{t}{N}H_2
	\right)
	+\rmi\frac{t}{N}\left(I-\rme^{-\rmi\frac{t}{N}H_2}\right)(H_1+H_2).
\end{equation*}
Therefore, for $\varphi\in\mathcal{D}(H_1)\cap\mathcal{D}(H_2)$ with $(H_1+H_2)\varphi=0$,
one finally gets 
\begin{equation*}
	[U^{(1)}(2t/N)-I]\varphi
	=R^{(1)}_1\varphi =
	\rme^{-\rmi\frac{t}{N}H_2}\left(
	\rme^{-\rmi\frac{t}{N}H_1}-I+\rmi\frac{t}{N}H_1
	\right)
	\varphi
	-
	\rme^{-\rmi\frac{t}{N}H_2}
	\left(
	\rme^{\rmi\frac{t}{N}H_2}-I-\rmi\frac{t}{N}H_2
	\right)
	\varphi.
\end{equation*}
This reproduces the bound~(\ref{eq:first_order_bound}) on the first-order Trotter error, which is estimated as done in Appendix~\ref{sec:proof_trotter_thm}, even when the domain condition $\varphi\in\mathcal{D}(H_1^2)\cap\mathcal{D}(H_2^2)$ for Thm.~\ref{thm:trotter_thm} is not fulfilled.

\subsection{Fractional scalings of the second-order Trotter bounds (Thm.~\ref{thm:second_order_alpha})}
Let us now estimate the second-order Trotter error when the domain condition $\varphi\in\mathcal{D}(H_1^3)\cap\mathcal{D}(H_2H_1^2)\cap\mathcal{D}(H_2^3)$ for Thm.~\ref{thm:2nd_order} is not fulfilled.
For the second-order Trotterization, the integral actions $S_1(s)$ and $S_2(s)$ are given by Eqs.~(\ref{eq:S1_second_order}) and~(\ref{eq:S2_second_order}), respectively.
Notice here that, for the second-order Trotter evolution,
\begin{equation*}
U^{(2)\dag}(s)
=\begin{cases}
\medskip
\displaystyle
\rme^{\rmi \frac{s}{2} H_1},&s\in\left[0,\frac{t}{N}\right),\\
\medskip
\displaystyle
\rme^{\rmi\frac{t}{2N}H_1}
\rme^{\rmi(s-\frac{t}{N})H_2},
&s\in\left[\frac{t}{N},\frac{2t}{N}\right),\\
\displaystyle
\rme^{\rmi\frac{t}{2N}H_1}\rme^{\rmi\frac{t}{N}H_2}\rme^{\rmi \frac{1}{2}(s-\frac{2t}{N})H_1},&
s\in\left[\frac{2t}{N},\frac{3t}{N}\right).
\end{cases}
\end{equation*}
Then, we have the following lemma, which replaces the bound~(\ref{eq:lemma_bound1}) of Lemma~\ref{lem:bound_refined}, and can be used to estimate the second-order Trotter bound for $\varphi\not\in\mathcal{D}(H_1^3)\cap\mathcal{D}(H_2H_1^2)\cap\mathcal{D}(H_2^3)$.

\begin{lemma}\label{lem:FractionalDelta}
The second-order Trotter bound can be estimated by
\begingroup
\allowdisplaybreaks
\begin{align}
\xi_N^{(2)}(t;\varphi)
\le{}&
2N\left\|
\left(
\rme^{-\rmi\frac{t}{2N}H_1}
-I
+\rmi\frac{t}{2N}H_1
+\frac{t^2}{8N^2}H_1^2
\right)
\varphi
\right\|
\nonumber\\
&{}
+
N\left\|
\left[
\rme^{-\rmi\frac{t}{2N}H_2}
\left(
I+\rmi\frac{t}{2N}H_2
\right)
-
\rme^{\rmi\frac{t}{2N}H_2}
\left(
I
-\rmi\frac{t}{2N}H_2
\right)
\right]
\varphi
\right\|
+
\frac{t^2}{8N}
\left\|
\left(
\rme^{-\rmi\frac{t}{N}H_2}
-I
\right)
H_1^2
\varphi
\right\|,
\label{eq:fractional_second_2}
\intertext{for $\varphi\in\mathcal{D}(H_1^2)\cap\mathcal{D}(H_2)$, and by}
\xi_N^{(2)}(t;\varphi)
\le{}&
2N\left\|
\left(
\rme^{-\rmi\frac{t}{2N}H_1}
-I
+\rmi\frac{t}{2N}H_1
\right)
\varphi
\right\|
+
N\left\|
\left[
\rme^{-\rmi\frac{t}{2N}H_2}
\left(
I+\rmi\frac{t}{2N}H_2
\right)
-
\rme^{\rmi\frac{t}{2N}H_2}
\left(
I
-\rmi\frac{t}{2N}H_2
\right)
\right]
\varphi
\right\|,
\label{eq:fractional_second_1}
\end{align}
\endgroup
for $\varphi\in\mathcal{D}(H_1)\cap\mathcal{D}(H_2)$.
\end{lemma}
\begin{proof}
The second-order integral action $S_2(s)$ is given by Eq.~(\ref{eq:S2_second_order}), and the right-hand side of Eq.~(\ref{eqn:Dyson_error_q}), $R^{(p)}_q \varphi$, for $p=q=2$ reads
\begingroup
\allowdisplaybreaks
\begin{align*}
R^{(2)}_2 ={}&
\rmi\int_0^{3t/N}\rmd s\,U^{(2)}(3t/N)U^{(2)\dag}(s)H^{(2)}(s)S_2(s)\nonumber\\
={}&
\rmi\int_0^{t/N}\rmd s\,\rme^{-\rmi\frac{t}{2N}H_1}\rme^{-\rmi\frac{t}{N}H_2}\rme^{-\rmi \frac{1}{2}(\frac{t}{N}-s)H_1}\frac{1}{16}s^2H_1^3
\nonumber\\
&{}
+\rmi\int_{t/N}^{2t/N}\rmd s\,\rme^{-\rmi\frac{t}{2N}H_1}\rme^{-\rmi(\frac{2t}{N}-s)H_2}\left[
\frac{1}{2}\left(s-\frac{t}{N}\right)^2H_2^3
+\left(s-\frac{t}{N}\right)\frac{t}{2N}H_2^2H_1
+\frac{t^2}{8N^2}H_2H_1^2
\right]
\nonumber\\
&{}
+\rmi\int_{2t/N}^{3t/N}\rmd s\,\rme^{-\rmi \frac{1}{2}(\frac{3t}{N}-s)H_1}\left[
\frac{1}{16}\left(s-\frac{t}{N}\right)^2H_1^3
+\left(s-\frac{2t}{N}\right)\frac{t}{4N}H_1^2H_2
+\frac{t^2}{4N^2}H_1H_2^2
+\frac{t^2}{4N^2}H_1H_2H_1
\right].
\end{align*}
By integrating, one gets
\begin{align*}
R^{(2)}_2
={}&
\rme^{-\rmi\frac{t}{2N}H_1}\rme^{-\rmi\frac{t}{N}H_2}
\left(
\rme^{-\rmi\frac{t}{2N}H_1}
-I
+\rmi\frac{t}{2N}H_1
+\frac{t^2}{8N^2}H_1^2
\right)
\nonumber\\
&{}
+\rme^{-\rmi\frac{t}{2N}H_1}
\rme^{-\rmi\frac{t}{2N}H_2}
\left[
\rme^{-\rmi\frac{t}{2N}H_2}
\left(
I
+
\rmi\frac{t}{2N}
H_2
\right)
-
\rme^{\rmi\frac{t}{2N}H_2}
\left(
I
-
\rmi\frac{t}{2N}H_2
\right)
\right]
\nonumber\\
&{}
-\frac{t^2}{8N^2}
\rme^{-\rmi\frac{t}{2N}H_1}
\left(
\rme^{-\rmi\frac{t}{N}H_2}
-I
\right)
H_1^2
\nonumber\\
&{}
-\rmi\frac{t}{2N}
\rme^{-\rmi\frac{t}{2N}H_1}
\left(
\rme^{-\rmi\frac{t}{N}H_2}
-I
+\rmi\frac{t}{N}H_2
\right)
(H_1+H_2)
\nonumber\\
&{}
-\rme^{-\rmi\frac{t}{2N}H_1}
\left(
\rme^{\rmi\frac{t}{2N}H_1}
-I
-
\rmi\frac{t}{2N}H_1
+
\frac{t^2}{8N^2}H_1^2
\right)
\nonumber\\
&{}
-\rmi\frac{t}{N}
\left(
\rme^{-\rmi\frac{t}{2N}H_1}
-I
+
\rmi\frac{t}{2N}H_1
\right)
(H_1+H_2)
\nonumber\\
&{}
-
\frac{t^2}{2N^2}
\left(
\rme^{-\rmi\frac{t}{2N}H_1}
-I
\right)
H_2(H_1+H_2).
\end{align*}
\endgroup
Therefore, for $\varphi\in\mathcal{D}(H_1^2)\cap\mathcal{D}(H_2)$ with $(H_1+H_2)\varphi=0$, one automatically has that $\varphi\in \mathcal{D}\bm{(}X(H_1+H_2)\bm{)}$ for every operator $X$, and
\begin{align}
	[U^{(2)}(3t/N)-I]\varphi
	={}&R_2^{(2)}\varphi
	\nonumber\\
={}&
\rme^{-\rmi\frac{t}{2N}H_1}\rme^{-\rmi\frac{t}{N}H_2}
\left(
\rme^{-\rmi\frac{t}{2N}H_1}
-I
+\rmi\frac{t}{2N}H_1
+\frac{t^2}{8N^2}H_1^2
\right)\varphi
\nonumber\\
&{}
+\rme^{-\rmi\frac{t}{2N}H_1}
\rme^{-\rmi\frac{t}{2N}H_2}
\left[
\rme^{-\rmi\frac{t}{2N}H_2}
\left(
I
+
\rmi\frac{t}{2N}
H_2
\right)
-
\rme^{\rmi\frac{t}{2N}H_2}
\left(
I
-
\rmi\frac{t}{2N}H_2
\right)
\right]\varphi
\nonumber\\
&{}
-\frac{t^2}{8N^2}
\rme^{-\rmi\frac{t}{2N}H_1}
\left(
\rme^{-\rmi\frac{t}{N}H_2}
-I
\right)
H_1^2\varphi
\nonumber\\
&{}
-\rme^{-\rmi\frac{t}{2N}H_1}
\left(
\rme^{\rmi\frac{t}{2N}H_1}
-I
-
\rmi\frac{t}{2N}H_1
+
\frac{t^2}{8N^2}H_1^2
\right)
\varphi.
\end{align}
Bounding it by the triangle inequality yields Eq.~(\ref{eq:fractional_second_2}).
This works for $\varphi\in\mathcal{D}(H_1^2)\cap\mathcal{D}(H_2)$.

If the input state $\varphi$ is not in $\mathcal{D}(H_1^2)$, to estimate the second-order Trotter bound we can use the remainder term $R^{(p)}_q\varphi$ in Eq.~(\ref{eqn:Dyson_error_q}) for $p=2$ and $q=1$.
The first-order integral action $S_1(s)$ is given by Eq.~(\ref{eq:S1_second_order}), and the remainder term $R^{(p)}_q$ for $p=2$ and $q=1$ reads
\begin{align*}
R^{(2)}_1
={}&{-}\int_0^{3t/N}\rmd s\,U^{(2)}(3t/N)U^{(2)\dag}(s)H^{(2)}(s)S_1(s)
\nonumber\\
={}&
{-}\int_0^{t/N}\rmd s\,\rme^{-\rmi\frac{t}{2N}H_1}\rme^{-\rmi\frac{t}{N}H_2}\rme^{-\rmi \frac{1}{2}(\frac{t}{N}-s)H_1} \frac{1}{4} sH_1^2
\nonumber\\
&{}
-\int_{t/N}^{2t/N}\rmd s\,
\rme^{-\rmi\frac{t}{2N}H_1}\rme^{-\rmi(\frac{2t}{N}-s)H_2}
\left[
\left(s-\frac{t}{N}\right)H_2^2
+\frac{t}{2N}H_2H_1
\right]
\nonumber\\
&{}
-\int_{2t/N}^{3t/N}\rmd s\,
\rme^{-\rmi \frac{1}{2}(\frac{3t}{N}-s)H_1}
\left[
\frac{1}{4}\left(s-\frac{t}{N}\right)H_1^2
+\frac{t}{2N}H_1H_2
\right].
\end{align*}
By integrating it, one has
\begin{align*}
R^{(2)}_1
={}&
\rme^{-\rmi\frac{t}{2N}H_1}\rme^{-\rmi\frac{t}{N}H_2}
\left(
\rme^{-\rmi\frac{t}{2N}H_1}
-I
+\rmi\frac{t}{2N}H_1
\right)
\nonumber\\
&{}
+
\rme^{-\rmi\frac{t}{2N}H_1}
\rme^{-\rmi\frac{t}{2N}H_2}
\left[
\rme^{-\rmi\frac{t}{2N}H_2}
\left(
I+\rmi\frac{t}{2N}H_2
\right)
-
\rme^{\rmi\frac{t}{2N}H_2}
\left(
I
-\rmi\frac{t}{2N}H_2
\right)
\right]
\nonumber\\
&{}
-
\rmi\frac{t}{2N}\rme^{-\rmi\frac{t}{2N}H_1}
\left(
\rme^{-\rmi\frac{t}{N}H_2} -I
\right)
(H_1+H_2)
\nonumber\\
&{}
-
\rme^{-\rmi\frac{t}{2N}H_1}
\left(
\rme^{\rmi\frac{t}{2N}H_1}
-I
-\rmi\frac{t}{2N}H_1
\right)
\nonumber\\
&{}
-\rmi\frac{t}{N}
\left(
\rme^{-\rmi\frac{t}{2N}H_1} -I
\right)
(H_1+H_2),
\end{align*}
whence, for $\varphi\in\mathcal{D}(H_1)\cap\mathcal{D}(H_2)$ with $(H_1+H_2)\varphi=0$, one gets
\begin{align*}
[U^{(2)}(2t/N)-I]\varphi
={}&R_1^{(2)}\varphi
\nonumber\\
={}&
\rme^{-\rmi\frac{t}{2N}H_1}\rme^{-\rmi\frac{t}{N}H_2}
\left(
\rme^{-\rmi\frac{t}{2N}H_1}
-I
+\rmi\frac{t}{2N}H_1
\right)
\varphi
\nonumber\\
&{}
+
\rme^{-\rmi\frac{t}{2N}H_1}
\rme^{-\rmi\frac{t}{2N}H_2}
\left[
\rme^{-\rmi\frac{t}{2N}H_2}
\left(
I+\rmi\frac{t}{2N}H_2
\right)
-
\rme^{\rmi\frac{t}{2N}H_2}
\left(
I
-\rmi\frac{t}{2N}H_2
\right)
\right]
\varphi
\nonumber\\
&{}
-
\rme^{-\rmi\frac{t}{2N}H_1}
\left(
\rme^{\rmi\frac{t}{2N}H_1}
-I
-\rmi\frac{t}{2N}H_1
\right)
\varphi.
\end{align*}
Bounding it by the triangle inequality yields Eq.~(\ref{eq:fractional_second_1}).
\end{proof}

In order to prove Thm.~\ref{thm:second_order_alpha}, we need to bound
\begingroup
\allowdisplaybreaks
\begin{gather*}
\Delta^{(0)}(s)
=\|(\rme^{-\rmi sH}-1)\varphi\|^2
=\int_{\mathbb{R}} f_0(s\lambda)\,\rmd\mu_\varphi(\lambda),
\\
\Delta^{(1)}(s)
=\|(\rme^{-\rmi sH}-1+\rmi sH)\varphi\|^2
=\int_{\mathbb{R}} f_1(s\lambda)\,\rmd\mu_\varphi(\lambda),
\\
\Delta^{(2)}(s)
=\left\|\left(\rme^{-\rmi sH}-1+\rmi sH+\frac{1}{2}s^2H^2\right)\varphi\right\|^2
=\int_{\mathbb{R}} f_2(s\lambda)\,\rmd\mu_\varphi(\lambda),
\\
\tilde{\Delta}^{(2)}(s)
=\|[\rme^{-\rmi sH}(1+\rmi sH)-\rme^{\rmi sH}(1-\rmi sH)]\varphi\|^2
=\int_{\mathbb{R}}\tilde{f}_2(s\lambda)\,\rmd\mu_\varphi(\lambda),
\end{gather*}
\endgroup
where
\begin{gather*}
f_0(x)
=|\rme^{-\rmi x}-1|^2
=2(1-\cos x)
\le x^2,
\vphantom{\frac{1}{4}}
\\
f_1(x)
=|\rme^{-\rmi x}-1+\rmi x|^2
=(1-\cos x)^2+(x-\sin x)^2
\le\frac{1}{4}x^4,
\\
f_2(x)
=\left|
\rme^{-\rmi x}-1+\rmi x+\frac{1}{2}x^2
\right|^2
=\left(1-\frac{1}{2}x^2-\cos x\right)^2+(x-\sin x)^2
\le\frac{1}{36}x^6,
\\
\tilde{f}_2(x)
=|
\rme^{-\rmi x}(1+\rmi x)
-
\rme^{\rmi x}(1-\rmi x)
|^2
=4(
\sin x
-x\cos x
)^2
\le\frac{4}{9}x^6.
\end{gather*}
As done in the proof of Lemma~\ref{lem:FractionalScaling} 
in Appendix~\ref{sec:proof_trotter_thm}, if 
\begin{equation*}
\mu_{\psi}(\{|x|\geq \lambda\})
\le\left(\frac{\lambda_0}{\lambda}\right)^{2\delta},
\end{equation*}
for some $\lambda_0>0$, then we have
\begingroup
\allowdisplaybreaks
\begin{align}
\Delta^{(0)}(s)
&\le
\begin{cases}
\medskip
\displaystyle	
\frac{2\pi\delta}{\Gamma(2\delta+1)\sin(\delta\pi)}(\lambda_0 s)^{2\delta},&
0<\delta<1,\\
\medskip
\displaystyle
\Lambda^2s^2
+
\frac{2\lambda_0^2}{\Lambda^2}
g_0(\Lambda s),&
\delta=1,\\
\displaystyle
\|H\psi\|^2s^2,
\vphantom{\frac{1}{4}}&
\delta>1,
\end{cases}
\label{eq:bound_Delta0}
\\
\Delta^{(1)}(s)
&\le 
\begin{cases}
\medskip
\displaystyle	
\frac{\pi}{\Gamma(2\delta-1)\sin[(\delta-1)\pi]}(\lambda_0 s)^{2\delta},&
1<\delta<2,\\
\medskip
\displaystyle
\frac{1}{4}\Lambda^2\|H\psi\|^2s^4
+
\frac{\lambda_0^4}{\Lambda^4}
g_1(\Lambda s),&
\delta=2,\\
\displaystyle
\frac{1}{4}\|H^2\psi\|^2s^4,&
\delta>2,
\end{cases}
\label{eq:bound_Delta1}
\\
\Delta^{(2)}(s)
&\le
\begin{cases}
\medskip
\displaystyle
\frac{\pi}{2\Gamma(2\delta-2)\sin(\delta\pi)]}
(\lambda_0 s)^{2\delta},&
2<\delta<3,\\
\medskip
\displaystyle
\frac{1}{36}\Lambda^2\|H^2\psi\|^2s^6
+
\frac{\lambda_0^6}{6\Lambda^6}
g_2(\Lambda s),&
\delta=3,\\
\displaystyle
\frac{1}{36}\|H^3\psi\|^2s^6,&
\delta>3,
\end{cases}
\\
\tilde{\Delta}_2(s)
&\le
\begin{cases}
\medskip
\displaystyle
\frac{2^{2\delta-1}(\delta-2)\pi}{\Gamma(2\delta-1)\sin[(\delta-2)\pi]} (\lambda_0 s)^{2\delta},&
1<\delta<3,\\
\medskip
\displaystyle
\frac{4}{9}\Lambda^2\|H^2\psi\|^2s^6
+
\frac{8\lambda_0^6}{3\Lambda^6}
\tilde{g}_2(\Lambda s),&
\delta=3,\\
\displaystyle
\frac{4}{9}\|H^3\psi\|^2s^6,&
\delta>3,
\end{cases}
\end{align}
\endgroup
where
\begingroup
\allowdisplaybreaks
\begin{align*}
g_0(x)
&=-x^2\Ci(x)+(1-\cos x)+x\sin x
\vphantom{\frac{1}{4}}
\\
&=-x^2\log x+\left(\frac{3}{2}-\gamma\right)x^2+O(x^4),
\\
g_1(x)
&=-x^4\Ci(x)+\left(1+\frac{1}{2}x^2\right)[(1-\cos x)^2+(x-\sin x)^2]-2x^3(x-\sin x)+\frac{3}{2}x^4\\
&=-x^4\log x+\left(\frac{7}{4}-\gamma\right)x^4+O(x^6),
\\
g_2(x)
&=-x^6\Ci(x)
+
12+\frac{9}{2}x^4
-(12-6x^2+x^4)\cos x
-x(12+2x^2-x^4)\sin x
\vphantom{\frac{1}{4}}
\\
&=-x^6\log x+(2-\gamma)x^6+O(x^8),
\vphantom{\frac{1}{4}}
\\
\tilde{g}_2(x)
&=-x^6\Ci(2x)
+
\frac{1}{8}[
6+9x^2
-(6-3x^2+2x^4)\cos2x
-2x(2-x^2)(3+2x^2)\sin2x
]
\vphantom{\frac{1}{4}}
\\
&=-x^6\log2x+\left(\frac{7}{4}-\gamma\right)x^6+O(x^8).
\end{align*}
\endgroup
Theorem~\ref{thm:second_order_alpha} follows by using these bounds for Eqs.~(\ref{eq:fractional_second_2}) and~(\ref{eq:fractional_second_1}) of Lemma~\ref{lem:FractionalDelta}\@.
We summarize the scalings of the first- and second-order Trotter bounds for the hydrogen atom in Table~\ref{tab:fractional_scalings}\@.
\begin{table}
\caption{Scalings of the state-dependent bounds on the first- and second-order Trotter errors for the hydrogen atom with $H_1=-\frac{\hbar^2}{2m_\mathrm{e}}\Delta$ and $H_2=-\frac{\hbar^2}{m_\mathrm{e}a_0r}$.
	The state $\Psi_{n\ell m}$ is an eigenfunction of the Hamiltonian $H_1+H_2$ of the hydrogen atom.
	The parameter $\delta_j$ specifies the scaling of the tail of the spectral density $\rho_{j,\Psi_{n\ell m}}(\lambda)\sim|\lambda|^{-2\delta_j-1}$ of the Hamiltonian $H_j$ ($j=1,2$) at $\Psi_{n\ell m}$.
	Similarly, $\delta_{12}$ specifies the scaling of the tail of the spectral density $\rho_{2,H_1^2\Psi_{n\ell m}}(\lambda)\sim|\lambda|^{-2\delta_{12}-1}$ of the Hamiltonian $H_2$ at the vector $H_1^2\Psi_{n\ell m}$, while $\delta_{21}$ specifies the scaling of the tail of the spectral density $\rho_{1,H_2^2\Psi_{n\ell m}}(\lambda)\sim|\lambda|^{-2\delta_{21}-1}$ of the Hamiltonian $H_1$ at the vector $H_2^2\Psi_{n\ell m}$.
	``$H_2H_1$'' indicates the scaling of the bound on the first-order Trotter product with cycle $\rme^{-\rmi\frac{t}{N}H_2}\rme^{-\rmi\frac{t}{N}H_1}$ from Thm.~\ref{thm:trotter_thm} and Thm.~\ref{thm:first_order_alpha}\@. 
	``$H_1H_2H_1$'' and ``$H_2H_1H_2$'' indicate the scalings of the bounds on the second-order Trotter products with cycles $\rme^{-\rmi\frac{t}{2N}H_1}\rme^{-\rmi\frac{t}{N}H_2}\rme^{-\rmi\frac{t}{2N}H_1}$ and $\rme^{-\rmi\frac{t}{2N}H_2}\rme^{-\rmi\frac{t}{N}H_1}\rme^{-\rmi\frac{t}{2N}H_2}$, respectively, from Thm.~\ref{thm:2nd_order} and Thm.~\ref{thm:second_order_alpha}\@.
	Blank items in the table indicate that the corresponding entry is not applicable.
	}
\label{tab:fractional_scalings}
\centering
\begin{tabular}{|c|cc|cc|c|cc|}
	\hline
	\begin{tabular}[c]{@{}c@{}} $\ell$ \end{tabular} 
	&	\begin{tabular}[c]{@{}c@{}} $\delta_1$ \end{tabular}
	&	\begin{tabular}[c]{@{}c@{}} $\delta_2$ \end{tabular}
	&	\begin{tabular}[c]{@{}c@{}} $\delta_{12}$ \end{tabular}
	&	\begin{tabular}[c]{@{}c@{}} $\delta_{21}$ \end{tabular}
	&	\begin{tabular}[c]{@{}c@{}} $H_2H_1$ \end{tabular}
	&	\begin{tabular}[c]{@{}c@{}} $H_1H_2H_1$ \end{tabular}
	&	\begin{tabular}[c]{@{}c@{}} $H_2H_1H_2$ \end{tabular}
	\\
	\hhline{:========:}
	$0$&${5}/{4}$&${3}/{2}$& & &$N^{-1/4}$&$N^{-1/4}$&$N^{-1/4}$\\
	$1$&${7}/{4}$&${5}/{2}$& &${1}/{4}$&$N^{-3/4}$&$N^{-3/4}$&$N^{-3/4}$\\
	$2$&${9}/{4}$&${7}/{2}$&${1}/{2}$&${3}/{4}$&$N^{-1}$&$N^{-5/4}$&$N^{-5/4}$\\
	$3$&${11}/{4}$&${9}/{2}$&${1}/{2}$&${5}/{4}$&$N^{-1}$&$N^{-3/2}$&$N^{-7/4}$\\
	$\ell\ge4$&$\frac{1}{2}\ell+\frac{5}{4}$&$\ell+\frac{3}{2}$&$\ge\frac{5}{2}$&$\frac{1}{2}\ell-\frac{1}{4}$&$N^{-1}$&$N^{-2}$&$N^{-2}$\\
	\hline
\end{tabular}
\end{table}

\subsection{Fractional scalings of general $p$th-order Trotter bounds}
The above strategy for estimating the fractional scaling works for a general $p$th-order Trotter product when $\varphi\not\in\mathcal{D}\bm{(}H^{(p)}(s)S_p(s)\bm{)}$\@.
Inserting the expression~(\ref{eq:kth_action_explicit_expanded0}) for the $q$th-order integral action into the remainder term~(\ref{eqn:Dyson_error_q}) for a $p$th-order Trotter product, we get 
\begin{align}
&[U^{(p)}(T_M)-I]\varphi
\nonumber\\
&\ \
=
(-\rmi)^{q+1}\int_0^{T_M}\rmd s\,U^{(p)}(T_M)U^{(p)\dag}(s)H^{(p)}(s)S_q(s)\varphi
\nonumber\\
&\ \
=
(-\rmi)^{q+1}
\sum_{j=1}^M
\mathop{\mathop{\sum}_{\bm{\beta}\in\mathbb{N}^j}}_{|\bm{\beta}|=q}
\rme^{-\rmi\frac{t}{N}\tau_MH_{\sigma(M)}}
\cdots
\rme^{-\rmi\frac{t}{N}\tau_{j+1}H_{\sigma(j+1)}}
\nonumber\\
&\qquad\qquad\qquad
{}\times
\int_0^{\frac{t}{N}}\rmd s\,
\rme^{-\rmi(\frac{t}{N}-s) \tau_j H_{\sigma(j)}}
\frac{1}{\beta_j!}
s^{\beta_j} (\tau_j H_{\sigma(j)})^{\beta_j+1}
\left(
\frac{t}{N}
\right)^{q-\beta_j}
\frac{1}{\beta_{j-1}!}
(\tau_{j-1}H_{\sigma(j-1)})^{\beta_{j-1}}
\cdots
\frac{1}{\beta_1!}
(\tau_1H_1)^{\beta_1}
\varphi
\nonumber
\displaybreak[0]
\\
&\ \
=
\sum_{j=1}^M
\mathop{\mathop{\sum}_{\bm{\beta}\in\mathbb{N}^j}}_{|\bm{\beta}|=q}
\rme^{-\rmi\frac{t}{N}\tau_MH_{\sigma(M)}}
\cdots
\rme^{-\rmi\frac{t}{N}\tau_{j+1}H_{\sigma(j+1)}}
\nonumber\\
&\qquad\qquad\qquad
{}\times
\Biggl[
\rme^{-\rmi\frac{t}{N}\tau_jH_{\sigma(j)}}
-I
-\sum_{k=1}^{\beta_j}\frac{1}{k!}\left(-\rmi\frac{t}{N}\tau_jH_{\sigma(j)}\right)^k
\Biggr]
\left(
-\rmi\frac{t}{N}
\right)^{q-\beta_j}
\frac{1}{\beta_{j-1}!}
(\tau_{j-1}H_{\sigma(j-1)})^{\beta_{j-1}}
\cdots
\frac{1}{\beta_1!}
(\tau_1H_1)^{\beta_1}
\varphi,
\label{eq:remainder_q}
\end{align}
where we have used the Taylor expansion~(\ref{eq:TaylorExpansion}) of Remark~\ref{rmk:TaylorExpansion}\@.
It can be bounded by triangle inequality to yield an error bound
\begin{align*}
\xi_N^{(p)}(t;\varphi)
\le
N\sum_{j=1}^M
\mathop{\mathop{\sum}_{\bm{\beta}\in\mathbb{N}^j}}_{|\bm{\beta}|=q}
&\left(
\frac{t}{N}
\right)^{q-\beta_j}
\frac{|\tau_{j-1}|^{\beta_{j-1}}\cdots|\tau_1|^{\beta_1}}{\beta_{j-1}!\cdots\beta_1!}
\nonumber\\[-4truemm]
&{}\times
\Biggl\|
\Biggl[
\rme^{-\rmi\frac{t}{N}\tau_jH_{\sigma(j)}}
-I
-\sum_{k=1}^{\beta_j}\frac{1}{k!}\left(-\rmi\frac{t}{N}\tau_jH_{\sigma(j)}\right)^k
\Biggr]
H_{\sigma(j-1)}^{\beta_{j-1}}
\cdots
H_1^{\beta_1}
\varphi
\Biggr\|.
\end{align*}

This is just a rough bound.
By taking care of cancellations among terms due to $(H_1+H_2)\varphi=0$ before taking the norm, one can get a better bound.
For instance, for $q=1$, the error operator~(\ref{eq:remainder_q}) simplifies to
\begin{align*}
[U^{(p)}(T_M)-I]\varphi
={}&
\sum_{j=1}^M
\rme^{-\rmi\frac{t}{N}\tau_MH_{\sigma(M)}}
\cdots
\rme^{-\rmi\frac{t}{N}\tau_{j+1}H_{\sigma(j+1)}}
\nonumber\\
&\qquad
{}\times
\left[
\left(
\rme^{-\rmi\frac{t}{N}\tau_jH_{\sigma(j)}}
-I
+\rmi\frac{t}{N}\tau_jH_{\sigma(j)}
\right)
-
\rmi\frac{t}{N}
\left(
\rme^{-\rmi\frac{t}{N}\tau_jH_{\sigma(j)}}
-I
\right)
\sum_{k=1}^{j-1}
\tau_k
H_{\sigma(k)}
\right]
\varphi
\nonumber
\displaybreak[0]
\\
={}&
\sum_{j=1}^M
\rme^{-\rmi\frac{t}{N}\tau_MH_{\sigma(M)}}
\cdots
\rme^{-\rmi\frac{t}{N}\tau_{j+1}H_{\sigma(j+1)}}
\nonumber\\
&\qquad
{}\times
\left[
\left(
\rme^{-\rmi\frac{t}{N}\tau_jH_{\sigma(j)}}
-
I
+
\rmi\frac{t}{N}
\tau_j
H_{\sigma(j)}
\right)
+
\rmi\frac{t}{N}
\left(
\rme^{-\rmi\frac{t}{N}\tau_jH_{\sigma(j)}}
-
I
\right)
\sum_{k=1}^{j-1}
(-1)^{j-1-k}
\tau_k
H_{\sigma(j)}
\right]
\varphi,
\end{align*}
and the Trotter error is bounded by
\begingroup
\allowdisplaybreaks
\begin{equation*}
\xi_N^{(p)}(t;\varphi)
\le
N\sum_{j=1}^M
\left\|
\left[
\left(
\rme^{-\rmi\frac{t}{N}\tau_jH_{\sigma(j)}}
-
I
+
\rmi\frac{t}{N}
\tau_j
H_{\sigma(j)}
\right)
+
\rmi\frac{t}{N}
\left(
\rme^{-\rmi\frac{t}{N}\tau_jH_{\sigma(j)}}
-
I
\right)
\sum_{k=1}^{j-1}
(-1)^{j-1-k}
\tau_k
H_{\sigma(j)}
\right]
\varphi
\right\|.
\end{equation*}
\endgroup
This is valid for $\varphi\in\mathcal{D}(H_1)\cap\mathcal{D}(H_2)$, and reproduces Eq.~(\ref{eq:lemma_bound1}) of Lemma~\ref{lem:bound_refined} for the first-order Trotterization and Eq.~(\ref{eq:fractional_second_1}) of Lemma~\ref{lem:FractionalDelta} for the second-order Trotterization.
Using the bounds in Eqs.~(\ref{eq:bound_Delta0}) and~(\ref{eq:bound_Delta1}), we see that it scales as
$\xi_N^{(p)}(t;\varphi)=\mathcal{O}(N^{-\delta})$ with some $\delta\in(0,1)$, if $\varphi\in\mathcal{D}(H_1)\cap\mathcal{D}(H_2)$ but $\varphi\not\in\mathcal{D}(H_1^2)\cap\mathcal{D}(H_2^2)$, which is the case for the hydrogen eigenstates with $\ell=0,1$.
See Table~\ref{tab:fractional_scalings}\@.

For $q=2$, the error operator~(\ref{eq:remainder_q}) yields
\begingroup
\allowdisplaybreaks
\begin{align*}
[U^{(p)}(T_M)-I]\varphi
={}&
\sum_{j=1}^M
\rme^{-\rmi\frac{t}{N}\tau_MH_{\sigma(M)}}
\cdots
\rme^{-\rmi\frac{t}{N}\tau_{j+1}H_{\sigma(j+1)}}
\nonumber\\
&\qquad
{}\times
\Biggl[
\left(
\rme^{-\rmi\frac{t}{N}\tau_jH_{\sigma(j)}}
-I
+\rmi\frac{t}{N}\tau_jH_{\sigma(j)}
+\frac{t^2}{2N^2}\tau_j^2H_{\sigma(j)}^2
\right)
\nonumber\\
&\qquad\qquad
{}
-\rmi\frac{t}{N}
\left(
\rme^{-\rmi\frac{t}{N}\tau_jH_{\sigma(j)}}
-I
+\rmi\frac{t}{N}\tau_jH_{\sigma(j)}
\right)
\sum_{k=1}^{j-1}
\tau_kH_{\sigma(k)}
\nonumber\\
&\qquad\qquad
{}-\frac{t^2}{N^2}
\left(
\rme^{-\rmi\frac{t}{N}\tau_jH_{\sigma(j)}}
-I
\right)
\left(
\frac{1}{2}
\sum_{k=1}^{j-1}
\tau_k^2H_{\sigma(k)}^2
+
\sum_{k_1=2}^{j-1}
\sum_{k_2=1}^{k_1-1}
\tau_{k_1}H_{\sigma(k_1)}
\tau_{k_2}H_{\sigma(k_2)}
\right)
\Biggr]\,
\varphi
\nonumber\\
={}&
\sum_{j=1}^M
\rme^{-\rmi\frac{t}{N}\tau_MH_{\sigma(M)}}
\cdots
\rme^{-\rmi\frac{t}{N}\tau_{j+1}H_2}
\nonumber\\
&\qquad
{}\times
\Biggl[
\left(
\rme^{-\rmi\frac{t}{N}\tau_jH_{\sigma(j)}}
-I
+\rmi\frac{t}{N}\tau_jH_{\sigma(j)}
+\frac{t^2}{2N^2}\tau_j^2H_{\sigma(j)}^2
\right)
\nonumber\\
&\qquad\qquad
{}
+\rmi\frac{t}{N}
\,\Biggl(
\sum_{k=1}^{j-1}
(-1)^{j-1-k}\tau_k
\Biggr)
\left(
\rme^{-\rmi\frac{t}{N}\tau_jH_{\sigma(j)}}
-I
+\rmi\frac{t}{N}\tau_jH_{\sigma(j)}
\right)
H_{\sigma(j)}
\nonumber\\
&\qquad\qquad
{}-\frac{t^2}{N^2}
\,\Biggl(
\frac{1}{2}
\sum_{k_1=1}^{\lfloor\frac{j}{2}\rfloor}
\sum_{k_2=1}^{\lfloor\frac{j}{2}\rfloor}
\tau_{2k_1-1}
\tau_{2k_2-1}
-
\sum_{k_1=2}^{\lfloor\frac{j}{2}\rfloor}
\sum_{k_2=1}^{k_1-1}
\tau_{2k_1-1}
\tau_{2k_2}
\Biggr)
\left(
\rme^{-\rmi\frac{t}{N}\tau_jH_{\sigma(j)}}
-I
\right)
H_1^2
\nonumber\\
&\qquad\qquad
{}-\frac{t^2}{N^2}
\,\Biggl(
\frac{1}{2}
\sum_{k_1=1}^{\lceil\frac{j}{2}\rceil-1}
\sum_{k_2=1}^{\lceil\frac{j}{2}\rceil-1}
\tau_{2k_1}
\tau_{2k_2}
-
\sum_{k_1=1}^{\lceil\frac{j}{2}\rceil-1}
\sum_{k_2=1}^{k_1}
\tau_{2k_1}
\tau_{2k_2-1}
\Biggr)
\left(
\rme^{-\rmi\frac{t}{N}\tau_jH_{\sigma(j)}}
-I
\right)
H_2^2
\Biggr]\,
\varphi,
\end{align*}
and the Trotter error is bounded by
\begin{align*}
\xi_N^{(p)}(t;\varphi)
\le{}&
N\sum_{i=1}^{\lceil\frac{M}{2}\rceil}
\,\Biggl\|
\Biggl[
\left(
\rme^{-\rmi\frac{t}{N}\tau_{2i-1}H_1}
-I
+\rmi\frac{t}{N}\tau_{2i-1}H_1
+\frac{t^2}{2N^2}\tau_{2i-1}^2H_1^2
\right)
\nonumber\\
&\qquad\qquad
{}
+\rmi\frac{t}{N}
\,\Biggl(
\sum_{k=1}^{2i-2}
(-1)^k\tau_k
\Biggr)
\left(
\rme^{-\rmi\frac{t}{N}\tau_{2i-1}H_1}
-I
+\rmi\frac{t}{N}\tau_{2i-1}H_1
\right)
H_1
\nonumber\\
&\qquad\qquad
{}-\frac{t^2}{N^2}
\,\Biggl(
\frac{1}{2}
\sum_{k_1=1}^{i-1}
\sum_{k_2=1}^{i-1}
\tau_{2k_1-1}
\tau_{2k_2-1}
-
\sum_{k_1=2}^{i-1}
\sum_{k_2=1}^{k_1-1}
\tau_{2k_1-1}
\tau_{2k_2}
\Biggr)
\left(
\rme^{-\rmi\frac{t}{N}\tau_{2i-1}H_1}
-I
\right)
H_1^2
\Biggr]\,
\varphi
\Biggr\|
\nonumber\\
&{}+
N\sum_{i=1}^{\lfloor\frac{M}{2}\rfloor}
\,\Biggl\|
\Biggl[
\left(
\rme^{-\rmi\frac{t}{N}\tau_{2i}H_2}
-I
+\rmi\frac{t}{N}\tau_{2i}H_2
+\frac{t^2}{2N^2}\tau_{2i}^2H_2^2
\right)
\nonumber\\
&\qquad\qquad
{}
-\rmi\frac{t}{N}
\,\Biggl(
\sum_{k=1}^{2i-1}
(-1)^k\tau_k
\Biggr)
\left(
\rme^{-\rmi\frac{t}{N}\tau_{2i}H_2}
-I
+\rmi\frac{t}{N}\tau_{2i}H_2
\right)
H_2
\nonumber\\
&\qquad\qquad
{}-\frac{t^2}{N^2}
\,\Biggl(
\frac{1}{2}
\sum_{k_1=1}^{i-1}
\sum_{k_2=1}^{i-1}
\tau_{2k_1}
\tau_{2k_2}
-
\sum_{k_1=1}^{i-1}
\sum_{k_2=1}^{k_1}
\tau_{2k_1}
\tau_{2k_2-1}
\Biggr)
\left(
\rme^{-\rmi\frac{t}{N}\tau_{2i}H_2}
-I
\right)
H_2^2
\Biggr]\,
\varphi
\Biggr\|
\nonumber\\
&{}+
\frac{t^2}{N}
\sum_{i=1}^{\lceil\frac{M}{2}\rceil}
\,\Biggl(
\frac{1}{2}
\sum_{k_1=1}^{i-1}
\sum_{k_2=1}^{i-1}
\tau_{2k_1}
\tau_{2k_2}
-
\sum_{k_1=1}^{i-1}
\sum_{k_2=1}^{k_1}
\tau_{2k_1}
\tau_{2k_2-1}
\Biggr)
\left\|
\left(
\rme^{-\rmi\frac{t}{N}\tau_{2i-1}H_1}
-I
\right)
H_2^2
\varphi
\right\|
\nonumber\\
&{}+
\frac{t^2}{N}
\sum_{i=1}^{\lfloor\frac{M}{2}\rfloor}
\,\Biggl(
\frac{1}{2}
\sum_{k_1=1}^i
\sum_{k_2=1}^i
\tau_{2k_1-1}
\tau_{2k_2-1}
-
\sum_{k_1=2}^i
\sum_{k_2=1}^{k_1-1}
\tau_{2k_1-1}
\tau_{2k_2}
\Biggr)
\left\|
\left(
\rme^{-\rmi\frac{t}{N}\tau_{2i}H_2}
-I
\right)
H_1^2
\varphi
\right\|.
\end{align*}
\endgroup
This is valid for $\varphi\in\mathcal{D}(H_1^2)\cap\mathcal{D}(H_2^2)$, and reproduces Eq.~(\ref{eq:fractional_second_2}) of Lemma~\ref{lem:FractionalDelta} for the second-order Trotterization.
It scales as $\xi_N^{(p)}(t;\varphi)=\mathcal{O}(N^{-\delta})$ with some $\delta\in(1,2)$, if $\varphi\in\mathcal{D}(H_1^2)\cap\mathcal{D}(H_2^2)$ but $\varphi\not\in\mathcal{D}(H_1^3)\cap\mathcal{D}(H_1H_2^2)\cap\mathcal{D}(H_2H_1^2)\cap\mathcal{D}(H_2^3)$, which is the case for the hydrogen eigenstates with $\ell=2,3$.
See Table~\ref{tab:fractional_scalings} again.

The method presented here is iterative and can be extended to all $q\leq p$ for a $p$th-order product formula.
Thereby, error bounds with fractional scalings are obtained for any state $\varphi\not\in\mathcal{D}\bm{(}H^{(p)}(s)S_p(s)\bm{)}$\@.

\section{Mathematica script for computing higher-order Trotter bounds}
\label{sec:Mathematica}
In this appendix, we provide a script written in Wolfram Mathematica to compute the higher-order Trotter error.
As an input, one has to give the order $p$ of the product formula as well as the switching coefficients $\{\tau_j\}$.

\subsection{Compute $S_p(s)$ at time step $j$}
First, we compute $S_p(s)$ by Eq.~\eqref{eq:kth_action_explicit_expanded0}\@.
For this purpose, we use the string $a$ to denote the Hamiltonian $H_1$ and the string $b$ to denote the Hamiltonian $H_2$.
\begin{tcolorbox}[breakable,sharp corners,colframe=white,left=1mm]
\begin{Verbatim}[numbers=left]
ClearAll[T, KK, t, n, s];
	
(*
Definition of a power of the Hamiltonian.
*)
Ham[
   K_(*k:int*),
   power_(*i:int*)
   ]
   :=
   Module[{H},
   If[IntegerQ[K],
     H = If[EvenQ[K], "b", "a"];
     Return[StringRepeat[H, power]],
     (*Else*)
     Return[Hold[Ham[K, power]]];
     ](*If*);
   ];
  
(*
All combinations to distribute the integer p in lists of length L.
This function solves the non-commutative integer partition problem to distribute
the exponents in the j-th term of the integral action.
*)
Combinations[
   p_(*p:int*),
   len_(*L:int*)
   ] 
   := 
   Module[{AllSums},
   AllSums = IntegerPartitions[p, {len}];
   If[Length[AllSums] > 1,
    AllSums = MapAt[Permutations, AllSums, All];
    AllSums = Flatten[AllSums, 1],
    (*Else*)
    AllSums = Permutations @@ AllSums;
    ](*If*);
   Return[AllSums];
   ];
      
(*
Use the fact that the input state is an eigenvector of eigenvalue 0 to simplify.
*)
ZeroEigenstate[
   expressionVector_(*integral action:array*)
   ]
   := 
   Module[{result, bookKeeping},
      (*If the position is divisible by 2 but not by 4: move the entry to the left.*)
      (*The next entry after this: move it to the right.*)
      result = expressionVector;
      bookKeeping = False;(*helper variable*)
      Table[
         If[EvenQ[i] \[And] Mod[i, 4] != 0,
            result[[i - 1]] -= result[[i]];
            result[[i]] = 0;
            bookKeeping = True,
            If[bookKeeping,
               result[[i + 1]] -= result[[i]];
               result[[i]] = 0;
               bookKeeping = False;
            ](*If*);
         ](*If*);
      , {i, 1, Length[expressionVector]}](*Table*);
   Return[result];
];
      
(*
Compute the integral action at time step j.
The result is stored in a vector of length 2^p, where each entry corresponds to a
product of operators H_1="a" and H_2="b", sorted in alphabetical order.
*)
Sp[
   order_(*p:int*),
   mMax_(*M:int*),
   time_(*s:var*),
   tauList_(*List of the switching coefficients:array*),
   boundaryQ_ : False(*whether or not to compute the integral action
                        at the boundary:bool*),
   simplifyQ_ : False(*Whether or not to use the function `ZeroEigenstate` to
                        simplify the result:bool*)
   ]
   := 
      Module[{SummationList, tuples, delete, result, ii, L, AllSums, 
         times, Hamiltonians, jj, power, prefactor, HamTerm, position},
         tuples = Map[StringJoin, Tuples[Sort@{"a", "b"}, order], 2];
         delete = 0;(*deletion parameter q*)
         result = ConstantArray[0, 2^order];
         While[delete < order,
            L = order - delete;
            AllSums = Combinations[order, L];
            For[ii = 1, ii <= Length[AllSums], ii++,
               prefactor = 1/Times @@ Map[Factorial, AllSums[[ii]]];
               power = AllSums[[ii, 1]];
               times = T[KK[1]]^power;
               Hamiltonians = {Ham[KK[1], power]};
                  Table[
                     power = AllSums[[ii, jj]];
                     times *= T[KK[jj]]^power;
                     Hamiltonians = AppendTo[Hamiltonians, Ham[KK[jj], power]]
                     , {jj, 2, L}
                  ](*Table*);
                  Table[
                     HamTerm = Map[StringJoin, Evaluate @@@ Hamiltonians, All];
                     position = Position[tuples, HamTerm][[1]];
                     result[[position]] += prefactor*times
                     , Evaluate[
                        Sequence@@Prepend[
                           Table[{KK[index], 1, KK[index - 1] - 1}, {index, 2, L}]
                           ,{KK[1], 1, mMax}](*Prepend*)
                        ](*Evaluate*)
                  ](*Table*);
                  If[boundaryQ == False,
                     result = 
                     result /. 
                     T[mMax] -> (time - Sum[Abs[T[ll]], {ll, 1, mMax - 1}])*
                     Sign[tauList[[mMax]]];
                  ](*If*);
               ](*For ii*);
               delete++;
            ](*While*);
            result = result /. Table[T[j] -> tauList[[j]], {j, 1, mMax}];
            If[simplifyQ,
               Return[ZeroEigenstate[result]],
            (*Else*)
               Return[result];
            ](*If*);
         ];
\end{Verbatim}
\end{tcolorbox}

\subsection{Compute the error bound $\xi_N^{(p)}(t;\varphi)$}
Having defined $S_p(s)$ for time $s$ in the $j$th time slot, we can now compute the error bounds using Eq.~\eqref{eqn:error_estimate_p} of Lemma~\ref{lem:bound_pth_order}\@.
Here, we will use the convention that $(H_1+H_2)\varphi=0$.
\begin{tcolorbox}[breakable,sharp corners,colframe=white,left=1mm]
\begin{Verbatim}[numbers=left, firstnumber=129]
(*
Multiply the vector representing the integral action at time step j by the Hamiltonian
H_{\sigma(j)}.
*)
MultiplyBy[
   expressionVector_(*integral action:array*), 
   step_(*j:int*)
   ]
   :=
   Module[{zeros, result},
   zeros = ConstantArray[0, Length[expressionVector]];
   If[EvenQ[step],
    result = Join[zeros, expressionVector],
    result = Join[expressionVector, zeros];
    ](*If*);
   Return[result];
   ];
   
(*
Compute the actual error bound.
*)
ErrorBound[
   order_(*p:int*), mMax_(*M:int*), 
   tauList_(*List of the switching coefficients:array*), 
   simplifyQ_ : True(*Whether or not to use the function `ZeroEigenstate` to
                        simplify the result:bool*)
   ]
   := 
   Module[{Action, integralBounds, bound, operators},
      integralBounds = Table[Sum[Abs[tauList[[j]]], {j, 1, i}], {i, 1, Length[tauList]}];
      integralBounds = PrependTo[integralBounds, 0];
      bound = (t^(order + 1)/n^order)*Sum[
         Abs@
         MultiplyBy[
            Integrate[
               Sp[order, i, s, tauList, False, simplifyQ]
            , {s, integralBounds[[i]], integralBounds[[i + 1]]}](*Integrate*)
         , i](*MultiplyBy*)
      , {i, 1, mMax}](*Sum*);
      operators = 
      Map[StringJoin, 
         Tuples[Sort@{"a", "b"}, order + 1]
         // Map[Append["\[CurlyPhi]||"]] // Map[Prepend["||"]]
      , 2](*Map*);
      Return[Chop[bound, 10^-MachinePrecision].operators];
   ];
\end{Verbatim}
\end{tcolorbox}

\subsection{Examples}\label{sec:Mathematica_examples}
To show how the code works, we compute the examples of $p=1,2,4,6$.
For $p=1$, we enter
\begin{tcolorbox}[breakable,sharp corners,colframe=white,left=1mm]
\begin{Verbatim}[numbers=left, firstnumber=175]
ErrorBound[1, 2, {1, 1}, True]
\end{Verbatim}
\end{tcolorbox}
\noindent
and get the output
\begin{equation*}
	\frac{\text{$||$aa$\varphi ||$} t^2}{2 n}+\frac{\text{$||$bb$\varphi ||$} t^2}{2 n}.
\end{equation*}
For $p=2$, we enter
\begin{tcolorbox}[breakable,sharp corners,colframe=white,left=1mm]
\begin{Verbatim}[numbers=left, firstnumber=176]
ErrorBound[2, 3, {1/2, 1, 1/2}, True]
\end{Verbatim}
\end{tcolorbox}
\noindent
to obtain
\begin{equation*}
	\frac{\text{$||$aaa$\varphi ||$} t^3}{24 n^2}+\frac{\text{$||$baa$\varphi ||$} t^3}{8 n^2}+\frac{\text{$||$bbb$\varphi ||$} t^3}{12 n^2}.
\end{equation*}
The case $p=4$ can be computed with
\begin{tcolorbox}[breakable,sharp corners,colframe=white,left=1mm]
\begin{Verbatim}[numbers=left, firstnumber=177]
s4 = 1/(2 - 2^(1/3));
times4 = {s4/2, s4, (1-s4)/2, 1-2s4, (1-s4)/2, s4, s4/2}//N;
ErrorBound[4, 7, times4, True]
\end{Verbatim}
\end{tcolorbox}
\noindent
and results in
\begin{align*}
	&\frac{0.00427803 \text{$||$aaaaa$\varphi ||$} t^5}{n^4}
	+\frac{0.00950917 \text{$||$aabaa$\varphi ||$} t^5}{n^4}
	+\frac{0.00633945 \text{$||$aabbb$\varphi ||$}
   t^5}{n^4}
   +\frac{0.0243899 \text{$||$abaaa$\varphi ||$} t^5}{n^4}
   \\
   &+\frac{0.0731697 \text{$||$abbaa$\varphi ||$} t^5}{n^4}
   +\frac{0.0487798 \text{$||$abbbb$\varphi ||$}
   t^5}{n^4}
   +\frac{0.0299684 \text{$||$baaaa$\varphi ||$} t^5}{n^4}
   +\frac{0.092188 \text{$||$babaa$\varphi ||$} t^5}{n^4}
   \\
   &+\frac{0.0614587 \text{$||$babbb$\varphi ||$}
   t^5}{n^4}
   +\frac{0.212061 \text{$||$bbaaa$\varphi ||$} t^5}{n^4}
   +\frac{0.393394 \text{$||$bbbaa$\varphi ||$} t^5}{n^4}
   +\frac{0.22989 \text{$||$bbbbb$\varphi ||$}
   t^5}{n^4}.
\end{align*}
Finally, for $p=6$, we input
\begin{tcolorbox}[breakable,sharp corners,colframe=white,left=1mm]
\begin{Verbatim}[numbers=left, firstnumber=180]
s4 = 1/(2 - 2^(1/3));
s6 = 1/(2 - 2^(1/5));
times6 = {
	   s4*s6/2, s4*s6, (1-s4)s6/2, (1-2 s4)s6, (1-s4)s6/2, s4*s6,
	   s4*s6/2+s4*(1-2 s6)/2, s4 (1-2s6), (1-s4)(1-2s6)/2, (1-2s4)(1-2s6),
	   (1-s4)(1-2s6)/2, s4(1-2s6), s4(1-2s6)/2+s4*s6/2, s4*s6,
	   (1-s4)s6/2, (1-2s4)s6, (1-s4)s6/2, s4*s6, s4*s6/2
         }//N;
ErrorBound[6, 19, times6, True]
\end{Verbatim}
\end{tcolorbox}
\noindent
and obtain
\begingroup
\allowdisplaybreaks
\begin{align*}
	&\frac{0.000342245 \text{$||$aaaaaaa$\varphi ||$} t^7}{n^6}+\frac{0.00466073 \text{$||$aaaabaa$\varphi ||$} t^7}{n^6}+\frac{0.00310716
		\text{$||$aaaabbb$\varphi ||$} t^7}{n^6}+\frac{0.0128515 \text{$||$aaabaaa$\varphi ||$} t^7}{n^6}\\
	&+\frac{0.0385546 \text{$||$aaabbaa$\varphi ||$}
		t^7}{n^6}+\frac{0.0257031 \text{$||$aaabbbb$\varphi ||$} t^7}{n^6}+\frac{0.0128622 \text{$||$aabaaaa$\varphi ||$} t^7}{n^6}+\frac{0.0417926
		\text{$||$aababaa$\varphi ||$} t^7}{n^6}\\
	&+\frac{0.0278617 \text{$||$aababbb$\varphi ||$} t^7}{n^6}+\frac{0.0164871 \text{$||$aabbaaa$\varphi ||$}
		t^7}{n^6}+\frac{0.0382987 \text{$||$aabbbaa$\varphi ||$} t^7}{n^6}+\frac{0.0240441 \text{$||$aabbbbb$\varphi ||$} t^7}{n^6}\\
	&+\frac{0.00717357
		\text{$||$abaaaaa$\varphi ||$} t^7}{n^6}+\frac{0.0241779 \text{$||$abaabaa$\varphi ||$} t^7}{n^6}+\frac{0.0161186 \text{$||$abaabbb$\varphi ||$}
		t^7}{n^6}+\frac{0.0620135 \text{$||$ababaaa$\varphi ||$} t^7}{n^6}\\
	&+\frac{0.18604 \text{$||$ababbaa$\varphi ||$} t^7}{n^6}+\frac{0.124027
		\text{$||$ababbbb$\varphi ||$} t^7}{n^6}+\frac{0.0213565 \text{$||$abbaaaa$\varphi ||$} t^7}{n^6}+\frac{0.234396 \text{$||$abbabaa$\varphi ||$}
		t^7}{n^6}\\
	&+\frac{0.156264 \text{$||$abbabbb$\varphi ||$} t^7}{n^6}+\frac{0.168267 \text{$||$abbbaaa$\varphi ||$} t^7}{n^6}+\frac{0.205275
		\text{$||$abbbbaa$\varphi ||$} t^7}{n^6}+\frac{0.0969135 \text{$||$abbbbbb$\varphi ||$} t^7}{n^6}\\
	&+\frac{0.00379873 \text{$||$baaaaaa$\varphi ||$}
		t^7}{n^6}+\frac{0.0501259 \text{$||$baaabaa$\varphi ||$} t^7}{n^6}+\frac{0.0334173 \text{$||$baaabbb$\varphi ||$} t^7}{n^6}+\frac{0.137266
		\text{$||$baabaaa$\varphi ||$} t^7}{n^6}\\
	&+\frac{0.411798 \text{$||$baabbaa$\varphi ||$} t^7}{n^6}+\frac{0.274532 \text{$||$baabbbb$\varphi ||$}
		t^7}{n^6}+\frac{0.127458 \text{$||$babaaaa$\varphi ||$} t^7}{n^6}+\frac{0.420193 \text{$||$bababaa$\varphi ||$} t^7}{n^6}\\
	&+\frac{0.280128
		\text{$||$bababbb$\varphi ||$} t^7}{n^6}+\frac{0.290088 \text{$||$babbaaa$\varphi ||$} t^7}{n^6}+\frac{0.467689 \text{$||$babbbaa$\varphi ||$}
		t^7}{n^6}+\frac{0.258116 \text{$||$babbbbb$\varphi ||$} t^7}{n^6}\\
	&+\frac{0.0631099 \text{$||$bbaaaaa$\varphi ||$} t^7}{n^6}+\frac{0.141983
		\text{$||$bbaabaa$\varphi ||$} t^7}{n^6}+\frac{0.0946554 \text{$||$bbaabbb$\varphi ||$} t^7}{n^6}+\frac{0.36417 \text{$||$bbabaaa$\varphi ||$}
		t^7}{n^6}\\
	&+\frac{1.09251 \text{$||$bbabbaa$\varphi ||$} t^7}{n^6}+\frac{0.72834 \text{$||$bbabbbb$\varphi ||$} t^7}{n^6}+\frac{0.0509254
		\text{$||$bbbaaaa$\varphi ||$} t^7}{n^6}+\frac{1.05772 \text{$||$bbbabaa$\varphi ||$} t^7}{n^6}\\
	&+\frac{0.705143 \text{$||$bbbabbb$\varphi ||$}
		t^7}{n^6}+\frac{0.67482 \text{$||$bbbbaaa$\varphi ||$} t^7}{n^6}+\frac{0.679037 \text{$||$bbbbbaa$\varphi ||$} t^7}{n^6}+\frac{0.281524
		\text{$||$bbbbbbb$\varphi ||$} t^7}{n^6}.
\end{align*}
\endgroup
\end{widetext}

\bibliography{hydrogen.bib}

\begin{thebibliography}{100}
\bibitem{note3}
The bound could diverge for two different reasons: (i) one of $\varphi_\ell$
  does not satisfy the corresponding domain condition; (ii) $L=\infty$ and,
  even if the domain conditions are satisfied for all $\varphi_\ell$, the sum
  of the series diverges.
\end{thebibliography}
\end{document}